\PassOptionsToPackage{pagebackref,draft=false}{hyperref}
\PassOptionsToPackage{dvipsnames}{xcolor}
\documentclass[a4paper,onecolumn,superscriptaddress,10pt,accepted=2024-01-05, issue=3,volume=6, shorttitle=papers]{compositionalityarticle}
\pdfoutput=1
 
\usepackage[utf8]{inputenc}
\usepackage[english]{babel}
\usepackage[T1]{fontenc}

\usepackage{amsmath} % Enhanced math environments
\usepackage{amssymb} % Additional math symbols
\usepackage{amsthm} % Theorem environments
\usepackage{array} % Extended array and tabular environments
\usepackage{bm} % Bold math symbols
\usepackage[labelfont=bf]{caption} % Custom captions for floats
\usepackage{enumitem} % Customized lists
	\setenumerate{label=(\roman*), itemsep=0pt, parsep=2pt plus 4pt minus 2pt, partopsep=0pt, topsep=2pt plus 2pt}
	\setitemize{itemsep=0pt, parsep=2pt plus 4pt minus 2pt, partopsep=0pt, topsep=2pt plus 2pt}
\usepackage{etoolbox} % Toolbox of programming tools
\usepackage{graphicx} % Enhanced graphics support
	\graphicspath{ {./figures/} }
\usepackage{makecell} % Tabular column heads and multilined cells
\usepackage{mathtools} % Enhancements to amsmath package
\usepackage[verbose=true]{microtype} % Microtypography enhancements
\usepackage{nameref} % References by name instead of number
\usepackage[numbers,sort&compress]{natbib}
\usepackage{tablefootnote} % Footnotes in tables
\usepackage{thmtools} % For the triangle after example and remark
\usepackage{tikz} % Drawing package
\usepackage{tikz-cd} % Commutative diagrams with TikZ
\usepackage{tikzit} % String diagrams
\tikzstyle{morphism}=[fill=white, draw=black, shape=rectangle]
\tikzstyle{medium box}=[fill=white, draw=black, shape=rectangle, minimum width=0.8cm, minimum height=0.9cm]
\tikzstyle{large morphism}=[fill=white, draw=black, shape=rectangle, minimum width=1.7cm, minimum height=1cm]
\tikzstyle{bn}=[fill=black, draw=black, shape=circle, inner sep=1.5pt]
\tikzstyle{state}=[fill=white, draw=black, regular polygon, regular polygon sides=3, minimum width=0.8cm, shape border rotate=180, inner sep=0pt]
\tikzstyle{medium state}=[fill=white, draw=black, regular polygon, regular polygon sides=3, minimum width=1.3cm, inner sep=0pt, shape border rotate=180]
\tikzstyle{large state}=[fill=white, draw=black, regular polygon, regular polygon sides=3, minimum width=2.2cm, shape border rotate=180, inner sep=0pt]
\tikzstyle{wn}=[fill=white, draw=black, shape=circle, inner sep=1.5pt]
\tikzstyle{wide state}=[fill=white, draw=black, shape=isosceles triangle, minimum width=0.8cm, shape border rotate=270, inner sep=1.4pt, minimum height=0.5cm, isosceles triangle apex angle=80]
\tikzstyle{evalold}=[fill=white, draw=black, shape=isosceles triangle, minimum width=1.4cm, shape border rotate=90, inner sep=1.4pt, minimum height=0.4cm, isosceles triangle apex angle=110]
\tikzstyle{eval}=[fill=white, draw=black, shape=rectangle, minimum width=1.4cm, minimum height=0.55cm, inner sep=1.4pt, font={$\eval$}]

\tikzstyle{arrow}=[->]
\tikzstyle{dashed box}=[-, dashed]
\tikzstyle{mapsto}=[{|->}]
\tikzstyle{double wire}=[-, double]
\tikzstyle{protected}=[-, preaction={{ultra thick,white,draw}}]
\tikzstyle{ambient fill}=[-, draw=none, fill={rgb,255: red,245; green,220; blue,255}, tikzit draw={rgb,255: red,210; green,130; blue,255}] % Additional styles for TikZit diagrams
\usepackage[nottoc]{tocbibind} % Table of contents customization
\usepackage{xcolor} % Color management
\usepackage{hyperref}
\hypersetup{%
	linktocpage,
	hypertexnames=true
	}

\usepackage[noabbrev]{cleveref}

\newtheorem{theorem}{Theorem}[section]
\newtheorem{corollary}[theorem]{Corollary}
\newtheorem{lemma}[theorem]{Lemma}
\newtheorem{proposition}[theorem]{Proposition}
\newtheorem{definition}[theorem]{Definition}

\theoremstyle{definition}

\numberwithin{equation}{section}

\newcommand\xqed[1]{%
  \leavevmode\unskip\penalty9999 \hbox{}\nobreak\hfill
  \quad\hbox{#1}}
\newcommand{\demo}{\xqed{$\triangle$}}

 \declaretheoremstyle[
   spaceabove=\topsep, spacebelow=\topsep,
   postheadspace=.5em plus .5em minus .2em,
   numberlike=theorem,
   qed=\demo
 ]{mythmstyle}
 \declaretheorem[style=mythmstyle]{example}
 \declaretheorem[style=mythmstyle]{remark}

\newcommand{\N}{\mathbb{N}}

\newcommand{\R}{\mathbb{R}}

\definecolor{amber}{rgb}{1.0, 0.75, 0.0}

\newcommand\mtiny[1]{\mbox{\tiny\ensuremath{#1}}}
\newcommand{\parlabel}[1]{\textbf{#1.} }
\newcommand{\newemph}[1]{\emph{#1}}

\let\originalleft\left
\let\originalright\right
\renewcommand{\left}{\mathopen{}\mathclose\bgroup\originalleft}
\renewcommand{\right}{\aftergroup\egroup\originalright}

\newcommand{\moonrel}{\gtrdot}
\newcommand{\brel}{\mathord{ \moonrel_{\mtiny B}} } % Behavioral relation as an object
\newcommand{\brelop}{\mathrel{\brel}} % Behavioral relation as a binary operator
\newcommand{\mrel}{\mathord{ \moonrel }} % Behavioral relation for morphisms as an object
\newcommand{\mrelop}{\mathrel{\mrel}} % Behavioral relation for morphisms as a binary operator
\newcommand{\eval}{\mathrm{eval}} % Evaluation function
\newcommand{\cred}{\mathord{\rightharpoondown}} % Context-reduction relation for simulators, as an object
\newcommand{\credop}{\mathrel{\cred}} % Context-reduction relation for simulators as a binary operator
\newcommand{\oplaxcred}{\mathord{\rightharpoonup}} % Context-reduction relation for simulators, as an object
\newcommand{\oplaxcredop}{\mathrel{\oplaxcred}} % Context-reduction relation for simulators as a binary operator
\newcommand{\prog}{\mathbb{P}}		% The object of programs in a monoidal computer
\newcommand{\imrel}{\mathord{ \moonrel^{\mathrm{im}}_{\mtiny B}} }
\newcommand{\imrelop}{\mathrel{\imrel}}

\newcommand{\sngltn}{\mathord{\ast}}			% The element of a singleton set

\newcommand{\pwrset}{\mathcal{P}}
\newcommand{\brlsets}{\mathcal{B}}

\newcommand{\cop}{\mathrm{copy}}
\newcommand{\discard}{\mathrm{del}}
\newcommand{\id}{\mathrm{id}}

\DeclareMathOperator{\dom}{\mathrm{dom}}

\DeclareMathOperator{\im}{im}

\newcommand{\Beh}{\mathrm{Sha}}		% The `Shadow' functor

\newcommand{\behim}[1]{\overline{#1}}

\newcommand{\cat}[1]{\mathsf{#1}}
\newcommand{\cC}{\mathcal{A}}
\renewcommand{\det}{\mathrm{det}}	% deterministic morphisms
\newcommand{\cCdet}{\cC_\det}

\newcommand*{\coloniff}{\mathrel{\vcentcolon\Longleftrightarrow}}
\newcommand{\ph}{\mathord{\rule[-0.05em]{0.6em}{0.05em}}}		% Argument placeholder
\newcommand{\spen}{\xi}
\newcommand{\spsp}{\mathrm{spec}}
\newcommand{\intr}{\mathrm{int}}

\newcommand{\suppincop}{\mathsf{si}}	
\newcommand{\suppinc}[1]{% 						support inclusion
	\def\relstate{#1}%
	\ifx\relstate\empty
		\mathsf{si}%
	\else
		\suppincop_{#1}%
	\fi
}
\newcommand{\Suppop}{\mathsf{Supp}}
\newcommand{\Supp}[1]{%				 			support object (= the domain of the support inclusion)
	\def\relstate{#1}%
	\ifx\relstate\empty
		S%
	\else
		\Suppop_{#1}%
	\fi
}
\newcommand{\suppprojop}{\mathsf{sp}}
\newcommand{\suppproj}[1]{%						support projection in Dario's sense (not unique)
	\def\relstate{#1}%
	\ifx\relstate\empty
		\mathsf{sp}%
	\else
		\suppprojop_{#1}%
	\fi
}

	\let\abs\relax
	\let\norm\relax
	\DeclarePairedDelimiter{\abs}{\lvert}{\rvert}
	\DeclarePairedDelimiter{\norm}{\lVert}{\rVert}
	\DeclarePairedDelimiterXPP{\pnorm}[2]{}{\lVert}{\rVert}{_{#1}}{#2}
	\makeatletter
		\let\oldabs\abs
		\def\abs{\@ifstar{\oldabs}{\oldabs*}}
		\let\oldnorm\norm
		\def\norm{\@ifstar{\oldnorm}{\oldnorm*}}
		\let\oldpnorm\pnorm
		\def\pnorm{\@ifstar{\oldpnorm}{\oldpnorm*}}
	\makeatother

\providecommand{\given}{} % Just to make sure the \given command exists.
\newcommand{\SetSymbol}[1][]{%
    \nonscript\;\,#1\vert
    \allowbreak
    \nonscript\;\,
    \mathopen{}
}
\DeclarePairedDelimiterX{\Set}[1]{\{}{\}}{%
    \renewcommand{\given}{\SetSymbol[\delimsize]}
    #1
}
\makeatletter
    \let\oldSet\Set
    \def\Set{\@ifstar{\oldSet}{\oldSet*}}
\makeatother

\DeclarePairedDelimiterX{\Family}[1]{(}{)}{%
    \renewcommand{\given}{\SetSymbol[\delimsize]}
    #1
}
\makeatletter
    \let\oldFamily\Family
    \def\Family{\@ifstar{\oldFamily}{\oldFamily*}}
\makeatother

	\newcommand{\hyphenationsetting}{%
		\emergencystretch=0pt	% Emergency line extension
		\tolerance=2000		% Maximum \badness without hyphenation
		\pretolerance=1000		% Maximum \badness with hyphenation
		\righthyphenmin=4		% Minimum length of the left part of a broken/hyphenated word
		\lefthyphenmin=4		% Minimum length of the right part of a broken/hyphenated word
	}
	\hyphenationsetting

\begin{document}

\title{A Framework for Universality in Physics, Computer Science, and Beyond}
\date{}
\author{Tom\'a\v{s} Gonda}
\email{tomas.gonda@uibk.ac.at}
\orcid{0000-0002-1531-0058}
\affiliation{Institute for Theoretical Physics, University of Innsbruck, Austria}
\author{Tobias Reinhart}
\email{tobias.reinhart@uibk.ac.at}
\orcid{0000-0003-2993-5275}
\affiliation{Institute for Theoretical Physics, University of Innsbruck, Austria}
\author{Sebastian Stengele}
\email{sebastian.stengele@tum.de}
\orcid{0000-0002-1173-5067}
\affiliation{Department of Mathematics, Technical University of Munich, Germany}
\author{Gemma De les Coves}
\email{gemmadelescoves@gmail.com}
\orcid{0000-0002-0977-7727}
\affiliation{Institute for Theoretical Physics, University of Innsbruck, Austria}
\maketitle

%%%%%%%%%%%%%%%%%%%%%%%%%%%%%%%%%%%%%%%%%%%%%%%%%%%%%%%%%%%%%%%%%%%%%%
\newcommand{\authorsforheader}{Gonda, Reinhart, Stengele, and De les Coves}
\newcommand{\paperdoi}{https://doi.org/10.46298/compositionality-6-3}
\newcommand{\receiveddate}{2023-08-04}
\newcommand{\accepteddate}{2024-01-05}
%%%%%%%%%%%%%%%%%%%%%%%

	\begin{abstract}
		Turing machines and spin models share a notion of universality according to which some simulate all others. 
		Is there a theory of universality that captures this notion? 
		We set up a categorical framework for universality which includes as instances universal Turing machines, universal spin models, NP completeness, top of a preorder, denseness of a subset, and more. 
		By identifying necessary conditions for universality, we show that universal spin models cannot be finite.
		We also characterize when universality can be distinguished from a trivial one and use it to show that universal Turing machines are non-trivial in this sense. 
		Our framework allows not only to compare universalities within each instance, but also instances themselves.
		We leverage a Fixed Point Theorem inspired by a result of Lawvere to establish that universality and negation give rise to unreachability (such as uncomputability).  
		As such, this work sets the basis for a unified approach to universality and invites the study of further examples within the framework. 
	\end{abstract}

	\tableofcontents

\section{Introduction}
\label{sec:Introduction}

	Turing machines form a cornerstone of computer science with vast theoretical and technological consequences.
	They are a model of computation that formalizes information processing that can be carried out on manifold physical systems.  
	Equivalent models include $\lambda$-calculus, partial recursive functions, tag machines, Post machines, counter machines, and cellular automata \cite{Sv93,De09d,Wo02,Mo11}.  
	Naively, one may expect that to solve a specific information-theoretic \emph{problem}, it is necessary to build a tailor-made Turing machine (a \emph{particular solution}).
	However, the existence of a universal Turing machine (a \emph{universal solution}) invites a wholly new perspective: 
	One may invest in the development of a universal machine and use it to tackle any information-theoretic problem, i.e.\ to compute any function of interest. 
	This is possible because, by definition, a universal Turing machine can \emph{simulate} any other Turing machine as long as it is provided with a \emph{program} carrying the requisite instructions.
	As a consequence, instead of building machines, one would be writing programs. 
	Universality thus goes hand in hand with programmability and gives rise to a new kind of engineering{\,\textemdash\,}software engineering. 
	In this light, there is a pragmatic aspect to universality as it allows to trade resources of mechanical engineering for resources of software engineering.
	
	Spin systems are versatile toy models of complex systems \cite{Ho14b,So00,Th18}, which have also been shown to exhibit universality \cite{De16b,Re23}. 
	More concretely, a spin system (the \emph{particular solution}) is a composite system with local interactions and an energy function.
	Individual degrees of freedom are called spins and each interaction term contributes to the global energy additively by a real number that depends on the spin configuration. 
	The behavior of a spin system (the network-theoretic \emph{problem}) is thus a real-valued function of its spin configurations. 
	It was recently discovered that there are families of spin systems, such as the two-dimensional (2D) Ising model with fields, that are able to \emph{simulate} any other spin system. 
	In this sense, the 2D Ising model with fields is a universal spin model (the \emph{universal solution}).
	``Programming'' a given spin system amounts to fixing the values of certain spins and coupling strengths of a suitably chosen 2D Ising spin system. 
	
	In both situations, universal solution(s) can be programmed to solve any problem of interest, by simulating a particular solution to it. 
	This gives rise to the aforementioned shift in perspective{\,\textemdash\,}from building particular solutions to programming universal ones.
	Besides potential pragmatic significance, it can also be of conceptual importance.
	For example, in the case of computability, the shift allows the study of algorithms with respect to a universal Turing machine, which leads to notions such as algorithmic complexity, algorithmic randomness, and algorithmic probability \cite{Li2008}.
	One could imagine similar practical and conceptual developments for universality in spin systems, which was discovered much more recently.
	For instance, to physically implement the behavior of an arbitrary spin system, it is in principle sufficient to have the means to build 2D systems with tunable Ising interactions. 
	In addition, one could devise a complexity theory of spin systems based on the type of instructions for the universal spin model that they require.\footnotemark{}
	\footnotetext{In other works \cite{St21,Re21c}, spin models have been cast as formal languages and classified in the Chomsky hierarchy, which is an unrelated notion of complexity for spin systems.}
	The existence of universal spin models also invites the development of a novel theory of universality classes (in the sense of emergent properties of a given spin model around the critical point, see e.g.\ \cite{Ni11}). 
	Such a theory would be phrased in terms of the parameters of the universal spin model instead of properties of the given spin model, such as the lattice dimension, the number of spin levels, or its symmetries.  

	In light of these considerations, we ask: 
	\begin{quote}
		What is a meaningful structure of universality and is there a widely applicable theory behind it?
	\end{quote}
	To undertake this investigation, we are guided by the following two questions:
	\begin{enumerate}
		\item What do the examples of Turing machines and spin models have in common? 
		\item Can the theory apply to other scenarios where we expect to identify this structure?
	\end{enumerate}

	In this article, we present a formal attempt to study universality primarily inspired by the above two examples of universality. 
	To do so, we abstract features of the two examples to build a mathematical framework with universalities of Turing machines and spin models as special cases. 
	Its language can serve to develop some of the general theory of universality and to investigate other situations beyond the two above.
	
	Specifically, an instance of our framework consists of
	\begin{itemize}
		\item an ambient category $\cC$ specifying the systems of interest and relations between them, as well as the collection $T$ of all particular solutions, and
		\item a behavior structure specifying relevant problems as maps of type $C \to B$, which correspond to computable functions and behaviors of spin systems respectively in the two aforementioned examples.
	\end{itemize}
	A detailed discussion of motivations behind these features can be found in \cref{sec:motivation}.

	We formalize a set of solutions (that may or may not be universal) as a \emph{simulator} (\cref{def:simulator}), which describes a way in which ``programs'' can be used to solve relevant problems.
	That is, each program is compiled to produce a particular solution in $T$ and then evaluated to produce a map of type $C \to B${\,\textemdash\,}the problem it purports to solve. 
	\emph{Universality} (\cref{def:univ sim}) is a property of a simulator within an instance of the framework. 
	It arises if, for every target behavior that has a particular solution in $T$, we can find a corresponding program whose compiled behavior matches the target one. 
	
	Similar categorical frameworks, which however are developed to study computability and its compositional features, are those of monoidal computers\footnotemark{} \cite{pavlovic2018monoidal,pavlovic2023programs} and of Turing categories \cite{Co08d}. 
	\footnotetext{We spell out the connection to the framework of monoidal computers in \cref{sec:monoidal_computer}.}%
	On the other hand, general notions of universality have been explored from a conceptually related point of view in \cite{De09d,Wo02}. 
	
	\parlabel{Goals}
	The purpose of our framework is to offer a path to \emph{unification} of scenarios where universality plays a role.
	Its abstract language (supplemented with an interpretation) is meant to provide a dictionary facilitating the transfer of results between these scenarios.
	We illustrate this with an abstract expression of undecidability (or, more precisely, unreachability) within the framework itself. 
	Consequently, we give meaning to the notion of unreachability whenever the relevant assumptions are satisfied. 

	It is key that various ``universality scenarios'' can be \emph{instantiated as special cases} of the framework.
	Besides Turing machines and spin systems, we instantiate NP-completeness in computational complexity theory as well as several mathematical structures, such as dense subsets of topological spaces, maximal elements of ordered sets, and universal Borel sets.
	However, many examples of universality, which we expect to be fruitfully expressed in our framework, can for now only be instantiated in uninteresting ways{\,\textemdash\,}e.g.\ as being the top element of a preorder.
	We return to such instances in the \nameref{sec:Conclusions}.
	
	Another purpose is one of knowledge-organization.
	Apart from unification, we also gain a better \emph{understanding} of the individual universality scenarios.
	The interpretation attached to elements of the framework highlights the role that different parts of a given statement of universality play.
	Moreover, its abstract nature requires one to be transparent about the assumptions being made.
	In this sense, instantiating spin systems in the framework is not merely a restatement of the results of \cite{De16b}. 
	Rather, it is a complementary work that serves to situate universality of spin models in a broader context.
	
	Finally, the framework can provide \emph{new ideas and concepts}.
	Certainly, the notion of a simulator is a basic one.
	At a higher level, we establish a hierarchy of universal simulators that can be used to compare manifestations of universality within the same ambient category.
	For example, we show that a singleton simulator describing a universal solution (e.g.\ via a single universal Turing machine) is ``better'' than a so-called trivial simulator describing the collection of all particular solutions (e.g.\ via all Turing machines).

	\parlabel{Overview of this work}
	In the following section, we define the basic building blocks of our framework.
	Two versions thereof, one more abstract and the other more concrete appear in the crucial \cref{def:instance,def:intrinsic_beh}.
	They specify the relevant data needed to instantiate the framework.
	These include the notion ambient category from \cref{sec:ambient} and the notion of behavior structure from \cref{sec:behavior_structure}.
	The central actresses of each instance are simulators, defined in \cref{sec:simulators}. 
	
	\Cref{sec:Reductions} first defines universality of simulators in \cref{sec:universality_def}. 
	Besides introducing a singleton universal simulator associated to each universal Turing machine (\cref{ex:TM univ sim}), we also devote ample space to the details of universal spin models within our framework in \cref{sec:spinmodel}, and present other more mathematical examples in \cref{sec:math_examples}. 
	In \cref{sec:no-go_theorems}, we prove a no-go theorem for universality,  which can be used to obtain necessary conditions for universality of a given simulator.
	We exemplify this by showing that a universal spin model has to contain infinitely many spin systems (\cref{ex:nogo_spin}).
	
	 \Cref{sec:Morphisms} introduces a hierarchy of simulators via a simulator category (\cref{def:simulator_category}), which we use to compare universal simulators in \cref{sec:parsimony}. 
	We prove two results, \cref{thm:morph_stronger,thm:s 2 id}, which provide conditions under which a morphism between two simulators does and does not exist respectively.
	As mentioned above, we use them to establish a strict ordering between two universal simulators in the context of Turing machines (\cref{ex:trivial2universal,ex:universal2trivial}).
	
	In \cref{sec:Undecidability}, we study unreachability, i.e.\ the existence of problems (maps $C \to B$) with no particular solution in $T$.
	For Turing machines, this corresponds to uncomputability or undecidability. 
	By an abstract diagonal argument due to Lawvere, a specific kind of unreachability follows if there is a map $B \to B$ that has no fixed points (such as a negation on truth values).
	In particular, this result can be interpreted as the implication (cf.\ \cref{thm:Lawvere})
	\begin{center}
		negation $\implies$ unreachability (for universal solutions in $P \cong C$)
	\end{center}
	whenever there is an isomorphism between programs $P$ and contexts $C$.
	If $T$ is also isomorphic to $C$, then we obtain unreachability from negation. 
	To show this, we use \cref{thm:Lawvere}, which generalizes Lawvere's Fixed Point Theorem \cite{La69b}.
	As one of its applications, we show in \cref{ex:no_total_UTM} the known fact that there is no universal Turing machine that is also total.
	
	We also show the following connection to universality of simulators: 
	Given a universal simulator with a problem without a universal solution in $P$ (so-called unreachability for $P$), there is also no particular solution in $T$ for this problem.
	This can be interpreted as the implication (cf.\ \cref{prop:weak pt surj})
	\begin{center}
		unreachability (for $P$) + universality $\implies$ unreachability (for particular solutions in $T$),
	\end{center}
	so that combining the two results can be viewed as (cf.\ \cref{thm:fix point sim})
	\begin{center}
		negation + universality $\implies$ unreachability
	\end{center}
	and this holds even if $T$ and $C$ are a priori unrelated.

	In the somewhat technical \cref{sec:Functors}, we introduce ways to relate instances of the framework. 
	In particular, \cref{thm:simfun} shows that given a functor between the ambient categories of two instances, there is a corresponding relation of their respective simulator categories.
	We expect this notion to be valuable when contrasting instances of universality between disciplines.
	
	Finally, in \cref{sec:Conclusions}, we reflect on the current state of our framework and point towards future developments on the horizon. 
			
	\parlabel{On the use of category theory}
	At an intuitive level, much of the content can be understood without any knowledge of category theory. 
	To aid in this, we use the graphical language of string diagrams, which are to be read from bottom to top.
	A process like
	\begin{equation}
		\input{box.tikz}
	\end{equation}
	can be thought of as a (partial, potentially multi-valued) \emph{function} $f$ with input $A$ and \mbox{output $X$}.
	See \cref{sec:category_theory,sec:additional_def} for more details and additional definitions. 
	Each wire in a string diagram can be viewed as a \emph{set} (or, equivalently, the identity function).
	Using this intuitive interpretation of string diagrams, one can grasp the main concepts of the paper without any background on category theory.
	The main exception to this are
	\begin{itemize}
		\item \cref{sec:behavior_structure}, where we need category-theoretic concepts to spell out  the details of how to translate from the abstract perspective to the more intuitive set-like interpretation, and
		\item \cref{sec:Functors}, where we use functors as a way to compare instances of our framework.
	\end{itemize}
	
	\parlabel{Acknowledgements}
	TG, SS and TR acknowledge funding from the Austrian Science Fund (FWF) via the START Prize Y1261-N.
	SS acknowledges funding from the Deutsche Forschungsgemeinschaft (DFG, German Research Foundation) CRC TRR 352.

\section{The Set-Up}
\label{sec:Set-Up}
	
	In order to ground our discussion of universality in \cref{sec:Reductions}, we first spell out the background assumptions.
	We start by giving an informal introduction of our approach to the concept of universality (\cref{sec:motivation}).
	This includes motivations for the upcoming definitions that are based on the examples we intend to capture with our framework{\,\textemdash\,}universal Turing machines and universal spin models, in particular.
	The subsequent \cref{sec:ambient} defines the so-called target--context category. 
	Specifying its data is precisely the task of laying out the relevant background assumptions, which is why we refer to its role in the framework as the `ambient category'.
	Distinction in the phenomenology of universality (in the sense of our framework) for Turing machines (TMs) and for spin models arises \emph{because} they are formalized using distinct ambient categories.
	
	Note that `to be a universal TM' is not a property of the set of TMs, but of a particular TM.
	Similarly, universality in our framework cannot be directly attached to the ambient category.
	Rather, it is a property of simulators, whose definition is given in \cref{sec:simulators}. 
	
	In \cref{sec:behavior_structure}, we use relations to model morphisms of the ambient category and provide a more concrete account thereof.
	\Cref{sec:intrinsic_beh} discusses a special case of this construction, which is then used in \cref{sec:Undecidability} to connect universality with undecidability. 
	
	\begin{remark}[Levels of abstraction]\label{rem:levels_abstraction}
		We provide three variants on the structure that an ambient category may carry, each being a special case of the previous one.
		\begin{enumerate}
			\item \label{rem:levels_abstraction:i} Target--context category (\cref{fig:target--context-category}), described in \cref{sec:ambient}, is the most minimal and most abstract.
			\item \label{rem:levels_abstraction:ii} Target--context category with behaviors (\cref{fig:behavior-structure}), described in \cref{sec:behavior_structure}, is much more complicated, but it also comes with a concrete relational interpretation.
			\item \label{rem:levels_abstraction:iii} Target--context category with intrinsic behaviors (\cref{fig:intrinsic-behavior}), described in \cref{sec:intrinsic_beh}, makes only one additional assumption. 
		\end{enumerate}
		The reason for introducing these distinctions is that, for many purposes, reasoning with an abstract target--context category is much simpler. 
		However, some constructions and results, such as the simulator category (\cref{def:simulator_category}) or the existence of fixed points (\cref{thm:fix point sim}), require the additional assumptions of target--context categories with (intrinsic) behaviors.
	\end{remark}

	\begin{table}[t!]\centering
		\begin{tabular}{c|c|c|c|c} 
%			\hline
			\thead{Name of the \\ preorder}
			& \thead{Symbol}
			& \thead{Underlying set}
			& \thead{Data it \\ depends on}
			& \thead{Level of \\ abstraction} % \\ (\cref{rem:levels_abstraction})}
		\\ \hline 
			\hyperref[def:instance]{Ambient relation}
			& $\mrel$ 
			& \makecell{morphisms in \\ $\cC(A,T\otimes C)$} %$f \mrelop g$ for $f,g\in \cC(A,T\otimes C)$ 
			& \makecell{Ambient \\ category $\cC$}
			& \makecell{Any \hyperref[rem:levels_abstraction:i]{target--context} \\ \hyperref[rem:levels_abstraction:i]{category (TCC)}}
		\\ %\hline 
			\makecell{\hyperref[def:beh_structure]{Behavioral relation}}
			& $\brel$ 
			& \makecell{behaviors in $\behim{B}$} %$b_1 \brelop b_2$ for $b_1,b_2 \in \behim{B}$
			& \makecell{Set $\behim{B}$ of \\ behaviors} % \\ (\cref{def:beh_structure})}
			& \makecell{Any \hyperref[rem:levels_abstraction:ii]{TCC with} \\ \hyperref[rem:levels_abstraction:ii]{behaviors}}
		\\ %\hline 
			\hyperref[def:brel]{Imitation relation}
			&  $\imrel$ 
			&  \makecell{relations in \\ $\cat{Rel}(A,\behim{B})$} %$\mu \imrelop \nu$ for $\mu,\nu\in \cat{Rel}(A,\behim{B})$
			&  \makecell{Behavioral \\ relation $\brel$}
			&  \makecell{Any \hyperref[rem:levels_abstraction:ii]{TCC with} \\ \hyperref[rem:levels_abstraction:ii]{behaviors}}
		\\ %\hline 
			\makecell{\hyperref[def:mrel]{Ambient} \\ \hyperref[def:mrel]{imitation relation}}
			& $\mrel$
			& \makecell{morphisms in \\ $\cC(A,T\otimes C)$} %$f \mrelop g$ for $f,g\in \cC(A,T\otimes C)$ 
			& \makecell{\hyperref[def:beh_structure]{Behavior} \\ \hyperref[def:beh_structure]{structure}} %(\cref{fig:behavior-relation})
			&  \makecell{Any \hyperref[rem:levels_abstraction:ii]{TCC with} \\ \hyperref[rem:levels_abstraction:ii]{behaviors}}
		\\ %\hline 	
			\makecell{\hyperref[def:c_red_rel]{(Op)lax context-} \\ \hyperref[def:c_red_rel]{-reduction relation}}
			& \makecell{lax: $\cred$ \\ oplax: $\oplaxcred$}
			& \makecell{morphisms in \\ $\cC(A,T)$} %$f \mrelop g$ for $f,g\in \cC(A,T\otimes C)$ 
			& \makecell{Ambient \\ category $\cC$}
			& \makecell{Any \hyperref[rem:levels_abstraction:i]{TCC}}
		\\ %\hline 	
		\end{tabular}
		\caption{Summary of preorders appearing in our framework. 
			The third column gives the set whose elements can be compared with the respective relation.
			The fourth column tells us whether one needs to specify any additional data in order to construct the preorder relation in question.
			The last column indicates what is the most general version of our framework (\cref{rem:levels_abstraction}) in which the preorder can be used.}
		\label{tab:summary_relations}
	\end{table}
	
	\subsection{Motivation For the Set-Up}\label{sec:motivation}
	
		\parlabel{Basic idea (targets $T$)}
		In the broadest sense, we take the universality of a certain object $u$ to mean that $u$ is \emph{all-encompassing}. 
		In particular, universal is a notion relative to a set $T$ of \emph{all} particular solutions, and an \emph{encompassing} order \mbox{relation $\cred$} that specifies how the universal solution $u$ may be used to recover each particular solution \mbox{in $T$}. 
		We refer to elements of $T$ as targets.
		
		When studying computability, $T$ may stand for the set of all Turing machines, while $t_1 \credop t_2$ denotes that Turing machine $t_1$ simulates Turing machine $t_2$.
		A Turing machine $u \in T$ is universal if it is all-encompassing in the sense that for all $t \in T$, it satisfies $u \credop t$.
		In other words, such a $u$ is the top element of the simulation preorder (cf.\ \cref{ex:cofinal,rem:encompass}).
		The set $T$ together with the relation $\cred$ form the most basic set-up for statements of universality.
		
		\parlabel{Additional structure (compiler $\bm{s_T}$, context reduction $\bm{s_C}$, and evaluation $\bm{\eval}$)}
		Such a simple structure is, however, not sufficient for the analysis we strive for. 
		Consider the example of universal spin models \cite{De16b} (cf.\ \cref{sec:spinmodel}).
		While the targets of interest are spin systems, a universal spin model is not itself another spin system, but a \emph{collection} of spin systems. 
		To specify it, we may consider another set $P$ of parameters and a map
		\begin{equation}
			s_T \colon P \to T
		\end{equation}
		called the \emph{compiler}, whose image comprises all the spin systems in a spin model. 
		Only upon fixing the relevant parameters to a $p \in P$, a spin system $s_T(p)$ is specified.
		
		For Turing machines, on the other hand, a universal compiler can be a constant function whose image is a universal Turing machine $u$.
		To say that $u$ can simulate any other Turing machine $t$, there should be a description $r(t)$ of $t$, which we can provide to $u$, so that $u$ `acts as' $t$ on the rest of its input.
		Namely, we need to turn the pair of $r(t)$ together with an arbitrary input of $t$ into a suitable input of $u$.
		Denoting the set of input strings by $C$,\footnotemark{} this is a function of type
		\begin{equation}
			s_C \colon P \times C \to C,
		\end{equation}
		called the \emph{context reduction}.
		We interpret $P$ as a set of ``programs''.
		\footnotetext{Later, $C$ is interpreted more generally to describe \emph{contexts} in which targets in $T$ manifest their behavior.}%
		For instance, we can choose 
		\begin{equation}\label{eq:pairing}
			s_C \bigl( r(t),c \bigr) \coloneqq \langle r(t), c \rangle 
		\end{equation}
		where $c$ is an arbitrary input string and $\langle \ph , \ph \rangle$ is a pairing function\footnotemark{} that allows $u$ to distinguish $r(t)$ from $c$. 
		\footnotetext{By a pairing function $C \times C \to C$ we mean one with a computable left inverse that allows the recovery of the original pair of strings it encodes.}%
		We think of two Turing machines as equivalent if they implement the same partial function from input to outputs.
		Suppose we are given a function 
		\begin{equation}
			\eval \colon T \times C \to B
		\end{equation}
		that describes the result of evaluating a Turing machine in $T$ on a given input in $C$ to produce an output in a set $B$.\footnotemark{}
		\footnotetext{Later, $B$ is interpreted more generally as a set of possible \emph{behaviors} of some target $t \in T$ in the context $c \in C$.}%
		Then two Turing machines $t_1, t_2 \in T$ behave equivalently if, for any input, they produce the same output, i.e.\ if we have
		\begin{equation}\label{eq:equivalence_of_targets}
			\forall c \in C \quad \eval (t_1, c) = \eval (t_2, c).
	 	\end{equation}
	 	This condition makes sense only if the Turing machines halt $t_1$ and $t_2$ halt on the input $c$ and thus produce some output in finite time.
	 	If one demands that equivalent Turing machines need to have the same domain on which they halt, we can add this information in the $B$, but this makes the evaluation function uncomputable (see \cref{ex:TM_behaviors,ex:TM intrinsic}).
		
		\parlabel{Simulators and their universality}
		Putting these elements together, we can express universality in a more nuanced fashion.
		Universality is a property of the pair of a compiler $s_T$ and a context reduction $s_C$, which we combine into a single function $s \colon {P \times C \to T \times C}$ called a \emph{simulator} via
		\begin{equation}\label{eq:simulator_def}
			s(p,c) \coloneqq \bigl( s_T(p) , s_C(p,c) \bigr).
		\end{equation}	
		A simulator is termed \emph{universal} if there exists a map $r \colon T\to P$ selecting, for any target $t \in T$, a program $r(t)$ such that the behavior of $s(r(t),c) \in T \times C$ coincides with that of $(t,c)$ for every context $c\in C$.
		This means that, for all $c$ and all $t$, we have
		\begin{equation}\label{eq:univ_concrete}
			\eval \Bigl( s_T \bigl( r(t) \bigr) , s_C \bigl(r(t), c \bigr) \Bigr) = \eval (t, c).
		\end{equation}
		We call such an $r$ a \emph{reduction}.
	
		Whenever the image of $s_T$ is a single element $u$ of $T$, $s$ is a \emph{singleton} simulator, so that we can speak of $u$ itself as a \emph{universal target} in $T$.
		This is the case for Turing machines{\,\textemdash\,}there exists a universal Turing machine $u$.
		
		Note that we demand that the output of the compiler does not depend on the context, as expressed by \cref{eq:simulator_def}. 
		For instance, to simulate the behavior of a desired spin system $t$, we may choose a suitable spin system $s_T(r(t))$ from the universal spin model, but this choice should not be adapted in response to the context $c$ (which is, in this case, a spin configuration \mbox{of $t$}). 
	 	
		\parlabel{Behavioral relation $\brel$}
		While \cref{eq:equivalence_of_targets} is an acceptable notion of equivalence between Turing machines with identical domains, its abstract version does not cover all situations of interest to us.
		For the purposes of spin system simulations (\cref{sec:spinmodel}), behaviors include information about a specified energy threshold $\Delta$.
		When comparing the target spin system to be simulated with another one that is meant to achieve the simulation, we are generally only interested in the fact that their possible energy levels agree below the threshold of the target system.
		This is easier to model with an asymmetric relation $\brel$ among the elements of the set $B$ of behaviors.
		In a generic instance of our framework, we thus weaken \cref{eq:univ_concrete} to only hold up to this preorder relation $\brel$.
		
		We interpret relation 
		\begin{equation}\label{eq:beh_relation_of_targets}
			\eval (t_1, c_1) \brelop \eval (t_2, c_2)
		\end{equation}
		as saying that the left behavior (of $t_1$ in context $c_1$) is at least as informative as the one on the right-hand side.
		By informativeness, we refer to the information contained in the behavior that is relevant for the user of the framework.
		
		\parlabel{Ambient category}
		In this article we take an abstract approach and formulate universality purely in terms of simulators and their properties. 
		Among the relevant properties is whether $s$ is a singleton simulator as well as the complexity of implementing simulations described by a given simulator in practice.
		The latter amounts to a restriction on the allowed functions for the compiler $s_T$, the context reduction $s_C$, the reduction $r$, or other elements of the framework. 
		We can impose this restriction by specifying an \emph{ambient category} $\cC$ (more details in \cref{sec:ambient}) as part of the definition of an instance of our framework.
		Its morphisms are processes that satisfy any of the desired restrictions we wish to impose.
		
		\parlabel{Relational semantics}
		In \cref{rem:multi-valued_context_reduction}, we argue that to meaningfully capture universality of spin models, we need to consider context reductions $s_C$ that are multi-valued. 
		For this reason, our framework uses relations (i.e.\ partial, multi-valued functions) to give concrete interpretation to the abstract morphisms in the ambient category.
		Such relational semantics is provided by a so-called \emph{shadow functor} assigning a relation to each process in the ambient category.
		One of the main challenges that arises from using multi-valued functions is that a relation $\brel$ among behaviors does not extend to a unique relation among \emph{sets of behaviors}.
		Our choice of doing so{\,\textemdash\,}the \emph{imitation relation} on the power set $\pwrset(B)$ is chosen so that the interpretation of a universal spin model matches that of our notion of a universality of the relevant spin system simulators.
		See \cref{sec:spinmodel} for details of spin system simulators and \cref{rem:imitation_meaning} for a detailed interpretation of the imitation relation.
	
	\subsection{Ambient Category}\label{sec:ambient}
	
		An ambient category describes the basic building blocks of conceivable simulators and eventually also of ways to map between them.
		We assume that it is a symmetric monoidal category $\cC$ (see \cite[definition 8]{Coecke2009}), and therefore comes equipped with a tensor product $\otimes$ and a unit object $I$.
		%We omit the formal definition of a (strict) symmetric monoidal category.
		%For more details, see \cite[definition 8]{Coecke2009} for instance.
		Intuitively, a symmetric monoidal category has, on top of sequential composition $\circ$, also parallel composition $\otimes$ of morphisms (and thus also of objects), depicted as
		\begin{equation}
			\input{tensor.tikz}
		\end{equation}
		which is such that string diagrams such as 
		\begin{equation}
			\input{process.tikz}
		\end{equation}
		correspond to morphisms in the category. 
		The meaning of these string diagrams is invariant under planar rearrangement that preserves connectivity, including the crossing of wires implemented by distinguished swapping morphisms
		\begin{equation}
			\input{swap.tikz}
		\end{equation}
		for all objects $Y$ and $Z$.
		In addition, it is equipped with a unit object $I$ that satisfies $I \otimes A = A = A \otimes I$ and for this reason $I$ is ommited from the diagrams.
	
		We further postulate that for every object $A \in \cC$, there are two morphisms, copying ${\cop_A \colon A \to A \otimes A}$ and deletion $\discard_A \colon A \to I$, with the following string diagramatic representation
		\begin{equation}\label{eq:copy_delete}
			\input{copy_delete.tikz}
		\end{equation}
		that satisfy 
		\begin{equation*}\label{eq:comonoid}
			\input{comonoid.tikz}
		\end{equation*}
		and thus make each $A$ into a commutative comonoid. 
		Additionally, copying and deletion morphisms are compatible with tensor products, i.e.\ 
		\begin{equation}\label{eq:copy_product}
			\input{copy_product.tikz}
		\end{equation}
		hold for all $A,X \in \cC$.
		Finally, we also require $\discard_I = \id_I$. 
		
		Categories satisfying the above properties have been called CD (as an abbreviation of copy-delete) categories \cite{cho2017disintegration,piedeleu2023introduction}, 
		as well as \textbf{gs-monoidal} \cite{gadducci1996algebraic,fritz2022free,fritz2022lax}, 
		where ``gs'' is an abbreviation of garbage-sharing, as $\discard_A$ can be interpreted as a morphism that discards $A$, and $\cop_A$ as sharing the information carried by $A$ to two parties.
		The pair of $\cop_A$ and $\discard_A$ for each $A$ is also referred to as a data service in \cite{pavlovic2018monoidal}. 
%		We use the terminology of  gs-monoidal categories here.
		A stronger version of gs-monoidal categories, carrying morphisms dual to $\cop$ and $\discard$, has been studied as cartesian bicategories \cite{carboni1987cartesian,carboni2008cartesian} and used as generalized categories of relations.
		
		Similarly, an important example of gs-monoidal categories in our framework is $\cat{Rel}$, the category of sets and relations. 
		
		\begin{example}[Category of relations $\cat{Rel}$]\label{ex:Rel}
			Objects of $\cat{Rel}$ are sets and its morphisms of type ${A \to X}$ are relations, i.e.\ subsets of the product $A \times X$.
			Sequential composition is given, for two relations $f \colon A \to X$ and $g \colon X \to Y$, by 
			\begin{equation}
				(a,y) \in g \circ f  \quad \coloniff \quad  \exists \, x \in X \; : \;  (a,x) \in f \text{ and } (x,y) \in g.
			\end{equation}
			We also write $f(a)$ for the set of all $x$ in $X$ such that $(a,x) \in f$ holds.
			
			We use the cartesian product of sets as the tensor product on $\cat{Rel}$ to make it into a symmetric monoidal category whose unit object $I$ is the singleton set $\{ \sngltn \}$.
			$\cat{Rel}$ is also a gs-monoidal category.
			In particular, each copy morphism is defined as
			\begin{equation}
				(a_1, a_2) \in \cop_A(a)  \quad \coloniff \quad  a = a_1 = a_2
			\end{equation}	
			and a discarding morphism is the full relation $\discard_A = A \times I$.
			In other words,
			\begin{equation}
				\discard_A(a) = \{\sngltn\}
			\end{equation}
			holds for all $a \in A$.
		\end{example}	
		For any morphism $f \colon A \to X$ in $\cat{Rel}$, its domain $\dom(f)$ is the set of elements of $A$ for which $f(a)$ is non-empty.
		We can think of $\dom(f)$ as a relation of type $A \to A$ itself{\,\textemdash\,}one that is given by the restriction of $\id_A$ to the elements where $f$ is defined:
		\begin{equation}\label{eq:domain_rel}
			\dom(f) \, (a) = \begin{cases} \{a\} & \text{if } f(a) \neq \emptyset \\ \emptyset & \text{if } f(a) = \emptyset \end{cases}
		\end{equation}
		This is a special case of the following concept of a domain that can be defined in any gs-monoidal category.
		This is relevant for us because, e.g.\ in the context of computable functions ($\cat{Tur}^\intr$ from \cref{ex:TM}), we can think of the points in the domain as inputs for which the corresponding Turing machine halts.
		\begin{definition}[\cite{fritz2022lax}]\label{def:domain}
			For any morphism $f \colon A \to X$ in a gs-monoidal category, its \textbf{domain} $\dom(f) \colon A \to A$ is defined as
			\begin{equation}\label{eq:domain}
				\input{domain.tikz}
			\end{equation}
		\end{definition}
		
		It is clear from \eqref{eq:domain_rel} that restricting a relation $f$ to its domain yields the same relation, i.e.\ we have
		\begin{equation}\label{eq:normalized_rel}
			f \circ \dom(f) = f
		\end{equation}
		for every morphism $f$ in $\cat{Rel}$.
		This is not true in general{\,\textemdash\,}it fails for un-normalized probability distributions for instance (see \cref{rem:unnormalized}).
		The next definition follows the terminology of \cite{di2023evidential}.
		
		\begin{definition}\label{def:normalized}
			We say that a morphism $f$ in a gs-monoidal category is \textbf{quasi-total} if it satisfies
			\begin{equation}\label{eq:normalized}
				\input{normalized.tikz}
			\end{equation}
			A gs-monoidal category is termed \textbf{quasi-total} itself if all of its morphisms are quasi-total.
		\end{definition}

		More generally, asking whether the restriction of a morphism $f$ to the domain of another morphism $g$ coincides with $g$ defines a preorder relation among the quasi-total morphisms.
		Taking inspiration from works on restriction categories \cite[section 2.1.4]{cockett2002restriction}, we express it as follows.
		\begin{definition}\label{def:dom_eq}
			Consider two morphisms $f,g \colon A \to X$ in a gs-monoidal category.
			We say that \textbf{$\bm{f}$ agrees with $\bm{g}$ on the domain of $\bm{g}$}, denoted $f \sqsupseteq g$, if we have
			\begin{equation}\label{eq:restriction_order}
				\input{eq_on_domain.tikz}
			\end{equation}
		\end{definition}
		\Cref{eq:restriction_order} is a string diagrammatic expression of 
		\begin{equation}\label{eq:restriction_order_2}
			  f \circ \dom(g) = g.
		\end{equation}
		
		While quasi-total morphisms do not constitute a subcategory of a generic gs-monoidal category, there are important subcategories, all of whose morphisms are quasi-total.
		We now introduce two such subcategories{\,\textemdash\,}those of functional and of total morphisms respectively.
		\begin{definition}\label{def:functional}
			If a morphism $f \colon A \to X$ in a gs-monoidal category satisfies
			\begin{equation}\label{eq:deterministic}
				\input{deterministic.tikz}
			\end{equation}
			then we say that $f$ is \textbf{functional}.
		\end{definition}
		For instance, in $\cat{Rel}$ from \cref{ex:Rel}, functional morphisms are precisely the partial functions.
		Specifically, if we think of relations as partial, multi-valued functions, \cref{eq:deterministic} in $\cat{Rel}$ amounts to
		\begin{equation}
			\Set*[\big]{ (x,x') \in X \times X  \given x, x' \in f(a) } = \Set*[\big]{ (x,x) \in X \times X \given x \in f(a) }
		\end{equation}
		for every $a \in A$.
		This is equivalent to requiring $f$ to be (at most) single-valued.
		\begin{remark}[Functional $\neq$ single-valued]\label{rem:unnormalized}
			In categories other than $\cat{Rel}$, the two notions of being functional as in \cref{eq:deterministic} and of being single-valued need not coincide.
			For example, in the category of finite sets (as objects) and real matrices (as morphisms), consider the vector $p \coloneqq \left( 1/2, 0 \right)^T$.
			While it is single-valued, $p$ does not satisfy \cref{eq:deterministic}.
			This is precisely because $p \circ \dom(p) = \left(1/4, 0 \right)^T$ is not equal to $p$.
			Indeed, $p$ would satisfy \cref{eq:deterministic} if its right-hand side was replaced by the same morphism precomposed with $\dom(f)$.
			Under the assumption that $f \circ \dom(f) = f$ holds for all morphisms in the ambient category, the two ways to define functional morphisms, with and without $\dom(f)$ on the right-hand side of \cref{eq:deterministic}, are equivalent.
		\end{remark}
		
		\begin{definition}\label{def:total}
			If we have
			\begin{equation}\label{eq:normalization}
				\input{normalization.tikz}
			\end{equation}
			then $f$ is said to be a \textbf{total} morphism.
		\end{definition}
			In $\cat{Rel}$, total morphisms are total relations.
			Indeed, \cref{eq:normalization} in this category says 
			\begin{equation}
				\forall \, a \in A \; : \;  f(a) \neq \emptyset,
			\end{equation}
			because the right-hand side, $\discard_A$, is a relation that assigns, to each $a$, the unique non-empty subset of the singleton set $I$.	
		\begin{definition}\label{def:deterministic}
			A morphism that is both functional and total is termed \textbf{deterministic}.
			We denote the subcategory of $\cC$ consisting of deterministic morphisms  by $\cCdet$.
		\end{definition}
		It follows from the definitions that for any gs-monoidal $\cC$, the deterministic subcategory $\cCdet$ is also gs-monoidal.
		Since a relation is deterministic if and only if it is a function, we have that $\cat{Rel}_\det$ coincides with $\cat{Set}${\,\textemdash\,}the category of sets and functions.
		
		A morphism with trivial input, i.e.\ of type $I \to A$ for some object $A$, is called a \textbf{state}. 
		In $\cat{Rel}$, states correspond to subsets of the set $A$, while deterministic states are elements of $A$. 
		
		As mentioned in \cref{sec:motivation}, we interpret the hom-sets of $\cC$ (collections of morphisms of a particular type) as specifying the admissible processes in the situation of interest.
		For instance, when studying the universality of Turing machines, they may be the computable functions, while for completeness in computational complexity they may correspond to certain reductions. 
		
		Importantly, within $\cC$ we also identify an object $T$ of \textbf{targets} of interest (e.g.\ Turing machines or spin systems) and an object $C$ of \textbf{contexts} (e.g.\ input strings to Turing machines).
		Moreover, for any object $A$, the hom-set $\cC(A,T \otimes C)$ comes equipped with a preorder $\moonrel$ that is preserved by precomposition.
		Its interpretation can vary, but we think of it as an abstract version of a relation between pairs of a target and a context as in \eqref{eq:beh_relation_of_targets}.
		We are now ready to define a target--context category (see \cref{fig:target--context-category}). 
		
		\begin{figure}[t]\centering
			\includegraphics[width=.9\columnwidth]{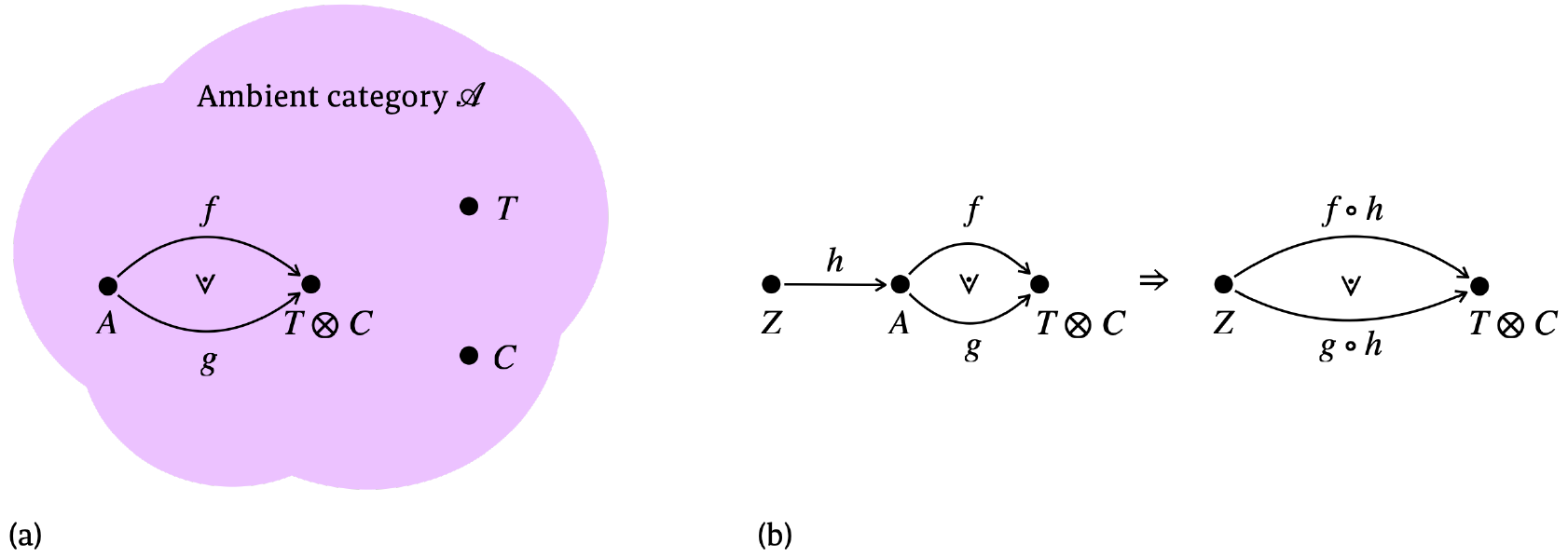}
			\caption{(a) A target--context category consists of an ambient category $\cC$ (which is quasi-total gs-monoidal) with distinguished objects $T$ (targets) and $C$ (contexts), as well as an ambient preorder relation $\mrel$ on every set $\cC(A,T \otimes C)$.  
				(b) This preorder is preserved by precomposition. }
			\label{fig:target--context-category}
		\end{figure}
		
		\begin{definition}\label{def:instance} 
			A \textbf{target--context category} $(\cC,T,C,\mrel)$ is a quasi-total gs-monoidal category $\cC$ with distinguished objects $T$ and $C$, as well as a preorder $\mrel$, called the \textbf{ambient relation}, on every set $\cC(A,T \otimes C)$, such that we have
			\begin{alignat}{4}
				\label{eq:mrel_domain} f &{}\sqsupseteq g   \quad  &&\implies \quad  & f &{}\mrelop g , 
				\intertext{where $\sqsupseteq$ is the restriction--to--domain preorder (\cref{def:dom_eq}), and} 
				\label{eq:mrel_precomp} f &{}\mrelop g  \quad  &&\implies \quad & f \circ h &{}\mrelop g \circ h
			\end{alignat}
			for arbitrary morphisms $f, g \colon A \to T \otimes C$ and $h \colon Z \to A$ of $\cC$.
		\end{definition}
		
		\begin{example}[Ambient categories for Turing machines]\label{ex:TM}
			A key example of universality is that of a universal Turing machine. 
			Let us construct two target--context categories, $\cat{Tur}$ and its variant $\cat{Tur}^\intr$, that serve the purpose of modelling these.
			The only distinction between the two is in the ambient relation $\mrel$.
			We will get back to this point in \cref{rem:TM_exampleS_diff}. 
			The construction of $\cat{Tur}$ follows \cite[example 3.2.1]{Co08d}.
			
			First, fix an arbitrary finite alphabet $\Sigma$ with at least two symbols. 
			We consider Turing machines, which have $\Sigma$ both as the input and the output alphabet.
			They compute partial functions of type $\Sigma^\star \to \Sigma^\star$ where $\Sigma^\star\coloneqq \cup_{n\geq 0} \Sigma^n$ denotes the set of all finite strings formed over $\Sigma$.
			By choosing a bijection between the set of Turing machines of this type and $\Sigma^\star$, we can think of $\Sigma^\star$ itself as the set of Turing machines.
			
			 Objects of the ambient category for both $\cat{Tur}$ and $\cat{Tur}^\intr$, denoted by $\cat{Comp}(\Sigma^\star)$, are $\Sigma^\star$, the singleton set $I$, and all finite cartesian products thereof.
			 Morphisms in $\cat{Comp}(\Sigma^\star)$ are partial computable functions. 
			 This defines a cartesian restriction category \cite{Co08d} and thus also a quasi-total gs-monoidal category with tensor product given by the cartesian product $\times$. 
			
			 Finally, to obtain a target--context category from $\cat{Comp}(\Sigma^\star)$ we take $T$ and $C$ to be $\Sigma^\star$.
			 Note that we interpret $(t,c) \in T \times C$ as consisting of a Turing machine $t$ and an input string $c$.
			 
			 Given any object $A$ of $\cat{Tur}$, the ambient relation $\mrel$ on the hom-set $\cat{Tur}(A, \Sigma^\star \times \Sigma^\star)$ is defined as follows:
			$f \mrelop g$ holds if and only if for all $a$ in the domain $\dom(g)$ of $g$, the value of $f(a)$ is defined and for the two pairs of a Turing machine and a string{\,\textemdash\,}namely $(t_{f}, c_{f}) \coloneqq f(a)$ and $(t_{g}, c_{g}) \coloneqq g(a)${\,\textemdash\,}it holds that either 
			\begin{itemize}
			  	\item both computations{\,\textemdash\,}running $t_{f}$ on $c_{f}$ and $t_{g}$ on $c_{g}${\,\textemdash\,}halt with equal outputs, or
			  	\item neither computation halts. %running $t_{g}$ on $c_{g}$ does not halt. \qedhere
			\end{itemize}
			 
			The only change for the second target--context category $\cat{Tur}^\intr$ is that we do not require that the Turing machine $t_{g}$ halts on the input $c_{g}$ whenever $t_{f}$ halts on $c_f$.
			To be more precise, in $\cat{Tur}^\intr$, $f \mrelop g$ holds if and only if for all $a$ in the domain $\dom(g)$ of $g$, $a$ is also in the domain of $f$ and either
			\begin{itemize}
			  	\item both computations{\,\textemdash\,}running $t_{f}$ on $c_{f}$ and $t_{g}$ on $c_{g}${\,\textemdash\,}halt with equal outputs, or
			  	\item running $t_{g}$ on $c_{g}$ does not halt. \qedhere
			\end{itemize}
		\end{example}
		
		\begin{remark}[Differences between the two ambient categories]\label{rem:TM_exampleS_diff}
			Both $\cat{Tur}$ and $\cat{Tur}^\intr$ model computations by Turing machines. 
			They differ only in their ambient relation.	
			More precisely, considering turing machines $t$, $t'$ as deterministic morphisms of type $I \to C$, the ambient relation $t \times \id_C \mrelop t' \times \id_C$ holds in $\cat{Tur}$ precisely if $t$ and $t'$ compute the same partial function. 
			In contrast, the ambient relation above holds in $\cat{Tur}^\intr$ if $t$ computes a so-called extension of \mbox{$t'$ \cite{En11}}.
			Consequently, $\cat{Tur}$ and $\cat{Tur}^\intr$ offer a slightly different treatment of universality. 
			As we explain in \cref{ex:TM univ sim}, (singleton) universal simulators in $\cat{Tur}$ can be associated with universal Turing machines, while (singleton) universal simulators in $\cat{Tur}^\intr$ can be associated with extensions of universal Turing machines.
			In this sense, $\cat{Tur}$ is a more appropriate framework instance to model Turing machine universality in the usual sense.			
			On the other hand, $\cat{Tur}^\intr$ has the beneficial property of intrinsic behavior structure (see \cref{ex:TM intrinsic}), which is also the reason for the notation we use to refer to it.
			Thanks to this property, we can use $\cat{Tur}^\intr$ to study consequences of universality in \cref{sec:Undecidability}.
			Since every universal simulator in $\cat{Tur}$ is also a universal simulator in $\cat{Tur}^\intr$ (see \cref{ex:TM_mrel_functor}), this approach automatically gives us also consequences of universality for $\cat{Tur}$.
		\end{remark}
	
		\Cref{ex:TM} helps us give an interpretation to the ambient relation $\mrel$. % (to be read as ``beak relation'').
		It expresses that for relevant inputs $a${\,\textemdash\,}those for which both $f$ and $g$ produce a behavior in terms of an output string{\,\textemdash\,}behaviors of $f$ and $g$ coincide.
		Provided such interpretation of the ambient relation $\mrel$ and an interpretation of $\cC$-morphisms and of the objects $T$ and $C$ of targets and context respectively that matches \cref{sec:motivation}, we can think of each target--context category as a concrete \emph{instance} of our framework for (studying) universality.
		The following technical lemma says that multiplying by a scalar $w$ necessarily moves morphisms down in the $\mrel$ order. 
		In $\cat{Rel}$, for example, a scalar is either the identity $\id_I$ or the empty relation.
		In the latter case, the lemma states that every relation (of the appropriate type) is more informative than the empty relation (of the same type). % same morphism multiplied by an undefined ``number''. 
		This lemma is used later in the proof of \cref{thm:s 2 id}.
		\begin{lemma}\label{lem:mrel_dom}
			Any $w\in \cC(I, I)$ and $f \in \cC(A,T\otimes C)$ in a target--context category satisfies 
			\begin{equation}\label{eq:mrel_dom}
				\input{mrel_dom.tikz}
			\end{equation}
		\end{lemma}
		\begin{proof}
			Let us prove $f \sqsupseteq (w \otimes f)$, from which relation \eqref{eq:mrel_dom} follows by the assumed implication \eqref{eq:mrel_domain}.
			First, observe that we have 
			\begin{equation}\label{eq:mrel_I}
				\input{mrel_I.tikz}
			\end{equation}
			The first equality is the definition of $\dom(w)$ (\cref{eq:domain}), the second follows since $\discard_I = \id_I$ and the last follows by the properties of $\cop$ and $\discard$ (\cref{eq:copy_delete}).
			Next we use that by the properties of $\cop$ and $\discard$ we have 
			\begin{equation}
				\dom(w \otimes f) = \dom(w) \otimes \dom(f) = w \otimes \dom(f)
			\end{equation}
			where the second equality is by \eqref{eq:mrel_I}.
			Finally, we show the requisite $f \sqsupseteq (w \otimes f)$ by computing 
			\begin{equation}
				f \circ \dom(w \otimes f) = w \otimes \bigl( f \circ \dom(f) \bigr) = w \otimes f
			\end{equation}
			where we use the assumption that every morphism is quasi-total (\cref{def:normalized}) in the second equality.
		\end{proof}
	
	\subsection{Simulators}\label{sec:simulators}
	
		We now turn our attention to simulators. 
		To specify a simulator for a given ambient category, we need to provide
		\begin{enumerate}
			\item an object $P$ of \textbf{programs}, and
			\item two morphisms of $\cC$: the \textbf{compiler} $s_T \colon P \to T$ and the \textbf{context reduction} $s_C \colon {P \otimes C \to C}$. 
		\end{enumerate}
		\begin{definition}\label{def:simulator}
			Let $(\cC,T,C,\mrel)$ be a target--context category.
		 	A \textbf{simulator} $s$ is a morphism $s \colon P \otimes C \to T \otimes C$ of $\cC$ that can be written as
			\begin{equation}\label{eq:simulator}
				\input{simulator.tikz}
			\end{equation}
			and satisfies
		 	\begin{equation}\label{eq:simulator_dom}
				\input{simulator_dom.tikz}
			\end{equation}
			where $s_T \colon P \to T$ is a functional morphism.
		\end{definition}
		\begin{remark}
Given any decomposition of $s$ as in \cref{eq:simulator}, not necessarily satisfying equations \eqref{eq:simulator_dom}, we can define $s'_T$ and $s'_C$ via\footnotemark{}
			\footnotetext{Strictly speaking, the compiler $s'_T$ is not defined by the left equation of \eqref{eq:simulator_rest}.
				However, as long as there exists a total state $c \colon I \to C$, it is uniquely specified by precomposing this equation with $\id_P {}\otimes{} c$.
				This applies in every example we consider.}%
		 	\begin{equation}\label{eq:simulator_rest}
				\input{simulator_rest.tikz}
			\end{equation}
			for which both \cref{eq:simulator} and \eqref{eq:simulator_dom} are satisfied, and for which $s'_T$ is functional because every domain of a quasi-total morphism is functional (\cref{lem:dom_fun}). 
			Specifically, \cref{eq:simulator} for the primed $s'_T$ and $s'_C$ on its right-hand side follows by
			\begin{equation*}
				\input{simulator_prime.tikz}
			\end{equation*}
			where the first equality is by \eqref{eq:simulator_rest}, the second one uses the associativity of copying, and the last one follows by the assumed decomposition of $s$ and the fact that both the compiler and the context reduction are quasi-total in the sense of \cref{def:normalized}.
			The appropriate versions of equations \eqref{eq:simulator_dom} are obtained by deleting $C$ and $T$ respectively in the assumed decomposition of $s$ into $s_T$ and $s_C$.
			
			In this sense, conditions \eqref{eq:simulator_dom} play no role in identifying which morphisms of type ${P \otimes C \to T \otimes C}$ constitute a simulator.
			Instead, they are used to pick out a distinguished decomposition into the compiler $s_T$ and the context reduction $s_C$:  
			For each simulator, there is a unique compiler and a unique context reduction, namely those given by equations \eqref{eq:simulator_dom}.
		\end{remark}
		
		The demand that the compiler be a functional morphism is a technical condition needed to establish equations \eqref{eq:processing_hit_2} and \cref{thm:s 2 id}.
		On the other hand, the splitting of $s$ into $s_T$ and $s_C$ as in diagram \eqref{eq:simulator} ensures that the target returned by $s$ (that is, the output of $\discard_C {}\circ{} s$) is independent of the context.
			
		\begin{definition}\label{def:trivial_sim}
			One simulator that always exists is the \textbf{trivial simulator}, which uses $T$ as a description of itself.
			It is given by the identity morphism ${\id \colon T \otimes C \to T \otimes C}$, where we identify $P \coloneqq T$, i.e.\ programs are given by targets.
		\end{definition}
		
		The trivial simulator plays a central role in our definition of universal simulators in \cref{sec:universality_def}. 
		In particular, it is universal by definition.
		If we think of the image of $s$ as the collection of target--context pairs it allows one to use, then the trivial simulator encompasses all targets in an arbitrary context by making use of any target in any context. 
		In that sense, it is particularly non-parsimonious (cf.\ \cref{sec:parsimony,ex:universal2trivial}).
		
		\begin{definition}\label{def:singleton_sim}
			A simulator $s$ is called a \textbf{singleton} simulator if there exists a functional state $t \colon I \to T$ such that its compiler can be written as:
			\begin{equation}
				\input{singleton_sim.tikz}
			\end{equation}
		\end{definition}
		Intuitively, a singleton simulator has a constant compiler. 
		If such a simulator is universal, it only uses a single target state $t$ to simulate all other targets. 
		
		\begin{example}[Singleton simulator for Turing machines]\label{ex:snglt_sim_TM}
			By definition of $\cat{Tur}$ from \cref{ex:TM},
			Turing machines are in one-to-one correspondence with elements of $T$ and hence with deterministic states $t \colon I \to T$.
			Given any Turing machine $t$ and any partial computable function ${\langle \ph, \ph \rangle \colon C \times C \to C}$, there is thus a singleton simulator
			\begin{equation}
				\input{point_sim.tikz}
			\end{equation}
			We denote simulators of this kind by $s_t$.
			Such simulators are of particular importance when $t$ is a universal Turing machine which comes equipped with a specific pairing function $\langle \ph, \ph \rangle$, as $s_t$ then captures the universality of $t$ in the sense of \cref{def:univ sim}.
			This is further explained in \cref{ex:TM univ sim}.
		\end{example}
	
		Further examples of simulators can be found in \cref{sec:Reductions}. 
				
	\subsection{Behavior Structure}\label{sec:behavior_structure}
		
		While property \eqref{eq:mrel_precomp} of the preorder suffices for the purposes of the general theory developed in \cref{sec:Reductions,sec:Morphisms}, we need a more specific definition in \cref{sec:Undecidability} when discussing generalizations of Lawvere's Fixed Point Theorem.
		This subsection introduces the notion of a behavior structure that is a formal account of the intuition presented in \cref{sec:motivation}.
		The reader may wish to refer to \cref{tab:summary_relations} for an overview of the preorder relations that appear throughout.
		
		Recall that we aim to avoid distinguishing targets in $T$ that have equivalent behavior in every context. 
		Moreover, we want to think of behaviors intuitively as elements of a set $B$.
		The specification of a behavior for a given target $t \in T$ and context $c \in C$ would then be a function $T \times C \to B$.
	
		Furthermore, even if behaviors $b_1, b_2 \in B$ are not equivalent, we may be able to say that behavior $b_1$ ``subsumes all aspects'' of behavior $b_2$.
		This amounts to a (not necessarily symmetric) relation $\brel$ on $B$.
		Even though the following definition is more general than this intuitive idea, we connect them in \cref{sec:pointed_shadow} using the so-called pointed shadow functor.
		Keeping this intuitive set-based notion in mind suffices to understand the usage of behavior structures in the remainder of this article.
		Likewise, it is not crucial to know in detail what a ``lax gs-monoidal functor'' (\cref{def:gs-monoidal_functor}) is.
		
		\begin{figure}[t]\centering
			\includegraphics[width=.75\columnwidth]{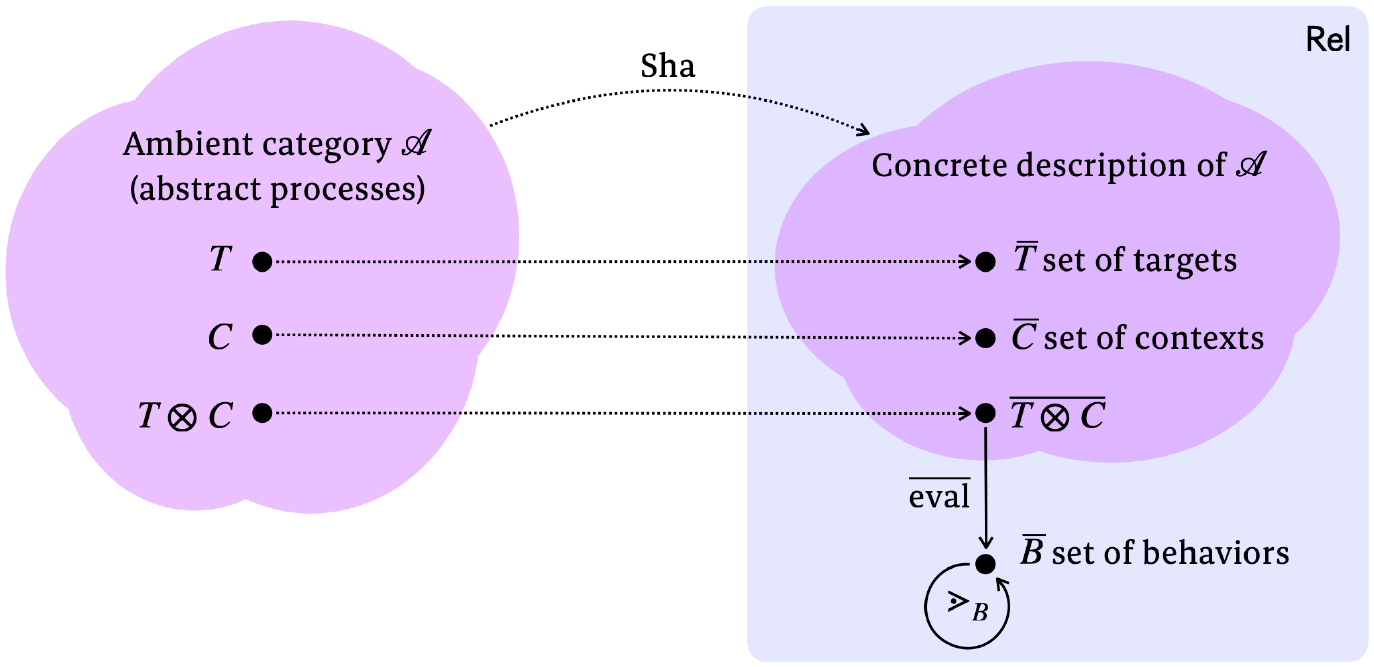}
			\caption{While the ambient category $\cC$ describes abstract processes, their ``observable consequences'' are manifested as relations between sets provided by the functor $\Beh$. To provide a behavior structure is to specify the latter (\cref{def:beh_structure})}
			\label{fig:behavior-structure}
		\end{figure}
	
		\begin{definition}\label{def:beh_structure}
			Let $\cC$, $T$ and $C$ be as in \cref{def:instance}.
			To provide a \textbf{behavior structure} $(\Beh, \behim{B}, \behim{\eval}, \brel)$ is to give a lax gs-monoidal functor $\Beh \colon \cC \to \cat{Rel}$, an object $\behim{B} \in \cat{Rel}$, a function 
			\begin{equation}
				\behim{\eval} \colon \Beh(T \otimes C) \to \behim{B},
			\end{equation}
			and a preorder relation $\brel$, called \textbf{behavioral relation}, on the set $\behim{B}$.
		\end{definition}
		
		Intuitively, we can think of $\Beh$ as a way to assign, to each object $A \in \cC$ and each morphism $f$ in $\cC$, a particular set denoted $\behim{A} \coloneqq \Beh(A)$ and a relation denoted $\behim{f} \coloneqq \Beh(f)$ respectively.\footnotemark{}
		\footnotetext{Even though we write $\behim{\eval}$ for the evaluation function and $\behim{B}$ for the set of behaviors, we \emph{do not} assume that there is a morphism $\eval$ and an object $B$ in $\cC$ whose images under $\Beh$ are $\behim{\eval}$ and $\behim{B}$ respectively (cf.\ intrinsic behavior structure in \cref{sec:intrinsic_beh}).}%
		Additionally, having a lax gs-monoidal functor means, in short, that there is a consistent way to relate each set $\behim{X} \times \behim{Y}$ to the set $\behim{X \otimes Y}$ (see \cref{def:gs-monoidal_functor} for more details).
	
		\begin{remark}[Role of the shadow functor]
			The \newemph{shadow functor} $\Beh$ provides a way to relate the intrinsic description of processes as abstract elements of the ambient category $\cC$ and their concrete description as certain (partial, multi-valued) functions{\,\textemdash\,}morphisms of $\cat{Rel}$.
			The distinction between these two descriptions allows us to 
			\begin{itemize}
				\item introduce constraints on the permitted morphisms of $\cC$ (such as constraints on their computability or scaling properties), independently of 
				
				\item the evaluation function $\behim{\eval}$, which tells us the ``observable'' behavior{\,\textemdash\,}an element of $\behim{B}${\,\textemdash\,}of a target $\behim{t} \in \behim{T}$ in some context $\behim{c} \in \behim{C}$.
			\end{itemize}
			One can contrast this set-up with the notion of intrinsic behavior structure (\cref{sec:intrinsic_beh}) where evaluation is a morphism of $\cC$, and therefore has to comply with any of these constraints. 
		\end{remark}
	
		\begin{figure}[t]\centering
			\includegraphics[width=.6\columnwidth]{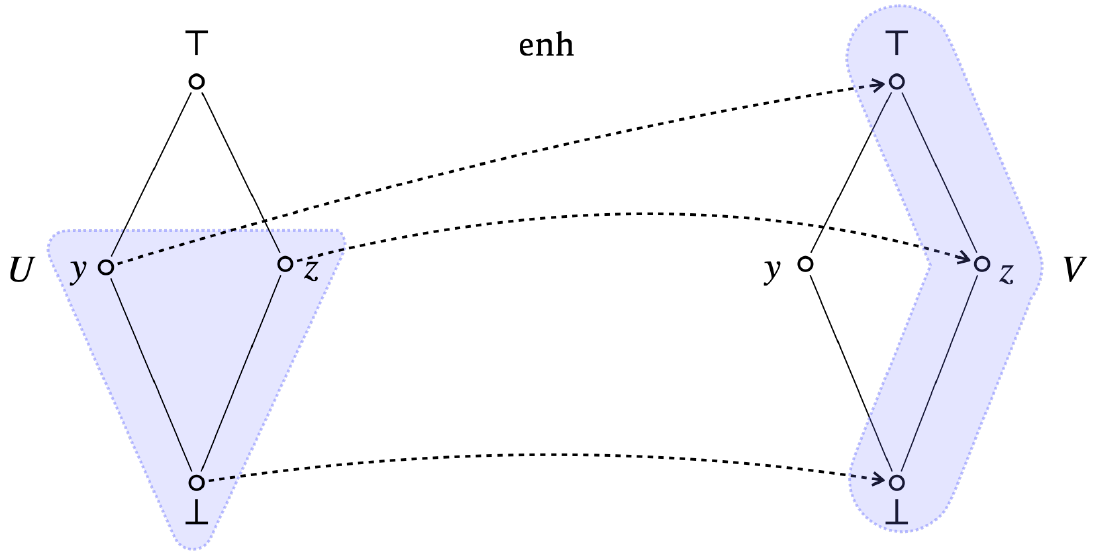}
			\caption{An example of an enhancement map $\textrm{enh} \colon U\to V$ where $U$ and $V$ are subsets of $X= \{\top, y, z,\bot\}$, and where $\succeq$ is generated by the Hasse diagram depicted with solid lines.
			The map takes any element in $U$ to an element in $V$ which is above in the $\succeq$ relation. 
			Note that in this example $\mathrm{enh}$ is not a degradation map (see \cref{fig:degradation}).}
			\label{fig:enhancement}
		\end{figure}
	
		\begin{figure}[t]\centering
			\includegraphics[width=.6\columnwidth]{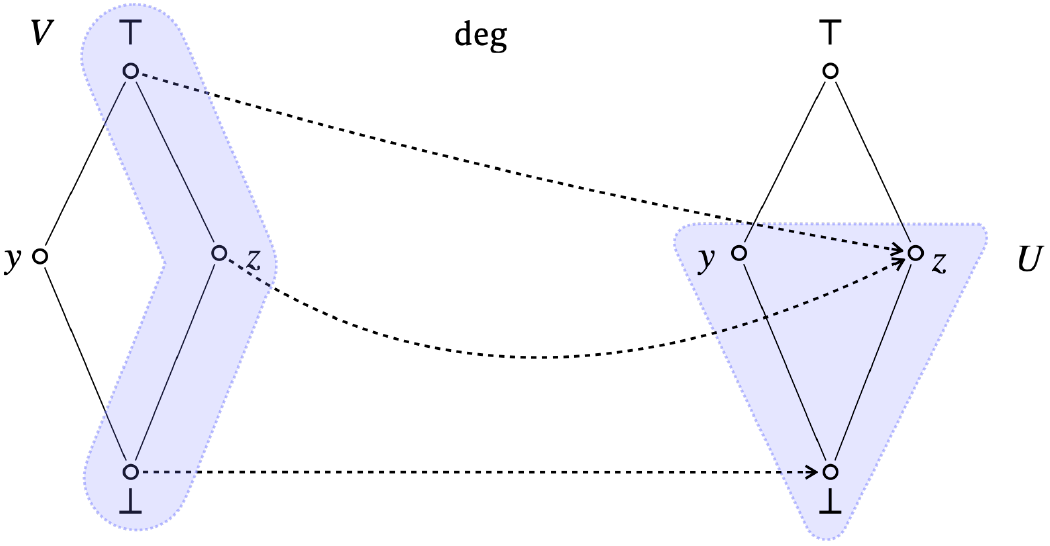}
			\caption{An example of a degradation map $\textrm{deg} \colon V\to U$ where $U$ and $V$ are subsets of $X= \{\top, y,z, \bot\}$,  and where $\succeq$ is generated by the Hasse diagram depicted with solid lines.
			The degradation map takes any element in $V$ to an element in $U$ which is below in the $\succeq$ relation. 
			Note that in this example $\mathrm{deg}$ is not an enhancement map (see \cref{fig:enhancement}).}
			\label{fig:degradation}
		\end{figure}
		
		Given a behavior structure we obtain a preorder relation $\mrel$ on the hom-set $\cC(A,T \otimes C)$ for any $A \in \cC$, as required by \cref{def:instance}.
		Before we give its details in \cref{def:mrel}, we introduce two order-theoretic concepts: enhancements that ``move elements upward'' (\cref{fig:enhancement}) and the dual notion of degradations that ``move elements downward'' (\cref{fig:degradation}).
		\begin{definition}\label{def:enh_deg}
			Let $(X,\succeq)$ be a preordered set with subsets $U$ and $V$.
			Consider two functions $f \colon U \to V$ and $g \colon V \to U$.
			We say that 
			\begin{itemize}
				\item $f$ is an \textbf{enhancement map} if $f(u) \succeq u$ holds for all $u \in U$.
				\item $g$ is a \textbf{degradation map} if $v \succeq g(v)$ holds for all $v \in V$.
			\end{itemize} 
		\end{definition}
		
		\begin{remark}[Interpretation of enhancement and degradation]
			In the first author's PhD thesis \cite[section 2.4]{Gonda2021}, enhancement and degradation maps are used to construct two ways of extending a given preorder from a set $X$ to its power set.
			The two orderings of subsets of $X$ carry distinct interpretation.
			If an enhancement map of type $U \to V$ exists, then every element of $U$ is below some element of $V$.
			On the other hand, if a degradation map of type $V \to U$ exists, then every element of $V$ is above some element of $U$.
			Both of these statements assert that $V$ is, in some sense, above $U$.
			They are nevertheless distinct.
			In our case, we use the intersection of these orderings in \cref{def:brel} to conform to the intended interpretation of the imitation relation $\mrel$ (\cref{rem:imitation_meaning}).
		\end{remark}
		
		\begin{definition}\label{def:brel}
			Consider an ordered set $(X, \succeq)$.
			Given two relations $\nu,\mu \in \cat{Rel}(A,X)$, we say that $\nu$ \textbf{imitates} $\mu$, denoted $\nu \succeq^{\rm im} \mu$, if for every element $a \in \dom(\mu)$ of the domain of $\mu$ (thought of as a subset of $A$), both of the following conditions hold:
			\begin{enumerate}
				\item There exists an enhancement map $\mathrm{enh}_a \colon \mu(a) \to \nu(a)$.
				\item There exists a degradation map $\mathrm{deg}_a \colon \nu(a) \to \mu(a)$.
			\end{enumerate}
			This defines the \textbf{imitation relation} $\succeq^{\rm im}$ on the set $\cat{Rel}(A,X)$ of relations for any $A$.
		\end{definition}
		\begin{remark}[Imitation relation for functions]\label{ex:imitation_functions}
			In the special case when both $\nu$ and $\mu$ are partial functions, $\nu$ imitates $\mu$ if and only if $\nu(a) \succeq \mu(a)$ holds for every $a$ in the domain of $\mu$.
			This is because, for singleton sets $\{n\}, \{m\} \subseteq X$, the following are equivalent \cite[\mbox{appendix A}]{gonda2023monotones}:
			\begin{enumerate}
				\item \label{it:preorder_fun} $n \succeq m$.
				\item \label{it:enhorder_fun} There exists an enhancement map $\{m\} \to \{n\}$.
				\item \label{it:degorder_fun} There exists a degradation map $\{n\} \to \{m\}$.\qedhere
			\end{enumerate}
		\end{remark}
		
		If $(X, \succeq)$ is the set of behaviors specified by a behavior structure, we obtain an imitation order $\imrel$ on any hom-set $\cat{Rel}(A,\behim{B})$.
		In \cref{def:mrel} and \cref{lem:brel_precomp}, we use it to show that every behavior structure gives rise to a target--context category.
		Ambient categories arising in this way thus come equipped with a more concrete description of the ambient relation $\mrel$ than a generic target--context category (\cref{def:instance}) does.
		
		\begin{definition}\label{def:mrel}
			Let $(\Beh, \behim{B}, \behim{\eval}, \brel)$ be a behavior structure.
			Provided $f,g \in \cC(A,T \otimes C)$, we say that $\bm{f}$ \textbf{ambient-imitates} $\bm{g}$, denoted $f \mrelop g$, if the relation $\behim{\eval} \circ \behim{f}$ imitates $\behim{\eval} \circ \behim{g}$ according to the imitation relation $\imrel$.
			This defines the ambient imitation relation $\mrel$ on any hom-set of the form $\cC(A,T \otimes C)$.
		\end{definition}
		\Cref{fig:behavior-relation} depicts the elements appearing in this definition.
		To understand its motivation, we first need to introduce universality of simulators in \cref{def:univ sim}.
		We thus return to the interpretation in \cref{rem:imitation_meaning}, which motivates the seemingly complicated definition of the ambient imitation relation.
	
		\begin{figure}[t]\centering
			\includegraphics[width=.9\columnwidth]{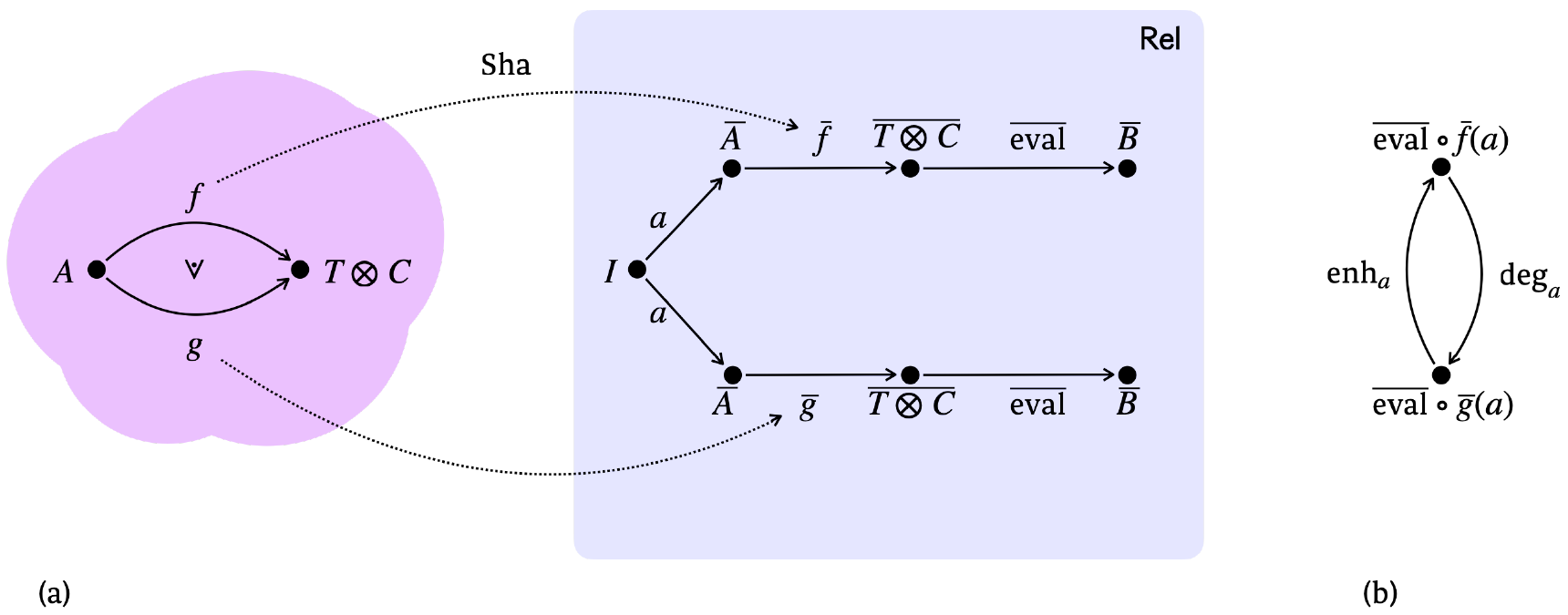}
			\caption[behavior-relation]{(a) The ambient imitation preorder $\mrel$ (\cref{def:mrel}) between morphisms of type $A \to T \otimes C$ in the ambient category (left-hand side) is defined via the imitation preorder of their shadows ($\behim{f}$ and $\behim{g}$) in the category of relations, when composed with the evaluation function.
			Evaluation is used to generate a ``behavior-valued relation'' of type $\behim{A} \to \behim{B}$. 	
			(b) Behaviors of shadows ($\behim{\eval}\circ\behim{f}$ and $\behim{\eval}\circ\behim{g}$) are related by the imitation relation $\imrel$ (\cref{def:brel}) precisely if there are $a$-dependent enhancement and degradation maps as specified in the figure.}
			\label{fig:behavior-relation}
		\end{figure}
	
		\begin{definition}\label{def:instance_beh}
			A \textbf{target--context category with behaviors} $(\cC,T,C,\Beh,\behim{B},\behim{\eval},\brel)$ is a target--context category $(\cC,T,C,\mrel)$ in which the preorders $\mrel$ are the ambient imitation relations from \cref{def:mrel}.
		\end{definition}
		
		\begin{example}[Behavior structure for Turing machines] \label{ex:TM_behaviors}
			In \cref{ex:TM}, we define target--context categories $\cat{Tur}$ and $\cat{Tur}^\intr$ for Turing machines. 
			In particular, they come with a preorder on each hom-set of morphisms whose output is $T \otimes C$.
			
			Both are target--context categories with behaviors, as we now spell out.
			\begin{enumerate}
				\item \emph{Behavior structure for $\cat{Tur}^\intr$}: 
					The candidate set $\behim{B}$ of behaviors is the set $\Sigma^\star$ of all finite strings, and the candidate $\behim{\eval}$ is the partial function that maps a pair $(t,c)$ to the output $t(c)$ of Turing machine $t$ when evaluated on input $c$.
					If we take $\brel$ to be the equality relation on $\behim{B}$ and $\Beh$ to be the inclusion functor of $\cat{Tur}$ into $\cat{Rel}$, then this choice makes $\cat{Tur}$ into a target--context category with behaviors.
					One can verify that the resulting ambient imitation relation coincides with the relation given in \cref{ex:TM}.
				\item \emph{Behavior structure for $\cat{Tur}$}:
					Here, satisfying the relation $f \mrelop g$ places an additional requirement.
					Namely, the domains of such $f$ and $g$ necessarily coincide.
					For this reason, the behaviors are different, as they also need to include information about the domain.
					While there are multiple implementations, one is to choose $\behim{B} \coloneqq \Sigma^\star \times \pwrset(\Sigma^\star)$, so that behaviors are pairs of a string and a set of strings.
					We think of the former as the reading of the output tape of a Turing machine, and the latter as the set of input strings for which a Turing machine halts.
					That is, $\behim{\eval}$ in this case maps a pair $(t,c)$ to the pair ${\bigl( t(c), \dom(t) \bigr) \in \behim{B}}$.
					Once again, $\brel$ is the equality relation on $\behim{B}$ and $\Beh$ is the inclusion functor of $\cat{Tur}$ into $\cat{Rel}$.
			\end{enumerate}
			Moreover, since $\brel$ is equality in both cases, \cref{prop:beh_equality} applies and we can characterize the ambient imitation relation via $\sqsupseteq$ as in equivalence \eqref{eq:beh_equality}.
		\end{example}

		The following two results show that ambient imitation relations indeed satisfy the assumptions of a target--context category.
			
		\begin{lemma}\label{lem:brel_precomp}
			Ambient imitation relations are preorders preserved by precomposition.
		\end{lemma}
		\begin{proof}
			The claim follows if we can show that the imitation relation itself is preserved by precomposition (in $\cat{Rel}$).
			Let us therefore argue that for relations $\nu,\mu \in \cat{Rel}(A, X)$ and $\xi \in \cat{Rel}(Z,A)$, we have
			\begin{equation}\label{eq:imrel_precomp}
				\nu \succeq^{\rm im} \mu  \quad \implies \quad  \nu \circ \xi \succeq^{\rm im} \mu \circ \xi.
			\end{equation}
			
			To establish the consequent of implication \eqref{eq:imrel_precomp}, consider an arbitrary $z \in Z$ that is in the domain of $\mu \circ \xi$.
			There is necessarily at least one element of $\xi(z) \subseteq A$ that is also in the domain \mbox{of $\mu$}.
			Let us denote the (necessarily non-empty) intersection by
			\begin{equation}
				W \coloneqq \xi(z) \cap \dom(\mu)
			\end{equation}
			where we think of $\dom(\mu)$ as a subset of $A$.
			
			Assuming that $\nu$ imitates $\mu$, we have enhancement and degradation maps
			\begin{align}
				\mathrm{enh}_w &\colon \mu(w) \to \nu(w)  &  \mathrm{deg}_w &\colon \nu(w) \to \mu(w)
			\end{align}
			for each $w \in W$.
			This allows us to pick enhancement and degradation maps
			\begin{align}
				\mathrm{enh} &\colon \mu \circ \xi(z) \to \nu \circ \xi(z)  &  \mathrm{deg} &\colon \nu \circ \xi(z) \to \mu \circ \xi(z).
			\end{align}
			For example, for any $b \in \mu \circ \xi(z)$, there is at least one $w_b$ such that $b$ is an element of $\mu(w_b)$ so that we can define
			\begin{equation}
				\mathrm{enh}(b) \coloneqq \mathrm{enh}_{w_b}(b)
			\end{equation} 
			and similarly for the degradation map.
			While $\mathrm{enh}$ and $\mathrm{deg}$ defined as such are not unique, we nevertheless obtain the desired statement that $\nu \circ \xi$ imitates $\mu \circ \xi$.
		\end{proof}
		
		We now show that if a morphism $f$, when restricted to the domain of $g$, coincides with $g$ (\cref{def:dom_eq}), then $f$ also behaviorally imitates $g$.
		\begin{proposition}[Morphisms behaviorally imitate restrictions to their domain]\label{prop:beh_restriction}
			Consider a behavior structure $(\Beh, \behim{B}, \behim{\eval}, \brel)$ and let $\mrel$ be the corresponding ambient imitation relation.
			Then we have
			\begin{equation}\label{eq:beh_restriction}
				f \sqsupseteq g  \quad \implies \quad  f \mrelop g
			\end{equation}
			 for all $f,g \in \cC(A, T \otimes C)$. 
		\end{proposition}
		\begin{proof}
			The relation $f \sqsupseteq g$, i.e.\ equality of $g$ and $f$ within the domain of $g$, says
			\begin{equation}
				g = f \circ \dom(g).
			\end{equation}
			Applying $\Beh$ to this equation, and using first its functoriality and then \cref{prop:sha_dom}, we obtain
			\begin{equation}\label{eq:dom_restriction_order}
				\behim{g} = \behim{f} \circ \behim{\dom(g)} = \behim{f} \circ \dom \bigl( \behim{g} \bigr).
			\end{equation}
			
			This means that for any element $a\in \dom(\behim{g})$ of the domain of $\behim{g}$ (thought of as a subset of $\behim{A}$), $\behim{g}(a)=\behim{f}(a)$ holds, and thus we also have
			\begin{equation}
				\behim{\eval}\circ \behim{g}(a)=\behim{\eval}\circ\behim{f}(a).
			\end{equation}
			Hence the identity $\id_{\behim{A}}$ is both an enhancement map of type
			\begin{equation}
				\behim{\eval}\circ\behim{g}(a) \to \behim{\eval}\circ\behim{f}(a)
			\end{equation}
			and a degradation map of type
			\begin{equation}
				\behim{\eval}\circ\behim{f}(a) \to \behim{\eval}\circ\behim{g}(a).
			\end{equation}
			This is precisely the required $f \mrelop g$ by \cref{def:mrel}.
		\end{proof}
		
		In many examples of ambient categories, we also have a converse of a slightly stronger version of implication \eqref{eq:beh_restriction} (i.e.\ the right-to-left implication in \eqref{eq:beh_equality}).
		This often happens because the behavioral relation $\brel$ is equality on the set $\behim{B}$ of behaviors.
		\begin{proposition}\label{prop:beh_equality}
			Consider a behavior structure $(\Beh, \behim{B}, \behim{\eval}, =)$ and let $\mrel$ be the corresponding ambient imitation relation.
			Then we have
			\begin{equation}\label{eq:beh_equality}
				f \mrelop g  \quad \iff \quad  \behim{\eval} \circ \behim{f} \,\sqsupseteq\, \behim{\eval} \circ \behim{g}
			\end{equation}
			 for all $f,g \in \cC(A, T \otimes C)$. 
		\end{proposition}
		\begin{proof}	
			Note that for two subsets $U$ and $V$ of the preordered set $(\behim{B}, =)$, we have that 
			\begin{itemize}
				\item there exists and enhancement $U \to V$ if and only if $U$ is a subset of $V$, and
				\item there exists a degradation $V \to U$ if and only if $V$ is a subset of $U$.
			\end{itemize}
			Therefore, the associated imitation relation $\nu =^{\rm im} \mu$ between $\mu, \nu \in \cat{Rel}(A,\behim{B})$ is precisely the restriction to domain relation $\nu \sqsupseteq \mu$ from \cref{def:dom_eq}.
			By \cref{def:mrel}, we thus have
			\begin{equation}
				f \mrelop g  \quad \iff \quad   \behim{\eval} \circ \behim{f} \,=^{\rm im}\, \behim{\eval} \circ \behim{g} \quad \iff \quad  \behim{\eval} \circ \behim{f} \,\sqsupseteq\, \behim{\eval} \circ \behim{g}
			\end{equation}
			which is what we aimed to show.
		\end{proof}
		
	\subsection{Intrinsic Behavior Structure}\label{sec:intrinsic_beh}

		In \cref{def:mrel} of a behavior structure, the evaluation function is (potentially) outside of the image of the shadow functor, see also \cref{fig:behavior-structure}.
		In such case, to compare morphisms with respect to the ambient imitation relation $\mrel$, we must consider their shadows.
			
		\begin{figure}[t]\centering
			\includegraphics[width=.7\columnwidth]{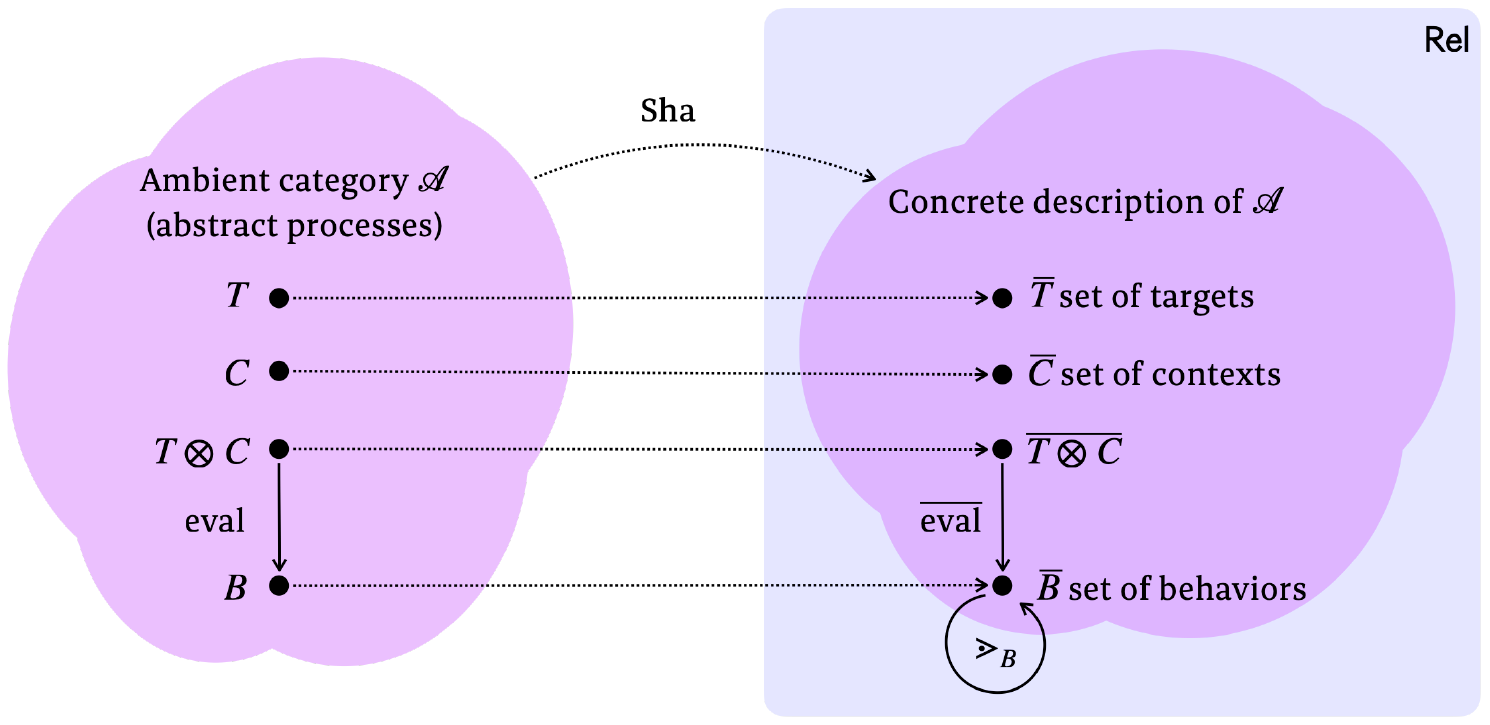}
			\caption{In an intrinsic behavior structure, the evaluation function $\eval$ comes from $\cC$ (compare with \cref{fig:behavior-structure}). }
			\label{fig:intrinsic-behavior}
		\end{figure}
		
		The alternative to that is if the evaluation function has a corresponding representation in the ambient category{\,\textemdash\,}a morphism $\eval$ whose shadow is the function $\behim{\eval}$ (see \cref{fig:intrinsic-behavior}).
		This is particularly useful if the shadow functor is faithful, but it can be used also more generally.
		For instance, a sufficient condition for $f$ to behaviorally imitate $g$ is that their evaluations in $\cC$ coincide, i.e.\ we have $\eval \circ f = \eval \circ g$, which can be shown within $\cC$ itself without considering the shadows.
		\begin{definition}\label{def:intrinsic_beh}
			Consider a target--context category with behaviors given by the tuple $(\cC,T,C,\Beh,\allowbreak \behim{B}, \behim{\eval}, \brel)$ as in \cref{def:instance_beh}.
			We say that $\cC$ has \textbf{intrinsic behaviors} if there exists an object $B \in \cC$ and a morphism ${\eval \colon T \otimes C \to B}$ such that we have
			\begin{equation}
				\behim{B} \coloneqq \Beh(B)  \quad \text{ and } \quad  \behim{\eval} \coloneqq \Beh(\eval).
			\end{equation}
		\end{definition}
		
		\begin{example}[Turing machines with intrinsic behavior structure] \label{ex:TM intrinsic}
			In \cref{ex:TM_behaviors}, we construct the behavior structures for $\cat{Tur}$ and $\cat{Tur}^\intr$, the two target--context categories for Turing machines. 
			Here, following \cite[example 3.2.1]{Co08d}, we illustrate that $\cat{Tur}^\intr$ has an intrinsic behavior structure in the sense of \cref{def:intrinsic_beh}.
			
			In order for the behavior structure from \cref{ex:TM_behaviors} to be intrinsic, one has to show that $\behim{\eval}$ is itself a morphism (denoted by $\eval$) of $\cat{Tur}^\intr${\,\textemdash\,}a (partial) computable function.
			This can be proven as a \emph{consequence} of the existence of a universal Turing machine.
			The computability of $\eval$ is precisely the statement that there exists a Turing machine $u$ that takes an arbitrary Turing machine $t$ as part of its input and reproduces its behavior on the rest of the input. 
			This argument further relies on another assumption, namely that the enumeration of Turing machines used to identify $T$ with the set of strings $\Sigma^\star$ in \cref{ex:TM} allows one to construct such a universal Turing machine $u$.
			While not every enumeration would work, a suitable one can be chosen.
			Such enumerations are often called admissible or acceptable numberings \cite{So99}.
			
			On the other hand, the above argument does not hold for the target--context category $\cat{Tur}$.
			This is because the halting problem for Turing machines is famously not computable. 
			Indeed, if the evaluation function from $\cat{Tur}$ was computable, i.e.\ if one could decide for every Turing machine and every input whether it is in the domain of the Turing machine or not, then one could also decide the halting problem. \qedhere
		\end{example}
	
	\begin{remark}
		In \cref{app:intrinsification} we prove that given any target--context category with behaviors (\cref{def:instance_beh}) we can construct a target--context category with intrinsic behaviors (\cref{def:intrinsic_beh}) that is based on the same ambient category, same targets and same contexts.
		This constructions is such that whenever the ambient relation $f \mrelop g$ holds in the one with intrnsic behaviors, then we also have the corresponding ambient relation in the former target--context category.
	\end{remark}

\section{Reductions and Universality of Simulators}
\label{sec:Reductions}

	We are now ready to define the universality of a simulator (\cref{def:univ sim}).
	Then we provide several examples of universal simulators: universality of simulators for Turing machines (\cref{ex:TM univ sim}), in computational complexity (\cref{ex:NP_complete}) and universality of spin models  (\cref{sec:spinmodel}). 
	In \cref{sec:math_examples}, we discuss a selection of examples that showcase how universal simulators recover other concepts from mathematics, such as universal sets, dense subsets, and cofinal subsets. 
	Finally, in \cref{sec:no-go_theorems}, we provide a no-go theorem for universality, i.e.\ we find conditions under which a simulator cannot be universal.  
	
	\subsection{Definition of Universality}\label{sec:universality_def}
	
		The upcoming definition of universality (\cref{def:univ sim}) hinges on the notion of a reduction. 
		Intuitively, reductions are transformations of simulators that only have access to the programs.
		By relating programs of a simulator $s$ to those of another one $s'$, they cannot enlarge the range of targets that a simulator encompasses.

		\begin{definition}\label{def:reduction}
			Consider two simulators $s$, $s'$ with programs given by $P$ and $P'$ respectively.
			A \textbf{reduction} $r^* \colon s \to s'$ of simulators is specified by a functional morphism $r \colon P' \to P$ of $\cC$ such that 
			\begin{equation}\label{eq:reduction}
				\input{reduction.tikz}
			\end{equation}
			holds.
		\end{definition}
		\begin{definition}\label{def:lax_reduction}
			A \textbf{lax reduction} is one for which we only require that $s \circ r$ behaviorally imitates $s'$, i.e.\ that \cref{eq:reduction} holds with the ambient imitation relation $\mrel$ in place of equality.
			On the other hand, for an oplax reduction it is $s'$ that behaviorally imitates $s \circ r$.
		\end{definition}
		\begin{definition}\label{def:univ sim}
			A simulator $s$ is a \textbf{universal simulator} if there exists a lax reduction ${r^* \colon s \to \id}$ from $s$ to the trivial simulator, i.e.\ such that we have
			\begin{equation}\label{eq:univ_sim_def_2}
				\input{univ_sim_def_2.tikz}
			\end{equation}
		\end{definition}
		A generic universal simulator is one that is able to reproduce the behavior of the trivial simulator.
		In particular, the trivial simulator is always universal, because $\mrel$ is a reflexive relation (\cref{lem:brel_precomp}).
	
		\begin{remark}[Uses of the word universal] 
			While \cref{def:univ sim} defines the notion of a universal simulator, note that `universal' is used in many other ways (universal Turing machine, universal construction, universal spin models, etc).
			As we show in \cref{ex:TM univ sim,sec:spinmodel,ex:cofinal}, some of these correspond to universal simulators. 
		\end{remark}
		
		\begin{remark}[Interpretation of the imitation relation]\label{rem:imitation_meaning}
			Having introduced the key concept in our framework{\,\textemdash\,}universality of simulators{\,\textemdash\,}we now return to the motivation for the ambient imitation relation (\cref{def:mrel}).
			First of all, note that the requirement that the imitation relation $\imrel$ from \cref{def:brel} only be tested for inputs within the domain of $g$ translates to the fact that we only care about those pairs of a target $t \in \cC_\det(I,T)$ and a context $c \in \cC_\det(I,C)$ which give rise to \emph{some behavior}, i.e.\ that the set
			\begin{equation}
				U \coloneqq \behim{\eval} \circ \bigl( \behim{t \otimes c} \bigr)
			\end{equation}
			in $\behim{B}$ is non-empty.
			Let us also define another set of behaviors via the left-hand side of \eqref{eq:univ_sim_def_2}:
			\begin{equation}
				V \coloneqq \behim{\eval} \circ \behim{s} \circ \bigl( \behim{(r \, t) \otimes c} \bigr)
			\end{equation}
			Relation \eqref{eq:univ_sim_def_2} then says that
			\begin{enumerate}
				\item \label{it:enh} there exists an enhancement map $\mathrm{enh} \colon U \to V$, and
				\item \label{it:deg} there exists a degradation map $\mathrm{deg} \colon V \to U$.
			\end{enumerate}
			
			The enhancement map ensures that, for each target behavior $u \in U$, there is a behavior $\mathrm{enh}(u) \in V$ of the compiled version
			\begin{equation}\label{eq:univ_sim_def_4}
				\input{univ_sim_def_4.tikz}
			\end{equation}
			of the program $r\circ t \in \cC_\det(I,P)$ in context
			\begin{equation}\label{eq:univ_sim_def_5}
				\input{univ_sim_def_5.tikz}
			\end{equation}
			that subsumes the behavior of $t$ in context $c \colon I \to C$.
			In other words, relation \eqref{eq:univ_concrete} holds.
		
			On the other hand, the degradation map ensures that any behavior in $v \in V$, if manifested by $s \circ r$, will subsume at least one of the target behaviors in $U$. 
			While condition \ref{it:enh}, when applied to all target--context pairs, can be read as
			\begin{center}
				$s \circ r$ is sufficient to behaviorally subsume $\id_{T \otimes C}$,
			\end{center}
			condition \ref{it:deg} says that
			\begin{center}
				 $s \circ r$ is faithful in its simulation of all behaviors manifested by target--context pairs.
			\end{center} 
			We require the latter because we think of situations in which the agent implementing the simulator may not have control over its output and over the specific behavior $v \in V$ that is manifested{\,\textemdash\,}a valid simulation should therefore succeed independently of these hidden variables.
		\end{remark}
	
		\begin{example}[Universal Turing machine as a universal simulator]\label{ex:TM univ sim}
			Consider the ambient category $\cat{Tur}$ from \cref{ex:TM}.
			Let $u$ be a universal Turing machine.
			This means that there exist total computable functions $r \colon T \to \Sigma^\star$ and $\langle \ph , \ph  \rangle \colon T \times C \to C$; such that, for all Turing machines $t \in T$ and all input strings $c \in C$, running $t$ on $c$ returns the same output as running $u$ on $\langle r(t), c \rangle$.
			
			Consider the singleton simulator $s_u$ defined by $u$ and $\langle \ph , \ph \rangle$ (see \cref{ex:snglt_sim_TM}). 
			Programs for this simulator are strings from $\Sigma^\star$.
			Then $s_u$ is a universal simulator in the sense of \cref{def:univ sim} if and only if $u$ is a universal Turing machine.
			The lax reduction to the trivial simulator is provided by the aforementioned $r$, since we have
			\begin{equation}
				\eval \bigl( s_u(r(t),c) \bigr) = \eval \bigl( u, \langle r(t), c\rangle  \bigr) = \eval(t,c)
			\end{equation}
			where the first equality follows by definition of $s_u$ and second equality follows from the assumption that $u$ is a universal Turing machine.
		\end{example}

		\begin{example}[An NP-complete language as a universal simulator]\label{ex:NP_complete}
			Fix a finite alphabet $\Sigma$ and let us construct the category $\cat{Poly}(\Sigma^\star)$ whose objects are $\Sigma^\star$, the singleton set $I$ and all finite cartesian products thereof.
			Its morphisms are chosen to be deterministic poly-time computable maps.
			More precisely, they are functions $f \colon (\Sigma^\star)^k \to (\Sigma^\star)^l$ such that there exists a polynomial $\mathfrak{p}_f$ and a deterministic multi-tape Turing machine $t$ with $\max(l,k)$ tapes, such that for all $(x_1, \ldots , x_k) \in (\Sigma^\star)^k$ 
			when running $t$ with $x_i$ written on its $i$-th tape, $t$ halts in runtime at most $\mathfrak{p}_f(\sum_i \abs{x_i})$ and with $f_j(x)$ written on its $j$-th tape.  
			From this definition it is clear that all identity functions are poly-time computable and, moreover, the composite of any two poly-time computable functions is poly-time computable. %, so this defines a category.
			
			To construct a target--context category, we take $T= \Sigma^\star$ and $C=I$.
			Next, we fix an arbitrary NP-complete language $L \subset \Sigma^\star$.
			We equip $\cat{Poly}(\Sigma^\star)$ with a behavior structure by choosing the shadow functor $\Beh$ to be the embedding of $\cat{Poly}(\Sigma^\star)$ into $\cat{Rel}$.
			Furthermore, the behaviors are Booleans: 
			We set $\behim{B} = \{0,1\}$ and equality as the behavioral relation $\brel$.
			Finally, the evaluation function $\behim{\eval} \colon \behim{T} \to \behim{B}$ is the characteristic function of $L$:
			\begin{equation}
				\behim{\eval}(t) = 
				\begin{cases}
					1 & \text{if }  t \in L \\
					0 & \text{otherwise}
				\end{cases}
			\end{equation}
			Note that since $L$ is NP-complete, $\behim{\eval}$ is not deterministic poly-time computable (unless P $=$ NP).
			However, it is poly-time computable by a \emph{non-deterministic} Turing machine.
			We now explain how NP-complete languages correspond to universal simulators in this target--context category and vice versa.
			
			First observe that each morphism of type $f \colon A \to T$ gives rise to an NP language. 
			Specifically, for any $f\colon A \to T$, we can define the language
			\begin{equation}
				L_f \coloneqq \Set*[\big]{ a \in A  \given  \behim{\eval} \circ \behim{f}(a) = 1 }.
			\end{equation}
			Since $\behim{\eval}$ is non-deterministic poly-time computable, $L_f$ is an NP language.
			Note that by construction $f$ is a reduction of languages from $L_f$ to $L$.
			Conversely, for any NP language $L' \subset A$, 
			there exists a poly-time reduction $f_{L'}$ from $L'$ to $L$ because $L$ is NP-complete. 
			This function defines a morphism $f_{L'} \colon A \to T$ in $\cat{Poly}(\Sigma^\star)$.
			
			Now take any NP-complete language $L' \subset A$ and consider the simulator $s$ with compiler given by $s_T = f_{L'}$ and context reduction $s_C = \discard_A$.
			Since $L'$ is NP-complete there necessarily exists a poly-time reduction $r$ from $L$ to $L'$, which is a morphism of type $r \colon T \to A$ in $\cat{Poly}(\Sigma^\star)$ that satisfies 
			\begin{equation}\label{eq:NP complete uni}
			\begin{split}
				\behim{\eval}\circ \behim{s_T} \circ \behim{r} (t) = 1 \quad & \iff \quad  s_T\circ r(t) \in L \\
					&\iff \quad  r(t) \in L' \\
					&\iff \quad t \in L \\
					&\iff \quad \behim{\eval}(t) = 1
			\end{split}
			\end{equation}
			
			In other words, $s$ is a universal simulator.
			Conversely, any universal simulator $s \colon A \to T$ gives rise to an NP-complete language. 
			By the above construction, the set $L_{s_T} \subset A$ is an NP language.
			Since $s$ is assumed to be universal, there exists a morphism $r \colon T \to A$ that satisfies \cref{eq:NP complete uni} and hence defines a reduction from the NP-complete $L$ to $L_{s_T}$, which proves that $L_{s_T}$ is likewise NP-complete.

			Note that similar constructions work for other complexity classes besides NP. 
		\end{example}
	
	\subsection{Universal Spin Models as Universal Simulators}\label{sec:spinmodel}
		In order to account for the universality of spin models \cite{De16b} in our framework, we explicitly distinguish between \emph{spin systems}{\,\textemdash\,}functions from a finite set of spins to energies{\,\textemdash\,}and \emph{spin models}{\,\textemdash\,}families of spin systems which are usually identified by a property they share.
		
		\parlabel{Formal spin systems}
		We take spin systems to be defined on hypergraphs, where spins sit on vertices and each $k$-body interaction is given by a hyperedge with $k$ elements.
		Specifically, for our purposes, it suffices to consider hypergraphs that are also (abstract) simplicial complexes.
		This is because we interpret a $k$-body interaction between a set of vertices as also including all $m$-body interactions among them for $m \le k$ (cf.\ \cref{def:spinsystem}). 
		A simplicial complex $G$ can be specified by a set $V_G$ of vertices and a set $E_G$ of facets, the latter of which is a subset of the power set $\pwrset(V_G)$ satisfying
		\begin{equation}\label{eq:hypergraph_quotient}
			e \subseteq e' \quad \implies \quad e = e'
		\end{equation}
		for any $e,e' \in E_G$.
	
		Provided such a simplicial complex together with a number of states of each spin, a spin system is given by its Hamiltonian viewed as a collection of local terms.
		\begin{definition}\label{def:spinsystem}
			Let $G$ be a finite simplicial complex and let $[q]$ be the set $\{0,1,2,\ldots,q - 1\}$.
			A $q$-level \textbf{spin system} on $G$ is a family
			\begin{equation}
				H = \Family*[\big]{ H_e \given e \in E_G }
			\end{equation}
			of local terms $H_e \colon [q]^e \to \R$ indexed by the facets of $G$.
		\end{definition}
		In words, $H$ associates, to each facet $e$, a function $H_e$ from the possible spin configurations of vertices within $e$ to real numbers (thought of as energies).
		A $q$-level \textbf{spin configuration} for a given hypergraph $G$ and a number of levels $q$ is given by a map 
		\begin{equation}
			\sigma \colon V_G \to [q] ,
		\end{equation}
		i.e.\ an element of $[q]^{V_G}$.
		The \textbf{energy} of a configuration $\sigma$ of a spin system $H$ is then given by
		\begin{equation}
			H(\sigma) = \sum_{e \in E_G} H_e(\sigma_{e})  
		\end{equation}
		where $\sigma_{e} \in [q]^e$ is the specification of spin values of the configuration $\sigma$ for vertices that are in the facet $e$.
		The \textbf{spectrum} $\spsp(H)$ of a spin system $H$ is the set of all its possible energies:
		\begin{equation}
			\spsp(H) = \Set*[\Big]{ H(\sigma) \given \sigma \in [q]^{V_G} } .
		\end{equation}
		For any $\Delta \in \R$ and any spectrum $\spsp$, we define the reduced spectrum as
		\begin{equation}
			\spsp(H)_{\leq \Delta} \coloneqq \Set*[\big]{\spen(H) \in \spsp \given \spen \leq \Delta } .
		\end{equation}
		
		\begin{remark}[Ground state]\label{rem:GSE}
			An important property of a spin system is its ground state energy{\,\textemdash\,}the minimal value of $H${\,\textemdash\,}and its ground state configuration(s), for which the energy is mimimized.
			To convert the task of minimizing the energy into a decision problem, one can ask whether the reduced spectrum $\spsp(H)_{\leq \Delta}$ is non-empty.
			If true, then the ground state energy is necessarily below $\Delta$.
			This is the ground state energy problem for general spin systems, whose computational complexity is NP-complete \cite{Bar82}. 
%			\togo{Add citations}
		\end{remark}
		
		\parlabel{Target-context category for spin systems}
		We now turn our attention to the target--context category for spin systems,  $\cat{SpinSys}$.
		We take $\cat{Rel}_{\mathrm{poly}}$ from \cref{sec:relpoly} as the ambient category.
		Additionally, we specify
		\begin{itemize}
			\item $T = \left\{ \left( G, q, H, \Delta \right) \right\}$ the set of all tuples such that $G$ is a finite simplicial complex, $q \in \N$, $H$ is a $q$-level spin system on $G$, and $\Delta \in \R$;
			\item $C = \left\{ (G,q,\sigma) \right\}$ the set of all tuples such that $G$ is a finite hypergraph, $q \in \N$, and $\sigma$ is a $q$-level spin configuration of $G$. 
		\end{itemize}
		Note that $T$ is the set of all spin systems according to \cref{def:spinsystem}, 
		together with a specification of an energy threshold $\Delta$, the hypergraph $G$ and the number $q$, which are  given as data to simplify the definitions. 
		Similarly, $C$ is the set of all spin configurations. 
		
		Next, we construct size measures on $T$ and $C$.
		The notion of simulation we want to capture ought to be efficient (i.e.\ polynomial) in the size of the spin systems. 
		First, for any simplicial complex $G$ and any $q \in \N$ we define the size measure as
		\begin{equation}
			\abs*[\big]{G,q} = \sum_{e \in E_G} q^{\abs{e}},
		\end{equation}
		which is the number of real parameters needed to fully specify a $q$-level spin configuration for $G$. 
		In this sense, it is a measure of complexity of spin systems. 
		We set the complexity of the empty hypergraph $G_{\emptyset}$ to $0$.
		The relevant size measure for $T$ is then the function $\abs{\ph}_T \colon T \to \mathbb{R}_{\geq 0}$ given by
		\begin{equation}
			\abs*[\big]{(G,q,H,\Delta)}_T \coloneqq \abs*[\big]{G,q}.
		\end{equation}
		Similar, the size measure for $C$ is given by 
		\begin{equation}
			\abs*[\big]{(G,q,\sigma)}_C \coloneqq \abs*[\big]{G,q}.
		\end{equation}
		
		\parlabel{Behavior structure}
		Finally, to describe a target--context category with intrinsic behaviors, we also need to provide the behavior structure. 
		We take the shadow functor to be the forgetful functor from $\cat{Rel}_{\mathrm{poly}}$ to $\cat{Rel}$, which is an example of a pointed shadow functor (\cref{sec:pointed_shadow}).
		Additionally, we define $B \coloneqq \R \times \pwrset(\R) \times \R$, where $\pwrset(\R)$ denotes the power set of $\R$.
		We interpret a behavior $(\spen, S, \Delta) \in B$ as a tuple of a set $S$ of possible energy values, a specific energy $\spen$ and a threshold energy $\Delta$. 
		We take $\abs{b}_B$ to be zero for every behavior $b$ in $B$.
		Note that then every relation of type $A \to B$ is polynomially bounded regardless of the choice of $\abs{\ph}_A$. 
		In particular, this holds true for the evaluation function
		\begin{equation}
			\eval \colon T \otimes C \to B
		\end{equation}
		that we choose to map a pair of a spin system $t \coloneqq (G,q,H,\Delta) \in T$ and a spin configuration $c \coloneqq (G',q',\sigma) \in C$ to the tuple containing the energy of the configuration $c$, the spectrum of $H$ and the threshold energy $\Delta$ via
		\begin{equation}\label{eq:spin_eval}
			\eval(t,c) \coloneqq \Bigl\{ \bigl( H(\sigma), \spsp(H), \Delta \bigr) \Bigr\} \subset B
		\end{equation}	
		as long as we have $G = G'$, $q = q'$ and $H(\sigma) \leq \Delta$.
		If, on the other hand, either the hypergraphs or the numbers of spin states differ or if the energy $H(\sigma)$ is above the threshold $\Delta$, then we leave $\eval$ undefined.

		Finally, the behavioral relation $\brel$ is defined by requiring that ${(\spen', S', \Delta') \brelop (\spen , S, \Delta)}$ holds precisely if we have
		\begin{equation}\label{eq:spin_brel_spec}
			S'_{\leq \Delta'} = S^{\phantom{\prime}}_{\leq \Delta'}, \qquad \text{and} \qquad \Delta' \geq \Delta  \qquad \text{and} \qquad \spen' = \spen \leq \Delta .
		\end{equation}

		This definition ensures that if $t'=(G',q',H',\Delta')$ lax context-reduces to $t=(G,q,H,\Delta)$ according to \cref{def:c_red_rel}, i.e.\ if there exists a morphism $v \colon C \to C$ satisfying 
		\begin{equation}
			t \otimes v \mrelop t \otimes \id_C,
		\end{equation} 
		then 
		\begin{itemize}
			\item the spectra of the two Hamiltonians agree up to the threshold $\Delta$ of the target, and
			\item for any target configuration $\sigma$ and for any corresponding configuration $\sigma'$ that is contained in $v(G,q,\sigma)$, either we have $H'(\sigma')=H(\sigma)$ or both energies are above the target \mbox{threshold $\Delta$}. 
		\end{itemize}
		
		The above relation between two spin systems $t$ and $t'$ is meant to specify that the source spin system $t$ \emph{simulates} the target spin system $t'$ in the language of \cite{De16b}.
		This interpretation is relevant because it justifies our choice of the behavioral relation.
		See \cref{rem:spin_sim} for an intuion for spin system simulations and \cref{rem:encompass} for a similar interpretation in the case of Turing machines and their simulation.

		This concludes the description of the target--context category chosen to describe universality of spin models in this article. 
		Next, we define a universal simulator therein, which uses 2D Ising systems to simulate all other spin systems.
		Its universality follows by the results of \cite{De16b}.
		
		\parlabel{2D Ising model as a universal simulator}
		We choose a subset $P$ of $T$ to be the 2D Ising model with fields.
		That is, $P$ is the set of tuples $(G,2,H,\Delta)$ for an arbitrary $\Delta\in \R$ and $G$ is a 2D rectangular lattice with nearest neighbor facets of size $2$. 
		The Hamiltonian $H$ assigns an Ising coupling $J_{ij}$ to each such facet and specifies a magnetic field value $B_i$ at each vertex, so that for a facet $e = \{ i,j\} $ we have
		\begin{equation}
			H_e \left(  \sigma_{e} \right) = J_{ij} (-1)^{\sigma_{i}} (-1)^{\sigma_{j}} + B_i (-1)^{\sigma_{i}} + B_j (-1)^{\sigma_{j}}.
		\end{equation}
		Additionally, we take the size measure $\abs{\ph}_P$ for $P$ to be the one inherited from $\abs{\ph}_T$.
		
		We now define a universal simulator of type $s \colon P \otimes C \to T \otimes C$.
		The compiler $s_T \colon P \hookrightarrow T$ is the inclusion of $P$ in $T$.
		The reduction $r \colon T\to P$ leaves $\Delta$ unchanged and maps a given arbitrary spin system $t$ to the 2D Ising system $r(t)$ that can be used to simulate it with threshold $\Delta$ according to the construction given in \cite{De16b,Re23}. 
		Furthermore, the context reduction $s_C$ is the relation that maps a configuration $\sigma$ of $t$ to the set of all configurations $\sigma'$ of $r(t)$ that have the `same behavior'.
		That is, the set $s_C(r(t),c)$ is non-empty and for any $c' \in s_C(r(t),c)$ we have
		\begin{equation}\label{eq:spectrum_equality}
			\spsp(H)_{\leq \Delta} = \spsp(H')_{\leq \Delta}
		\end{equation}
		and either 
		\begin{equation}\label{eq:energy_matching}
			H(\sigma) = H'(\sigma') \qquad \text{or} \qquad H(\sigma) > \Delta, \; H'(\sigma')> \Delta ,
		\end{equation}
		where $H$ and $\sigma$ are the target Hamiltonian and configuration and $H'$ and $\sigma'$ are the source (2D Ising) Hamiltonian and its configuration respectively.
		It is shown in \cite{De16b} that such  maps $r$ and $s_C$ exist and that they both are polynomially bounded. 
		This implies that the simulator $s$ made up of $s_T$ and $s_C$ is universal in the sense of \cref{def:univ sim}.

		\begin{remark}[Intuition for spin system simulations]\label{rem:spin_sim}
			To obtain some intuition for spin system simulations, let us assume that a (2D Ising) spin system $r(t)$ simulates a spin system $t$ and extract the consequences for their respective ground state energy problems (\cref{rem:GSE}).
			Equality \eqref{eq:spectrum_equality} of the reduced spectra means that their ground state energies agree.
			Therefore, a solution to the ground state energy problem for $r(t)$ also solves this problem for $t$.
			Moreover, every ground state configuration of $t$  (i.e.\ a witness for its ground state energy problem) is associated to a collection of ground state configurations of $t$.
			
			This analysis also exposes the relevance of the use of polynomially bounded relations as morphisms in the ambient category of $\cat{SpinSys}$.
			In particular, the fact that the reduction $r$ is polynomially bounded is necessary to obtain a polynomial-time reduction (in the sense of computational complexity) from the ground state energy problem of a generic spin system $t$ to the ground state energy problem of a 2D Ising spin system $r(t)$.
		\end{remark}
		
		\begin{remark}[Multi-valued context reduction]\label{rem:multi-valued_context_reduction}
			The context reduction $s_C$ of the 2D Ising simulator maps a configuration $\sigma$ not just to a single configuration but to all configurations $\sigma'$ of the source model with the same energy as $\sigma$.
			There are two benefits of this choice. On the one hand, inverting $s_C(r(t),\ph )$ yields a relation mapping configurations of the source model to the corresponding configurations of the target model. 
			In particular, one can find the ground states of the target model by finding a ground state of the source model and applying the inverse context reduction.

			On the other hand, having access to all configurations of the source model with a given energy enables the investigation of the partition function. 
			If the multiplicities of the source model (below the threshold $\Delta$) and the target model differ only by a constant factor, then the partition functions also only differ by a constant factor. 
		\end{remark}
		
		\begin{remark}[Universal simulators need not be universal spin models]
			Note that while universal spin models (in the sense of \cite{De16b}) give rise to universal simulators in $\cat{SpinSys}$, the converse is not true in general.
			The reason for this is an additional requirement that \cite{De16b} places on context reductions. 
			While in $\cat{SpinSys}$ context reductions ${s_C \colon P \otimes C \to C}$ are arbitrary relations that satisfy the efficiency constraint of \cref{eq:poly constr}, in \cite{De16b} it is additionally required that configurations $\sigma'$ are mapped to configurations $\sigma$ such that one can recover $\sigma'$ from $\sigma$ in a ``local'' way.
			More precisely, one assigns a set of spins $P_i \in V_{G}$ to each spin $i \in V_{G'}$. 
			Knowledge of $\sigma\vert_{P_i}$ then allows one to recover $\sigma'\vert_{i}$.
			We expect that it is possible to account for this by further restricting the morphisms $\cat{SpinSys}$ to respect the hypergraph structure in an appropriate way and thereby also capture simulation of spin models from \cite{De16b} via the lax context-reduction relation (\cref{def:c_red_rel}). 
		\end{remark}	
	
	\subsection{Abstract Examples of Universal Simulators}\label{sec:math_examples}
	
	We now present a few abstract examples of universal simulators. 
	In \cref{ex:cofinal,ex:dense}, the ambient category $\cC$ is $\cat{Set}${\,\textemdash\,}the category of sets and functions.
	In \cref{ex:universal_set}, the ambient category is a subcategory of $\cat{Set}$ called $\cat{Qbs}$.
	All examples use a target--context category with intrinsic behaviors (\cref{def:intrinsic_beh}) whose shadow functor is the inclusion into $\cat{Rel}$.
	By \cref{prop:inclusion_is_pointed}, it coincides with the pointed shadow functor in the first two examples. 
	Note that whenever $\Beh$ is an inclusion functor we may consider $\brel$ as a preorder on $B$ instead of $\behim{B}$.
	The specific target--context categories used in these examples are summarized in \cref{table:abstract_examples}.

	\begin{table}[!ht]
		\centering
		\begin{tabular}{c|c|c|c|c|c} 
%			 \hline
				 &  $\cC$ &  $T$ &  $C$ &  $\bigl( B, \brel \bigr)$ & \makecell{$\eval \colon T \times C \to B$} \\ 
			 \hline
			 	\makecell{\hyperref[ex:cofinal]{cofinal subset}} & $\cat{Set}$ & $X$  & $I$  & $(X, \succeq)$ & $\eval(x)=x$ 
			\\ 
			 	\hyperref[ex:dense]{dense subset} 
			 	& $\cat{Set}$ 
			 	& \makecell{$\mathbb{R}\times \mathbb{R}_+$} 
			 	& $I$ 
			 	& \makecell{$\bigl( \pwrset(\mathbb{R}), \subseteq \bigr)$} 
			 	& \makecell{$\eval(x,\delta) = \mathcal{B}(x, \delta)$} 
			 \\
			 	\makecell{\hyperref[ex:universal_set]{universal Borel set}} 
			 	& \makecell{$\cat{Qbs}$}
			 	& $[2]^X$ 
			 	& \makecell{$X$} 
			 	& $\bigl( [2], = \bigr)$ 
			 	& $\eval(t,x) = \chi_x(t)$ \\ 
%			 \hline
		\end{tabular}
		\caption{Data used to specify target--context categories from examples in \cref{sec:math_examples}.}
		\label{table:abstract_examples}
	\end{table}
	
	\begin{example}[Cofinal subset as a universal simulator]\label{ex:cofinal}
		Consider a preordered set $(X, \succeq)$.
		A subset $M$ of $X$ is termed cofinal if there exists an enhancement map $e \colon X \to M$.
		We now show how a cofinal subset gives rise to a universal simulator.
		To that end, consider $\cat{Set}$ as ambient category, let $T \coloneqq X$, $C \coloneqq I$ and let the preordered set of behaviors be given by $(X,\succeq)$.
		Further, let $\eval$ be the identity on $X$. 
		Next, consider the simulator $s \colon M \to X$ with $s_T$ being the inclusion of $M$ into $X$.
		Then, for any $x\in X$, we have $s\circ e(x) = e(x) \succeq x = \id(x)$ because $e$ is an enhancement map. 
		Since both $s \circ e$ and $\id$ are functions, the following are equivalent (cf.\ \cref{ex:imitation_functions}):
		\begin{enumerate}
			\item \label{it:preorder} $e(x) \succeq x$.
			\item \label{it:enhorder} There exists an enhancement map $\{x\} \to \{e(x)\}$.
			\item \label{it:degorder} There exists a degradation map $\{e(x)\} \to \{x\}$.
		\end{enumerate}
		Note that $s$ is a universal simulator, meaning that 
		\begin{equation}
			\input{cofinal_univ.tikz}
		\end{equation}
		holds for some function $e$, if and only if there is an $e$ such that, for all $x \in X$, both \ref{it:enhorder} and \ref{it:degorder} are satisfied.
		On the other hand, $M$ is a cofinal subset of $X$ if and only if there is an $e$ such that condition \ref{it:preorder} holds for all $x \in X$.
		As a consequence of \cref{ex:imitation_functions} then, $s$ given by the inclusion function as above is a universal simulator if and only if its image is a cofinal subset of $X$.
		The reduction is given precisely by the enhancement map $X \to M$.

		More generally, in \cref{sec:no-go_theorems} we show that the so-called functional image (\cref{eq:fun_im}) of a simulator's compiler is always a cofinal subset of functional states on $T$ with respect to the lax context-reduction preorder (\cref{def:c_red_rel}). 
	\end{example}	

	\begin{example}[Dense subset as a universal simulator]\label{ex:dense}
		We show how the notion of a dense subset can be viewed as the universality of certain simulators. 
		While we describe rational numbers as a dense subset of real numbers, the example readily extends to many other topological spaces.  
		Consider $\cat{Set}$ as the ambient category and let $T \coloneqq \mathbb{R} \times \mathbb{R}_{+}$, $P \coloneqq \mathbb{Q} \times  \mathbb{R}_{+}$. 
		Further let $C$ be a trivial (i.e.\ singleton) set of contexts.
		The behaviors are sets of real numbers, so that $B$ is the power set of $\mathbb{R}$.
		The behavioral relation $\brel$ is given by the set inclusion $\subseteq$ which is interpreted as an ordering describing precision.
		Finally, let $\eval \colon T \to B$ be given by
		\begin{equation}
			(x,\delta) \mapsto \mathcal{B}(x, \delta)
		\end{equation}
		where $\mathcal{B}(x, \delta)$ is the open ball around a real number $x$ with radius $\delta$. 
		If $s \colon P \to T$ is the simulator given by the inclusion of $\mathbb{Q}$ into $\mathbb{R}$,
		\begin{equation}
			s \left( q , \delta \right) \coloneqq \left( q , \delta \right),
		\end{equation}
		then its universality expresses the fact that $\mathbb{Q}$ is dense in $\mathbb{R}$.
		The reduction $r^*$ to the trivial simulator can be specified by choosing, for any real number $x$ and any radius $\delta$, some rational number $q_{(x,\delta)}$ within the ball $\mathcal{B} \left( x, \delta/2 \right)$ and setting
		\begin{equation}
			r(x,\delta) \coloneqq \left( q_{(x,\delta)}, \delta/2 \right),
		\end{equation}
		so that it is indeed the case that $s \circ r$ behaviorally imitates $\id_T$ by \cref{ex:imitation_functions}.
	\end{example}

	\begin{example}[Universal Borel set as a universal simulator]\label{ex:universal_set}
		In descriptive set theory, there is a notion of a universal set \cite[definition 22.2]{kechris2012classical}.
		Such universal sets can be identified by the universality of a certain kind of simulators.
		To make the example concrete, let $\cC$ be the category $\cat{Qbs}$ of quasi-Borel spaces \cite{heunen2017convenient}, which has the useful property that it is cartesian closed.
		That is, we have (natural) bijections
		\begin{equation}\label{eq:ccc}
			\cat{Qbs}(Y \times X , Z) \cong \cat{Qbs} \left( Y , Z^X \right).
		\end{equation}
		Furthermore, consider standard Borel spaces $X$ and $Y$.
		Let the set $T$ of targets be the exponential object $[2]^X${\,\textemdash\,}the quasi-Borel space of subsets of $X$ associated with the so-called `Borel-on-Borel' $\sigma$-algebra \cite[section 27.1]{stein2021structural}.
		Its image under the shadow functor is $\brlsets(X)${\,\textemdash\,}the set of Borel subsets of $X$.
		The space of contexts $C$ is $X$ itself and behaviors $B$ are $[2] \coloneqq \{0,1\}$ with $\brel$ the equality relation.
		The evaluation function $\behim{\eval} \colon \brlsets(X) \times X \to [2]$ is given by $(t,x) \mapsto \chi_x(t)$ where $\chi_x$ is the characteristic function of $x$ as in \cref{ex:NP_complete}.
		That is, it gives $1$ whenever $x$ is an element of $t$ and $0$ otherwise.

		Let $s$ be a simulator where $P = Y$ and $s_C = \discard_Y {}\otimes{} \id_X$.
		Via the natural isomorphism \eqref{eq:ccc}, we can associate a morphism $Y \times X \to [2]$ (i.e.\ a subset of $Y \times X$) to any $s_T \colon Y \to [2]^X$.
		Then $s$ is a universal simulator if and only if this subset of $Y \times X$ is a $Y$-universal set for $[2]^X$.
		In other words, this condition says that every Borel subset $t$ of $X$ (i.e.\ an element of $[2]^X$) can be realized as the intersection of the $Y$-universal one with the set $\{y\} \times X$ for some $y \in Y$.
		The appropriate $y$ for each $t$ defines the map $r \colon [2]^X \to Y$ via $r(t) \coloneqq y$, so that $r^*$ is a lax reduction from $s$ to the trivial simulator.
		While we have not been able to show that such an $r$ can be in fact chosen so that it is a morphism in $\cat{Qbs}$ yet, we believe that this is the case.
	\end{example}

%% ===========================================
	\subsection{No-Go Theorem for Universality}\label{sec:no-go_theorems}
		
		We now derive conditions which guarantee that a simulator is \emph{not} universal{\,\textemdash\,}a no-go theorem for universality (\cref{thm:nogo:cor}).
		We apply them to show that in the target--context category of spin systems (\cref{sec:spinmodel}) every universal simulator needs to contain an infinite number of Hamiltonians in the functional image of its compiler.

		\begin{definition}\label{def:c_red_rel}
			Consider a target--context category $\cC$ and two morphisms $f,g \colon A \to T$.
			We say that $f$ \textbf{lax context-reduces} to $g$, denoted $f \credop g$, if there exists a morphism $f_C \colon A \otimes C \to C$ such that
			\begin{equation}\label{eq:c_red_rel}
				\input{c_red_relation.tikz}
			\end{equation}
			holds, in which case we call $f_C$ a context reduction.
			If, instead, relation \eqref{eq:c_red_rel} holds with the opposite of $\mrel$, then we say that $f$ \textbf{oplax context-reduces} to $g$, denoted $f \oplaxcredop g$.
		\end{definition}
		\begin{remark}[Encompassing as context-reducing]\label{rem:encompass}
			The idea behind \cref{def:c_red_rel} is much like that of the encompassing relation from the beginning of \cref{sec:motivation}.
			In particular, we can see that a given morphism $u \colon I \to T$ gives rise to a compiler $s_T = u \circ \discard_T$ of some (singleton) universal simulator if and only if $s_T$ lax context-reduces to the identity on $T$.
			This property then also implies that $u$ itself lax context-reduces to \emph{every} state $t \colon I \to T$, and in this sense $u$ is all-encompassing.
			
			For instance, in the target--context category $\cat{Tur}$ from \cref{ex:TM}, consider two Turing machines represented by deterministic states $t_1,t_2 \colon I \to T$.
			Then $t_1$ lax context-reduces to $t_2$ if and only if $t_1$ simulates $t_2$ in the standard sense as described in the part on ``additional structure'' in \cref{sec:motivation}.
		\end{remark}
		\begin{proposition}
			The lax context-reduction relation is a preorder on every hom-set $\cC(A,T)$.
		\end{proposition}
		\begin{proof}
			Transitivity of $\cred$ follows since the ambient relation $\mrel$ is transitive and preserved by precomposition. 
			In particular, if $f \credop g$ and $g \credop h$ holds, then we have $f \credop h$ with context reduction given by 
			\begin{equation}
				\input{c_red_trans.tikz}
			\end{equation}
			Reflexivity of $\cred$ follows by choosing $\discard_A \otimes \id_C$ to be the context reduction.
		\end{proof}
		
		Clearly, if $s$ is a universal simulator with compiler $s_T$ and lax reduction $r^*$, then the composite $s_T \circ r$ lax context-reduces to $\id_T$.
		In particular, for any functional $t\in \cC_{\rm fun}(I,T)$, we find
		\begin{equation}\label{eq:sim_rel_new}
			\input{sim_relation_new.tikz}
		\end{equation}
		since the ambient imitation relation is preserved by precomposition.
		Intuitively, a simulator is universal if it encompasses all targets. 
		We can understand relation \eqref{eq:sim_rel_new} as saying that for any target $t$, there ought to be one ($s_T \circ r \circ t$) in the image of the compiler that context-reduces to $t$.
		It follows that the set of targets that a universal simulator has access to (in the sense that they are in the image of $s_T$) must be cofinal with respect to the lax context-reduction preorder.
		
		More concretely, let us define the \textbf{functional image} of a morphism $f \in \cC_{\rm fun}(A,T)$ by
		\begin{equation}\label{eq:fun_im}
			\mathrm{Im}_{\rm fun}(f) \coloneqq \Set*[\big]{f \circ a \given a \in \cC_{\rm fun}(I,A) }.
		\end{equation}
		It is a subset of $\cC_{\rm fun}(I,T)${\,\textemdash\,}the set of all functional states on $T$.
		\Cref{eq:sim_rel_new} then says that there is an enhancement map
		\begin{equation}\label{eq:fun_enh}
			\mathrm{enh} \colon \cC_{\rm fun}(I,T) \to \mathrm{Im}_{\rm fun}(s_T) , \qquad  t \mapsto s_T \circ (r \circ t)
		\end{equation}
		with respect to the lax context-reduction preorder since every $r \circ t$ is itself a functional state.
		
		Notably, existence of the enhancement map is a statement that is independent of the choice of $r$.
		We can now use the construction from \cite[corollary 2.32]{Gonda2021} to arrive at the following no-go result for universal simulators.
		
		\begin{theorem}[No-go theorem for universality]\label{thm:nogo:cor}
			Let $\varphi \colon \cC_{\rm fun}(I,T) \to \mathbb{R}$ be monotone with respect to the lax context-reduction preorder $\cred$ on its domain.
			Let $s \colon P \otimes C \to T \otimes C$ be a simulator and let $\mathrm{Im}_{\rm fun}(s_T)$ be the functional image of its compiler.
			If we have
			\begin{equation}\label{eq:nogo:cor}
				\sup \varphi \bigl( {\mathrm{Im}_{\rm fun}(s_T)} \bigr) < \sup \varphi \bigl( {\mathrm{Im}_{\rm fun}(\id_T)} \bigr), 
			\end{equation}
			then $s$ is not a universal simulator.
		\end{theorem}
		Note that $\cC_{\rm fun}(I,T)$ is the functional image of the compiler for the trivial simulator, i.e.\ we have
		\begin{equation}\label{eq:trivial_image}
			\cC_{\rm fun}(I,T) = \mathrm{Im}_{\rm fun} (\id_T).
		\end{equation}
		\begin{proof}
			We prove the contrapositive. 
			To that end, assume that $s$ is universal. 
			Since the function $\mathrm{enh}$ from \eqref{eq:fun_enh} is an enhancement map and $\varphi$ is monotone, we have
			\begin{equation}
				\varphi \bigl( \mathrm{enh}(t) \bigr) \geq \varphi(t)
			\end{equation}
			for all $t \in \cC_{\rm fun}(I,T)$.
			This immediately implies the negation of inequality \eqref{eq:nogo:cor}.
		\end{proof}
		
		We now use this result to show that in the target--context category of spin systems there is no universal simulator for which the functional image of its compiler contains only a finite number of Hamiltonians.

		\begin{example}[Universal spin models have infinitely many Hamiltonians]\label{ex:nogo_spin}
			Let $T$ be the set of spin systems as in \cref{sec:spinmodel}.
			We define a function $\varphi \colon \cC_{\rm fun}(I,T) \to \mathbb{R}$ that maps a spin system to the cardinality of its spectrum below the energy threshold $\Delta$:
			\begin{equation}
				\varphi \bigl( (G,q,H,\Delta) \bigr) \coloneqq \abs{\spsp(H)_{\leq \Delta}}.
			\end{equation}
			Note that $\varphi$ is monotone with respect to the context-reduction preorder $\cred$. 
			This follows by the definition of the evaluation $\behim{\eval}$ and the behavioral relation $\brel$ in \cref{sec:spinmodel}, more precisely by \cref{eq:spin_eval,eq:spin_brel_spec}.
			
			Now consider a simulator in $\cat{SpinSys}$, for which the set of possible hamiltonians in the functional image of its compiler is finite.
			Then $\varphi$ necessarily attains a maximal value over $\mathrm{Im}_{\rm fun}(s_T)$, so that 
			\begin{equation}
				\sup \varphi \bigl( {\mathrm{Im}_{\rm fun}(s_T)} \bigr)
			\end{equation}
			is finite.
			
			However, there are spin systems with arbitrary large spectrum.
			To see this, consider the simplicial complex with vertices $V_n = \{0, 1, \ldots, n-1\}$ and facets $E_n = \bigl\{ \{0\}, \{1\}, \ldots, \{n-1\} \bigr\}$.
			For 2 local spin values (i.e.\ $q=2$) define each local term $H_i \colon [2] \to \mathbb{R}$ of the Hamiltonian by $H_i(0)=0$ and $H_i(1)=1$. 
			The resulting spin system has spectrum $\{0,1,\ldots,n\}$ with cardinality $n+1$.
			The right-hand side of \eqref{eq:nogo:cor} is thus infinite and the strict inequality is satisfied.
			Hence, a simulator whose functional image contains finitely many Hamiltonians cannot be universal. 
		\end{example}
		
\section{Comparing Simulators}
\label{sec:Morphisms}
	
	In the previous sections, we have introduced the mathematical framework and the concept of universality within as a property of simulators.
	We now turn to the exploration of relations between simulators, and particularly the universal ones. 
	For instance, the singleton simulator from \cref{ex:TM univ sim} and the trivial simulator are both universal in the ambient category for Turing machines (\cref{ex:TM}).  
	However, the former is clearly more interesting, as it only uses a single Turing machine to simulate the behavior of all others.
	Contrast this with the trivial simulator that `needs access' to all Turing machines, i.e.\ the functional image of its compiler is the whole of $T$.
	In this sense, the singleton simulator is less wasteful (or more parsimonious) than the trivial simulator.

	In this section, we present tools to compare (universal) simulators in general. 
	It should be noted, however, that these only allow us to relate simulators within the same ambient category.
	Relations between ambient categories are explored further in \cref{sec:Functors}.
	First, we introduce processings of simulators (\cref{sec:processings}), which are needed to define morphisms of simulators (\cref{def:simulator morphism}). 
	One of the key properties of simulator morphisms is that they cannot generate universality (\cref{prop:morph_laxred}). 
	They also give rise to the simulator category (\cref{sec:simulator_category}), which has simulators as objects. 
	We subsequently interpret the ordering of universal simulators coming from the existence of a morphism between them as telling us about their parsimony (\cref{sec:parsimony}).
	In particular, we find sufficient (\cref{thm:morph_stronger}) and necessary (\cref{thm:s 2 id}) conditions for the existence of a simulator morphism.
	These results are used to show that, in line with the intuition from the previous paragraph, the singleton simulator for Turing machines is strictly more parsimonious than the trivial simulator (\cref{ex:trivial2universal,ex:universal2trivial}).

	\subsection{Processings of Simulators}
	\label{sec:processings}
	By a simulator processing we mean a process that rearranges target--context pairs, possibly depending on the program, as follows:
	\begin{equation}\label{eq:processing}
		\input{processing.tikz}
	\end{equation}
	Such processings do not, in general, map a simulator to another simulator.
	To ensure this is the case, we require condition \ref{it:split processing} in \cref{def:simulator processing}.
	
	However, even then we have to impose additional constraints on $q$ so that it doesn't introduce universality. 
	Condition \ref{it:processing_weak} below is used to make sure that processing a universal simulator yields another universal simulator.
	One can motivate this desideratum on resource-theoretic grounds.
	If a processing $q$ maps simulator $s$ to simulator $s'$ and we interpret $q$ as a free operation in a resource theory \cite{Co16c}, then $s'$ should be a simulator that is `no more universal' than $s$ is.
	Specifically, if the target simulator $s'$ is already universal, so should be the source simulator $s$.
	As shown in \cref{prop:morph_laxred} below, property \ref{it:processing_weak} is indeed sufficient to ensure that processings do not generate universality.
	Additionally, we demand that processings satisfy a splitting property \ref{it:split processing} just like simulators do, which ensures that applying such a $q$ to any simulator yields another simulator.
	\begin{definition}\label{def:simulator processing}
		Consider two simulators $s, s'$, both with programs given by $P$.
		A \textbf{processing} $q_* \colon s \to s'$ of simulators is specified by a map $q \colon P \otimes T \otimes C \to T \otimes C$ such that
		\begin{enumerate}
			\item \label{it:split processing} $q$ splits via 
				\begin{equation}\label{eq:processing_split}
					\input{processing_split.tikz}
				\end{equation}
				where $q_T$ is functional, and such that it satisfies
			 	\begin{equation}\label{eq:processing_dom}
					\input{processing_dom.tikz}
				\end{equation}
				
	    		\item \label{it:processing_weak} $q$ satisfies
				\begin{equation}\label{eq:processing_weak}
					\input{impl_range_reflect_3.tikz}
				\end{equation}
				
			\item \label{it:reachability} $s'$ is equal to the ($P$-correlated) sequential composite of $q$ and $s$, i.e.\ we have
				\begin{equation}\label{eq:processing_hit}
					\input{processing_hit.tikz}
				\end{equation}
		\end{enumerate}
	\end{definition}
	Note that condition \ref{it:split processing} says precisely that $q$ is itself a simulator with programs $P \otimes T$, compiler $q_T$, and a context reduction $q_C$.
	By virtue of the splittings of simulators and the functionality of the compiler $s_T$, condition \ref{it:reachability} thus says
	\begin{equation}\label{eq:processing_hit_2}
		\input{processing_hit_2.tikz}
	\end{equation}
	In particular, as soon as $q$ satisfies conditions \ref{it:split processing} and \ref{it:processing_weak}, it defines a processing of type $s \to s'$ for an arbitrary $s$ and for $s'$ defined via equations \eqref{eq:processing_hit_2}.

	\begin{example}[Processings for Turing machines]\label{ex:TM_ex_processings}
		In the ambient category of Turing machines $\cat{Tur}$ (\cref{ex:TM} and \cref{ex:TM_behaviors}), an example of a processing can be obtained as follows: 
		Fix the programs to be $P \coloneqq \Sigma^\star${\,\textemdash\,}the set of all finite strings.
		Consider two computable families $(f_p)_{p \in P}$ and $(g_p)_{p \in P}$ of partial computable functions of type $P \to P$ such that for each $p\in P$ we have ${g_p \circ f_p = \id_P}$.
		Since $f_p$, $g_p$ are both computable, there exist corresponding Turing machines $t_{f_p}$ and $t_{g_p}$ that compute them.
		Let $f, g \colon P \to T$ be the two computable functions that assign, to any program $p$, the corresponding Turing machine $t_{f_p}$ and $t_{g_p}$, respectively.
		We assume that the two families, $(f_p)_{p \in P}$ and $(g_p)_{p \in P}$, are such that both $f$ and $g$ are partial computable and thus morphisms of $\cat{Tur}$.
		
		In total, this defines a processing $q = (q_T,q_C)$ with
		\begin{align}
			q_T(p,t) &\coloneqq t \circ t_{g_p}  &  q_C(p,t,c) &\coloneqq f_p (c)
		\end{align}
		where $t \circ t_{g_p}$ refers to the composition of Turing machine $t$ with Turing machine $t_{g_p}$, i.e.\ given any input, first running $t_{g_p}$ and then running $t$ on the corresponding output of $t_{g_P}$.  
		To show that $q$ satisfies relation \eqref{eq:processing_weak}, and thus is indeed a processing, note that we have
		\begin{equation}
			\eval\circ q (p,t,c) = \eval \bigl( t\circ t_{g_p}, f_p(c) \bigr) = \eval(t,c)
		\end{equation}
		where the second equality uses the fact that $g_p$ is a left inverse of $f_p$. 

		Applying this processing to the simulator $s_u$ of \cref{ex:TM univ sim}, corresponding to a universal Turing machine $u$, the resulting compiler $s'_T$ maps every program $p$ to the Turing machine $u\circ t_{g_p}$. 
		The context reduction $s'_C $ maps a context $c$ to $f_p(\langle p,c\rangle)$. 
		In words, this processing scrambles the inputs by $f_p$ and at the same time modifies the Turing machine such that it unscrambles the inputs via $g_p$ before running the universal Turing machine.
	\end{example}	

	\begin{example}[Processings for spin models]\label{ex:SM_ex_processings}
		One example of a processing in the ambient category of spin systems (\cref{sec:spinmodel}) is the following. 
		Consider a family of bijections ${\pi_G\colon V_G \to V_G}$, one for each simplicial complex $G$.
		We can then specify a particular processing for an arbitrary object $P$ as follows.
		We let $q_C$ be the parallel composition of $\discard_P \otimes \discard_T$ with a function ${C \to C}$ that reshuffles the spin values of a configuration $\sigma$ by applying $\pi_G$ and leaves the rest of the data unchanged.
		On the other hand, we let $q_T$ be the parallel composition of $\discard_P$ with a function ${T \to T}$ that reschuffles the spin couplings (i.e.\ the Hamiltonian $H$) by applying the inverse of $\pi_G$.
		This defines a processing, because the only data that changes is invisible under evaluation:
		\begin{equation}
			H(\sigma) = H (\sigma \circ \pi_G^{-1} \circ \pi_G).
		\end{equation}
		Note that both $q_T$ and $q_C$ are clearly polynomially bounded and thus define morphisms in $\cat{Rel}_{\mathrm{poly}}$.
	\end{example}
	
		\begin{example}[Processings for abstract examples]\label{ex:math_ex_processings}
		In the target--context category with cofinal subsets (\cref{ex:cofinal}), there is only one processing for every $P$, namely $\discard_P \times \id_T$.
		For the target--context category of dense subsets (\cref{ex:dense}), processings are more interesting.
		There, a function $q \colon P \times T \to T$ for $T = \mathbb{R} \times \mathbb{R}_+$ satisfies relation \eqref{eq:processing_weak} if and only if for every element $p \in P$ and every $(x,\delta) \in T$, we have
		\begin{equation}
			\mathcal{B}(x,\delta) \subseteq \mathcal{B} \bigl( q(p,x,\delta) \bigr).
		\end{equation}
		In words, $q$ may introduce a loss of precision, but any point $\delta$-close to $x$ must still be within the ball defined by its image $q(p,x,\delta)$.
		Note that equation \eqref{eq:processing_split} is vacuous since there is only one context in this case.
\end{example}

%%--------------------------------------------------------
	\subsection{The Simulator Category}
	\label{sec:simulator_category}

		Having introduced processings, we are now ready to define morphisms of simulators and the simulator category (\cref{fig:simulator-category}).
		\begin{figure}[t]\centering
			\includegraphics[width=.65\columnwidth]{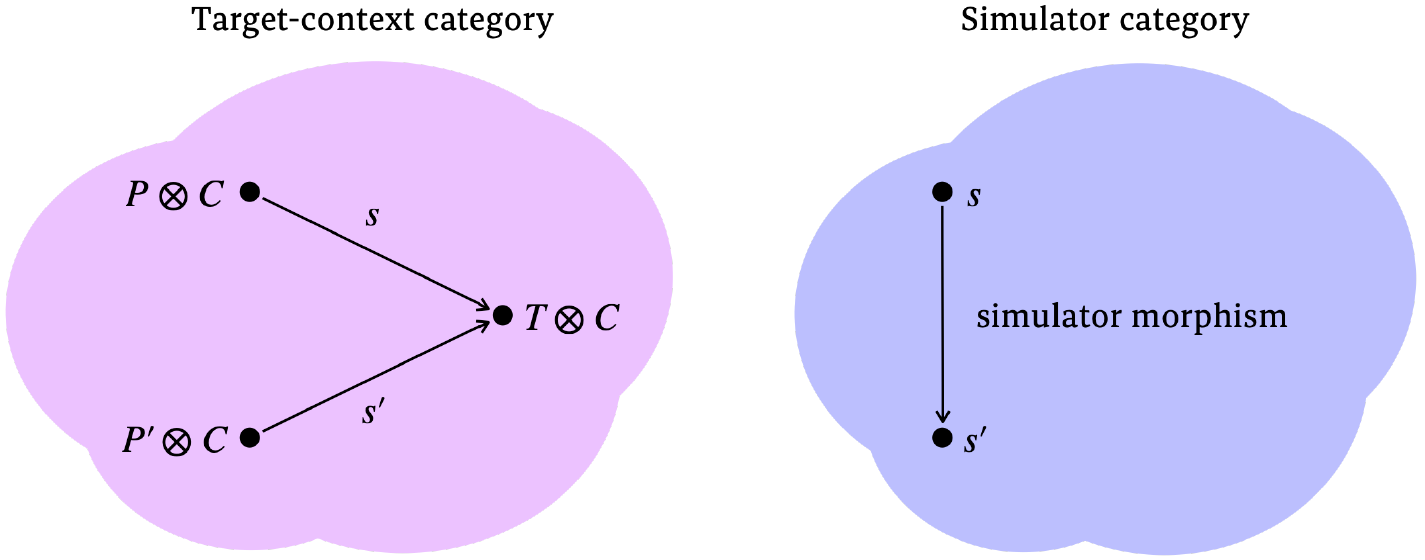}
			\caption{Given a target--context category, we construct the associated simulator category (\cref{def:simulator_category}).
				It consists of simulators related by simulator morphisms (\cref{def:simulator morphism}). 
				If there is a simulator morphism $s \to s'$, we say that $s'$ is more parsimonious than $s$. 
			}
			\label{fig:simulator-category}
		\end{figure}

	\begin{definition}\label{def:simulator morphism}
		Let $s$ and $s'$ be two simulators.
		A \textbf{simulator morphism} of type $s \to s'$ consists of a pair $(r^*,q_*)$ of a reduction $r^* \colon s \to r^* s$ and a processing $q_* \colon r^* s \to s'$.
	\end{definition}
	In the ambient category, a simulator morphism thus acts as
	\begin{equation}\label{eq:morphism}
		\input{morphism.tikz}
	\end{equation}
	where the right-hand side is $s'$. 
	The next result shows that whenever a morphism from a simulator $s$ to a simulator $s'$ exists, then there is also a lax reduction of type $s \to s'$.
	\begin{proposition}\label{prop:morph_laxred}
		Let $s$ and $s'$ be simulators and $(r^*,q_*) \colon s \to s'$ a morphism.
		Then $r^*$ is a lax reduction of type $s \to s'$.
	\end{proposition}
	\begin{proof}
		By relation \eqref{eq:processing_weak} and \cref{lem:brel_precomp}, we have 
		\begin{equation}\label{eq:morph_laxred}
			\input{morph_laxred.tikz}
		\end{equation}
		and by \cref{eq:processing_hit}, the right hand side equals $s'$ so that $r^*$ is indeed a lax reduction.
	\end{proof}
	As a corollary, since the sequential composition of lax reductions is a lax reduction, if there is a morphism from a simulator $s$ to a simulator $s'$, then we have
	\begin{equation*}
		s' \text{ is universal} \quad \implies \quad s \text{ is universal}.
	\end{equation*}
	In other words, simulator morphisms cannot generate universality.
	
	If the reduction part of a morphism is split monic, i.e.\ if $r$ has a left inverse $\nu$, then the following string diagram equation\footnotemark{}
	\footnotetext{We understand this equation as saying that the corresponding equation holds whenever the hole is filled with a simulator.}%
	\begin{equation}\label{eq:morphism v2}
		\input{morphism_v2.tikz}
	\end{equation}
	holds, so that we can describe the morphism also as a processing{\,\textemdash\,}given by $(q \circ \nu)_*${\,\textemdash\,}followed by reduction $r^*$ instead of the other way round as in \cref{def:simulator morphism} and in the left-hand side of \cref{eq:morphism v2}.
	
	We now show how morphisms can be composed to give a category of simulators.
	\begin{proposition}[Sequential composition of simulator morphisms]\label{prop:morphism_composition}
		Consider two morphisms of simulators ${(r_1^*,{q_1}_*) \colon s \to s_1}$ and ${(r_2^*,{q_2}_*) \colon s_1 \to s_2}$.
		The composition of the respective reductions, given by $(r_1 \circ r_2)^*$, is a reduction of type $s \to (r_1 \circ r_2)^* s$ and the map
		\begin{equation}\label{eq:processing_composition}
			\input{composition_processing.tikz}
		\end{equation}
		is a processing of type $(r_1 \circ r_2)^* s \to s_2$.
		Together, they define the \textbf{sequential composition}
		\begin{equation}
			(r_2^*,{q_2}_* ) \circ (r_1^*,{q_1}_*) \colon s \to s_2
		\end{equation}
		of the two morphisms.
	\end{proposition}
	\begin{proof}
		The fact that $(r_1 \circ r_2)^*$ is a reduction follows because functional morphisms form a subcategory of $\cC$.
		
		Using the splittings of $q_1$ and $q_2$, we can write $q'$ in the split form of \eqref{eq:processing_split} via
		\begin{equation}\label{eq:comp_proof_split_1}
			\input{comp_proof_split_1.tikz}
		\end{equation}
		and
		\begin{equation}\label{eq:comp_proof_split_2}
			\input{comp_proof_split_2.tikz}
		\end{equation}
		
		Let us now show that $q'$ from \eqref{eq:processing_composition} satisfies relation \eqref{eq:processing_weak}. 
		First, precomposing the corresponding property for $q_2$ with the part of $q'$ without $q_2$ implies
		\begin{equation}\label{eq:comp_proof_1}
			\input{comp_proof_1.tikz}
		\end{equation}
		by the assumptions of target--context categories, specifically by property \eqref{eq:mrel_precomp}. 
		Next, using relation \eqref{eq:processing_weak} for $q_1$ and \eqref{eq:mrel_precomp} again, we get
		\begin{equation}\label{eq:comp_proof_2}
			\input{comp_proof_2.tikz}
		\end{equation}
		via precomposition with $r_2 \otimes \id_{T} \otimes \id_C$.
		Defining a morphism $g \colon P_2 \to I$ to be the composite $\discard_{P_1} \circ r_2$, we have $\discard_{P_2} \sqsupseteq g$ by the counitality of copying:
		\begin{equation}\label{eq:comp_proof_3}
			\input{comp_proof_3.tikz}
		\end{equation}
		Since domains are preserved by tensoring with identities, we can also infer 
		\begin{equation}\label{eq:comp_proof_4}
			\input{comp_proof_4.tikz}
		\end{equation}
		by the definition of a target--context category, specifically by property \eqref{eq:mrel_domain}.
		Combining relations \eqref{eq:comp_proof_4}, \eqref{eq:comp_proof_2} and \eqref{eq:comp_proof_1} yields relation \eqref{eq:processing_weak} for $q'$ by the transitivity of the ambient relation $\mrel$.
		
		\Cref{eq:processing_hit} for $q'$ can be derived as follows
		\begin{equation}\label{eq:comp_proof_5}
			\input{comp_proof_5.tikz}
		\end{equation}
		where the first equation is by definition of $q'$ via \cref{eq:processing_composition}, second equation uses associativity of copying and the fact that $r_2$ is functional, and the last equation the assumption that $(r_1^*,{q_1}_*)$ converts simulator $s$ to $s_1$.
		Since the right-hand side of \cref{eq:comp_proof_5} is $s_2$ by the assumption that $(r_2^*,{q_2}_*)$ converts $s_1$ to $s_2$, the result follows.
	\end{proof}
	It is not hard to check that this sequential composition is associative and has identities provided by identity reductions together with the identity ($P$-independent) processing.
	That is, morphisms of simulators indeed form a category.
	
	\begin{definition}\label{def:simulator_category}
		Given a target--context category $\cC$, the \textbf{simulator category} $\cat{Sim}(\cC)$ has simulators as objects and simulator morphisms as morphisms.
	\end{definition}
		
	\subsection{Parsimony of Universal Simulators}
	\label{sec:parsimony}

	\begin{figure}[t]\centering
	\includegraphics[width=.35\columnwidth]{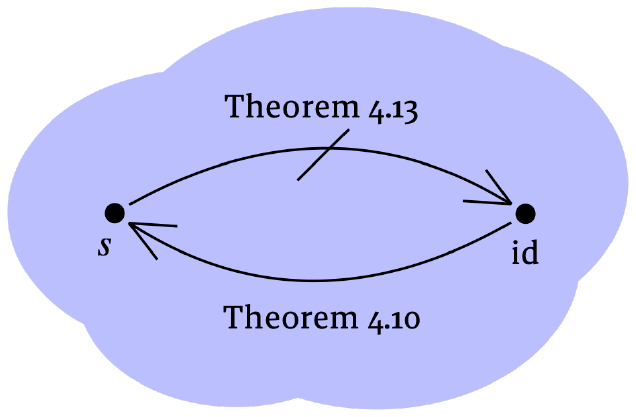}
	\caption{
		\Cref{thm:morph_stronger} gives sufficient conditions for a simulator $s$ to be more parsimonious than the trivial simulator. 
		\Cref{thm:s 2 id} gives necessary conditions, such as that $s$ is compressed. }
	\label{fig:properties-simulator-category}
	\end{figure}
	
	Every category gives rise to an order relation among its objects that expresses whether there exists a morphism from one object to another.
	In this way, simulator morphisms induce a hierarchy among all simulators and in particular among the universal ones.
	
	For instance, consider an identity reduction and a processing $q_* \colon s \to s'$ for which $q$ is independent of programs in $P$, i.e.\ it satisfies
	\begin{equation}
		\input{P-independent.tikz}
	\end{equation}
	for some $\tilde{q}$.
	Then we have $q_T \circ s_T = s'_T$, so that composition with $q_T$ implements a surjective function from the functional image of $s_T$ to that of $s'_T$.
	Therefore, the cardinality of the compiler's functional image of a simulator cannot increase under such processing.

	If we interpret the compiler's functional image as the collection of targets that a (universal) simulator makes use of, then in the above discussion $s'$ is \emph{not less parsimonious} than $s$ is in their use of targets.
	Such a strict conclusion cannot be made if $q$ is depends on $P$.
	However, any processing that does depend on $P$ is necessarily behaviorally imitated by one that is independent of $P$ as expressed by relation \eqref{eq:processing_weak}.
	Therefore, we can interpret the resulting preorder as describing the parsimony of universal simulators relative to the range of targets in $T$ they make use of.
	
	\begin{definition}\label{def:simulator_strength}
		Let $s$ and $s'$ be universal simulators.
		We say that $s'$ is a \textbf{more parsimonious}\footnotemark{} simulator than $s$ is if there exists a morphism $s \to s'$.
	\end{definition}
	\footnotetext{While it would be more accurate use the term ``not less parimonious'', we find ``more parsimonious'' to be easier to use.}%
	
	We now explore conditions under which two simulators can be compared with respect to their parsimony, starting with a sufficient condition for the parsimony ordering.
	We start by providing a sufficient condition for the parsimony ordering.
	\begin{theorem}[Sufficient condition for parsimony]\label{thm:morph_stronger}
		Let $s \colon P\otimes C \rightarrow  T \otimes C$ be a simulator and let $r^* \colon s \to \id$ be a lax reduction. 
		Assume that $r^*$ is also an oplax reduction, i.e.\ that additionally 
		\begin{equation}\label{eq:morph_stronger_1}
			\input{morph_stronger_1.tikz}
		\end{equation}
		holds.
		If there exists a morphism $m \colon P \to T$ that satisfies 
		\begin{equation}\label{eq:invert_red}
			\input{invert_red.tikz}
		\end{equation}
		then $s$ is more parsimonious than the trivial simulator.
	\end{theorem}
	Note that \cref{eq:invert_red} states a weak form of right invertibility of $r$ where $m$ acts as a right inverse of $r$ only after composing with $s$.
	In particular, whenever $r$ is right invertible, it trivially satisfies \cref{eq:invert_red}.
	{\renewcommand{\arraystretch}{1.4}
	\begin{table}[!ht]
		\centering
		\begin{tabular}{c|c|c} 
%			 \hline
			 	& \makecell{$r^*$ is an oplax reduction, \eqref{eq:morph_stronger_1}} 
			 	& \makecell{$s r m = s$, \cref{eq:invert_red}} 
			\\ \hline
			 \hyperref[ex:TM univ sim]{universal TM} & \checkmark & \checkmark  \\ 
			 \hyperref[ex:NP_complete]{NP completeness} & \checkmark & \checkmark \tablefootnote{Whether this condition holds may depend on the choice of NP-complete languages $L$ and $L'$ in \cref{ex:NP_complete}. 
				In particular, it holds if they are poly-time isomorphic to each other, which is known for the case of many NP-complete languages and conjectured to be the case for all NP-complete languages \cite{Be77a}.}  \\ 
			 \hyperref[sec:spinmodel]{universal spin model} & \checkmark & \checkmark \tablefootnote{While for the reduction $r$ in \cref{sec:spinmodel} that is obtained from \cite{De16b} this is not the case, it can be satisfied by modifying $r$ so that it acts as the identity on 2D Ising systems with fields.} \\
			 \hyperref[ex:cofinal]{cofinal subset} & $\times$ & \checkmark \\ 
			 \hyperref[ex:dense]{dense subset} & $\times$ & $\times$ \\
			 \hyperref[ex:universal_set]{universal Borel set} & \checkmark & \checkmark \\ 
%			 \hline
		\end{tabular}
		\caption{An overview of the universal simulator examples and whether they satisfy the two conditions appearing in \cref{thm:morph_stronger}.}
		\label{tab:morph_stronger}
	\end{table}
	}
	
	\begin{proof}[Proof of \cref{thm:morph_stronger}]
		We prove that $(m^*,q_*)$ with processing
		\begin{equation}
			\input{def_q_morph.tikz}
		\end{equation}
		defines a morphism of type $\id \to s$.
		First, let us argue that $q$ defines a valid processing. 
		Precomposing relation \eqref{eq:morph_stronger_1} with $\discard_P$ and using the fact that the ambient relation $\mrel$ is preserved by precomposition, we arrive at relation \eqref{eq:processing_weak}, which is necessary for $q$ to be a processing.
		Moreover, $q$ satisfies the splitting property of processings (i.e.\ \ref{it:split processing} in \cref{def:simulator processing}) because $r$ is functional and $s$ splits as in diagram \eqref{eq:simulator}. 
		It remains to show that $q_*$ maps $(m^* \id)$ to $s$.
		By construction of $q$, we have the first equality in 
		\begin{equation}
			\input{morph_uni.tikz}
		\end{equation}
		while the second one follows from \cref{eq:invert_red}.
		Consequently, $(m^*,q_*)$ indeed transforms the trivial simulator into $s$, i.e.\ it is a morphism of type $\id \to s$.
	\end{proof}
	We now apply \cref{thm:morph_stronger} to the universal simulator corresponding to a universal Turing machine (\cref{ex:TM univ sim}) to conclude that it is more parsimonious than the trivial simulator. 
	
	\begin{example}[$s_u$ more parsimonious than the trivial simulator]\label{ex:trivial2universal} 
		Let us revisit our running example of Turing machines (\cref{ex:TM}). 
		Let $u$ be a universal Turing machine and consider the universal simulator $s_u \colon P \otimes C \to T \otimes C$ from \cref{ex:TM univ sim}, where programs in $P$ are given by strings in $\Sigma^\star$. 
		By the same argument as in \cref{ex:universal2trivial}, every lax reduction $r^*$ of type $s_u \to \id$ is also an oplax reduction, i.e.\ relation \eqref{eq:morph_stronger_1} holds.
		Let us now argue that there is an encoding $r \colon T \to P$ for which both 
		\begin{itemize}
			\item $u$ is a universal Turing machine and
			\item $r$ has a right inverse $m \colon P \to T$. 
		\end{itemize}
		If $r$ is surjective, then we construct a computable right inverse $m$ as follows:
		Since $r$ is total computable and surjective, one can compute the first $t \in T$ (thinking of $T$ as the natural numbers $\mathbf{N}$) that satisfies $r(t) = p$. 
		We then let $m(p)$ be this Turing machine $t$.
		Clearly, the desired $r\circ m = \id_P$ then holds by definition.
		By \cref{thm:morph_stronger}, $s_u$ is thus more parsimonious than the trivial simulator.
		
		On the other hand, if $r$ is not surjective, let us consider two cases:
		First, assume that the image of $r$ is decidable.
		Then we can find\footnotemark{} a partial computable function $b$ such that $b \circ r$ is surjective and for which there is a total computable $m_b$ satisfying 
		\footnotetext{Note that this relies on the image being infinite, which follows from universality of $s_u$.}%
		\begin{equation}
			b \circ m_b = \id \qquad \text{and} \qquad m_b \circ b \circ r = r.
		\end{equation}
		Essentially, $b$ enumerates elements of the image and $m_b$ is constructed from $b$ similar to how $m$ is from $r$ above.
		Thus, there is a singleton universal simulator $s_u \circ (m_b \otimes \id_C)$ for $u$, whose lax reduction to the trivial simulator is surjective.
		
		Second, if the image of $r$ is not decidable, we can make it decidable by considering 
		\begin{equation}\label{eq:new_sim}
			\input{new_sim.tikz}
		\end{equation}
		instead of $s_u$ and $r$ respectively.
	\end{example}
	
	Next, we derive necessary conditions for the parsimony ordering. 
	Following the motivation from the beginning of \cref{sec:parsimony}, we think of one of them as a notion that tracks whether a simulator is ``compressed'' (relative to the trivial simulator).
	The idea is that a universal simulator $s$ is \emph{compressed} if there are at least two targets $t$, $g$ with distinct behavior (condition \ref{it:distinct_beh}) that are nevertheless simulated by the same target (condition \ref{it:same_compiled}). 

	\begin{definition}\label{def:compressed}
		A universal simulator $s \colon P\otimes C \to T \otimes C$ is \textbf{compressed} if for every lax reduction $r^*\colon s \to \id$, there exist functional states $t, g \colon I \to T$ such that 
		\begin{enumerate}
			\item \label{it:same_compiled} we have 
				\begin{equation}\label{eq:same_compiled}
					\input{same_compiled.tikz}
				\end{equation}
			\item \label{it:distinct_beh} and $g$ does not oplax context-reduce to $t$ (see \cref{def:c_red_rel}), i.e.\ there is no morphism $v \colon C \to C$ for which  
				\begin{equation}\label{eq:sim_preord}
					\input{sim_preord.tikz}
				\end{equation}
				holds.
		\end{enumerate}
	\end{definition}
	
	\begin{theorem}[Necessary condition for being more parsimonious]\label{thm:s 2 id}
		Let $(\cC,T,C,\mrel)$ be a target--context category and $s \colon P \otimes C \to T \otimes C$ a simulator.
		Assume that every lax reduction $r^* \colon s \to \id$ is also an oplax reduction, i.e.\ that
		\begin{equation}\label{eq:uni both 2}
			\input{morph_stronger_1.tikz}
		\end{equation}
		holds.
		If $s$ is a compressed simulator, then there exists no simulator morphism of type $s \to \mathrm{id}$, i.e.\ the trivial simulator is not more parsimonious than $s$.
	\end{theorem}
	
	\begin{proof}
		The proof is by contradiction. 
		Assume that a morphism $(r^*,q_*) \colon s \to \id$ exists so that by \cref{prop:morph_laxred} we have a lax reduction $r^* \colon s \to \id$.
		Since $s$ is compressed, there must exist $t$, $g$ as in \cref{def:compressed}.
		Let us construct a morphism $\nu \colon C \to C$, for which relation \eqref{eq:sim_preord} holds and thus leads to a contradiction.
		
		First, by \cref{lem:mrel_dom} we have 
		\begin{equation}\label{eq:uni 2 tri 0}
			\input{morph_uni_2_tri_0.tikz}
		\end{equation}
		Moreover, using relation \eqref{eq:uni both 2} and the fact that $\moonrel$ is preserved under precomposition, we find
		\begin{equation}\label{eq:uni 2 tri 1}
			\input{morph_uni_2_tri_4.tikz}
		\end{equation}
		while property \eqref{eq:processing_weak} satisfied by any processing implies
		\begin{equation}\label{eq:uni 2 tri 5}
			\input{morph_uni_2_tri_8.tikz}
		\end{equation}
		By the splitting properties \eqref{eq:simulator} and \eqref{eq:processing_split} of $s$ and $q$ respectively, as well as the functionality of $t$, $g$, $r$, and $s_T$, we can rewrite the right-hand side of relation \eqref{eq:uni 2 tri 5} as
		\begin{equation}\label{eq:uni 2 tri 6}
			\input{morph_uni_2_tri_9.tikz}
		\end{equation}
		where $\nu$ is given by 
		\begin{equation}\label{eq:morph_stronger_4}
			\input{morph_stronger_4.tikz}
		\end{equation}
		
		Combining relations \eqref{eq:uni 2 tri 0}--\eqref{eq:uni 2 tri 6} via transitivity of the ambient relation $\mrel$, the result follows if we can show that the right-hand side of \cref{eq:uni 2 tri 6} is in fact equal to $t \otimes \nu$.
		To this end, we have
		\begin{equation}\label{eq:uni 2 tri 7}
			\input{morph_uni_2_tri_10.tikz}
		\end{equation}
		where the first equation uses property \ref{it:same_compiled} of compressed simulators, the second holds by the functionality of $t$, and the third one follows because the morphism $(r^*,q_*)$ transforms simulator $s$ to the trivial simulator (see the left equation of \eqref{eq:processing_hit_2}).
		Thus, as we anticipated, this leads to a contradiction so that such a morphism  $(r^*,q_*) \colon s \to \id$ cannot exist.
 	\end{proof}
	
	 \begin{example}[Trivial simulator is not more parsimonious than $s_u$]\label{ex:universal2trivial}
		Let us show that given a universal Turing machine $u$ and its corresponding singleton universal simulator $s_u$ in the ambient category $\cat{Tur}$, there is no morphism of type $s_u \to \id$. 
		We can use \cref{thm:s 2 id} to establish this.
		
		First, we argue that every lax reduction of type $s \to \id$ for an arbitrary simulator $s$ in $\cat{Tur}$ is also an oplax reduction.
		By \cref{prop:beh_equality}, we have 
		\begin{equation}\label{eq:tur_reduction}
			\input{tur_reduction.tikz}
		\end{equation}
		for any such reduction, where the purple box indicated that the shadow functor is applied to both sides of the relation.
		Since $\behim{\eval}$ is a total function, there is only one function that agrees with it on its domain{\,\textemdash\,}the function itself{\,\textemdash\,}so that \eqref{eq:tur_reduction} is in fact an equality.
		Thus $r^* \colon s \to \id$ is an oplax reduction as well, once again by \cref{prop:beh_equality}.
		
		Next, let us show that $s_u$ is a compressed simulator.
		To do so, take any reduction $s_u \to \id$ implemented by a computable function $r \colon T \rightarrow P$ and any two Turing machines $t$, $g$ that evaluate to distinct, constant functions. 
		In particular, for all inputs $c,d \in C$, the output of $t$ applied to $c$, denoted by $\eval(t,c)$, is distinct from the output of $g$ applied to $d$, denoted by $\eval(g, d)$.
		This means that $g$ does not oplax context-reduce to $t$.
		Since for any $p\in P$, we have $s_T(p) = u$ by definition, we also necessarily have
		\begin{equation}
			s_T\bigl( r(t) \bigr) = s_T \bigl( r(g) \bigr),
		\end{equation}
		which is the required \cref{eq:same_compiled}.
		Hence, the conditions of \cref{thm:s 2 id} hold, from which it follows that no morphism of type $s_u \to \id$ exists.
\end{example}

	By a similar reasoning, if the image of $s_T$ only contains a finite number of Turing machines, there cannot exist a morphism of type $s \to \id$. 
	To see this, note that there exist infinitely many constant Turing machines that mutually do not oplax context-reduce to each other. 
	For any reduction $r$, there will thus always exist $t$ and $g$, mapped to the same Turing machine under $s_T \circ r$, such that $t$ does not oplax context-reduce to $g$.

\section{Universality and Unreachability}
\label{sec:Undecidability}

	In \cref{sec:Morphisms}, we introduced the simulator category and the parsimony of universal simulators from quite a general point of view.
	In this section, we instead provide a simple example of how the study of properties and consequences of universality can look like in our framework.
	For this, additional assumptions are often needed, which thus narrow the scope of the investigation.
	In particular, the general notion of a target--context category (\cref{def:instance}) is not rich enough for the discussion that follows.
	We use the more concrete target--context category with intrinsic behaviors (\cref{def:intrinsic_beh}) throughout this section.

	It has been argued in \cite{De20d} that universality can be related to undecidability (see also \cite[Chapter 7]{Mo11} for this relation in the context of computation). 
	The latter often stems from a contradiction arising from self-reference and negation, whose essence is known as the liar paradox, i.e.\ the statement ``I am a liar'' or similarly ``this sentence is false'' \cite{Ra19,Bo17}. 
	In the world of mathematics, one can arrive at similar statements via Cantor's diagonal argument and its variations.
	A general approach to these from the point of view of category theory has been provided by Lawvere in \cite{La69b}; this includes statements such as G\"odel's first incompleteness theorem, the uncomputability of the halting problem for Turing machines, Russell's paradox in set theory and Tarski's theorem on the non-definability of truth. 
	See also \cite{Ro21} for a recent discussion of a generalization of Lawvere's Fixed Point Theorem and \cite{Ya03} for a gentle introduction to it with many more examples.
	Hereafter, we investigate how undecidability connects to our framework for universality.

	\subsection{Undecidability as Unreachability}\label{sec:wps}
	
		In this section, we formulate a version of undecidability within a given target--context category.
		To say that the decision problem ``Is $x$ in $L$?'' for a language $L$ is \emph{undecidable} means that there is no halting Turing machine that accepts $x$ if and only if $x \in L$.
		In other words, it means that the function 
		\begin{equation}
			x \mapsto \begin{cases} 1 & x \in L \\ 0 & x \not \in L \end{cases}
		\end{equation}
		is not computable.
		That is, undecidability in this sense is a special case of uncomputability.
		
		The theory of computable functions is but one example of our framework. 
		In order to draw a more general connection between universality of simulators and undecidability, we have to generalize the latter so that we can interpret it in a generic instance of our framework.
		To this end, note that in a Turing category (\cref{ex:TM}) or a monoidal computer (\cref{sec:monoidal_computer}), all morphisms of type $C \to B$ are assumed to be ``programmable'', which, in these frameworks, means that they are computable. 
		Programmability is the assumption that for any such morphism $f$, there exists a Turing machine $t_f$ which implements the computable function $f$, a notion introduced in \cite{La69b} as weak point surjectivity.
		We use a different name that follows the terminology of \cite[definition 10]{Ro21}, where a negation of this concept is termed an \emph{incomplete parametrization}.
		
		\begin{definition}\label{def:wps_intrinsic}
			Consider a target--context category with intrinsic behaviors.
			We say that a morphism $F \colon P \otimes C \to B$ is an \textbf{$\bm{A}$-complete parametrization} of maps $C \to B$ if for every $f \colon A \otimes C \to B$, there exists a functional morphism $p_f \colon A \to P$ satisfying 
			\begin{equation}\label{eq:wps2}
				\input{complete_parametrization.tikz}
			\end{equation}
			where $\imrelop$ is given in \cref{def:brel}.
		\end{definition}
		The reading of string diagrams with purple background as in \cref{eq:wps2} is that we consider the image of the morphism under the shadow functor.
		To recover the notion of weak point surjectivity from \cite{La69b}, one has to replace the imitation relation in \eqref{eq:wps2} with equality and let $A$ be the unit \mbox{object $I$}.
		\begin{example}[Complete parametrization for Turing machines]\label{ex:wps_TM}
			Consider the running example of the ambient category $\cat{Tur}^\intr$ for Turing machines and let the $F$ from \cref{def:wps_intrinsic} be the associated (computable) evaluation function $\eval$ from \cref{ex:TM intrinsic}. 
			It is an $I$-complete parametrization as long as for every computable function $f \colon C \to B$, there is a Turing machine $t_f \colon I \to T$ satisfying
			\begin{equation}\label{eq:TM_for_function}
				\eval \bigl( t_f, \ph \bigr) = f(\ph).
			\end{equation}
			This is true by definition of course. 
			
			The evaluation function would be a complete parametrization if for every computable function $f \colon A \otimes C \to B$ with two inputs, there is a family
			\begin{equation}
				\Family{ t_f(a) \colon I \to T  \given  a \in A }
			\end{equation}
			of Turing machines, indexed by elements of $A$ (thought of as a set), such that we have 
			\begin{equation}
				\forall \, a \in A \qquad \eval \bigl( t_f(a), \ph \bigr) = f(a, \ph)
			\end{equation}
			\emph{and} such that the assignment $a \mapsto t_f(a)$ is a computable function of type $A \to T$. 
		\end{example}
		From the point of view of category theory, a better motivated concept than \cref{def:wps_intrinsic} would be to require that $F$ is an $A$-complete parametrization for all objects $A$ in $\cC$.
		In \cite{longo1990category}, this notion has been called Kleene universality, in \cite{Co08d} universal application, and in \cite{pavlovic2023programs} it is the defining property of the so-called program evaluator.
		
		If there was a morphism $C \to B$ in $\cat{Tur}$ for which no suitable Turing machine $t_f$ satisfying \cref{eq:TM_for_function} exists, it would be an \emph{uncomputable} function.
		This is the rationale behind the following concept of unreachability, which stipulates that there are behaviors that are inaccessible.
		\begin{definition}\label{def:unreachability}
			Consider a target--context category with intrinsic behaviors.
			A (universal) simulator $s$ has \textbf{unreachability} if the composite $\eval \circ s$ is not an $I$-complete parametrization of maps $C \to B$. 
		\end{definition}
		When we refer to unreachability for $P$, we implicitly have in mind a particular simulator of type $P \otimes C \to T \otimes C$ that it refers to.
		Note that unreachability for the trivial simulator is of particular conceptual importance since morphisms of type $C \to B$ that cannot be parametrized by $\eval$ are also not required (by \cref{def:univ sim}) to be parametrized by a universal simulator.
		In that sense such morphisms are not taken into account by the definition of universality. 
		
	\subsection{Unreachability from Universality and Negation}\label{sec:Lawvere}
	
		Let us now generalize Lawvere's Fixed Point Theorem \cite{La69b} in a way that applies to our framework.
		In its contrapositive form, this result says that if there is a morphism of type $B \to B$ without a fixed point (such as a negation), then no morphism of type $C \otimes C \to B$ is an $I$-complete parametrization. 
		For instance, if $T$ and $C$ are isomorphic, then this result entails unreachability for the trivial simulator.
		However, this need not be the case for distinct targets and contexts.
		Nevertheless, in \cref{prop:weak pt surj}, we show that for any universal simulator $s$ 
			unreachability for $s$ implies unreachability for the trivial simulator.
		Thus, if there is a negation $B \to B$ \emph{and} a universal simulator whose programs coincide with contexts, we obtain unreachability \mbox{for $T$}. 
		
		A fixed point of a morphism $g \colon B\to B$ is a $b \colon A \to B$ that equals $g \circ b$. 
		If merely
		\begin{equation}\label{eq:quasi-fixed}
			\behim{b} \imrelop \behim{g} \circ \behim{b}
		\end{equation}
		holds, where $\behim{b}$ stands for $\Beh(b)$ as before, we say that $b$ is a \textbf{quasi-fixed point} of $g$. 
		
		\begin{theorem}[Fixed Point Theorem]\label{thm:Lawvere}
			Let $\cC$ be a target--context category with intrinsic behaviors. 
			If there is a morphism $F \colon C \otimes C \to B$ that is an $A$-complete parametrization of maps $C \to B$ (for some $A$), then every morphism $g \colon B \to B$ has a quasi-fixed point.
		\end{theorem}
		\begin{remark}[Fixed Point Theorem ``without behaviors'']
			There is a similar version of the theorem in any gs-monoidal category without the structure of a target--context category with intrinsic behaviors.
			In this case, the notion of a complete parametrization uses equality instead of relation $\imrel$ in \eqref{eq:wps2} and the conclusion is that every $g$ has a fixed point (rather than a quasi-fixed point).
			To recover the original Lawvere's Fixed Point Theorem for products from \cite{La69b}, one can instantiate this version in a cartesian monoidal category (which is automatically gs-monoidal) with $A = I$. 
		\end{remark}
		
		\begin{proof}[Proof of \cref{thm:Lawvere}]
			Consider an arbitrary $g \colon B \to B$ and define a morphism
			\begin{equation}\label{eq:Lawvere_1}
				\input{Lawvere_1.tikz}
			\end{equation}
			which depends on $g$.
			Because $F$ is an $A$-complete parametrization, there exists a morphism ${c_f \in \cC_{\rm fun} (A,C)}$ satisfying
			\begin{equation}\label{eq:Lawvere_2}
				\input{Lawvere_2.tikz}
			\end{equation}
			Since $c_f$ is functional, the leftmost equation in the chain
			\begin{equation}\label{eq:Lawvere_3}
				\input{Lawvere_3.tikz}
			\end{equation}
			follows, while the second relation holds by \eqref{eq:Lawvere_2} and the third one by the definition of $f$ in \eqref{eq:Lawvere_1}.
			Consequently, $\gamma \circ c_f$ is a quasi-fixed point of $g$, where $\gamma \coloneqq F \circ \cop_C$.
		\end{proof}
	
		To gain an intuition of the proof of \cref{thm:Lawvere}, consider $F$ to be the evaluation function ${\eval \colon C \otimes C \to B}$, where we pick a particular isomorphism between strings $C$ and Turing machines $T$, i.e.\ a computable enumeration of Turing machines with a computable inverse. 
		Then the key object in the proof is the ``self-referential constructor'' $\gamma \colon C \to B$ given by $F \circ \cop_C$.
		It is a function that takes a generic string $c \colon I \to C$ and runs the Turing machine enumerated by it with input given by $c$ itself.
		The fixed point of $g$ is the outcome of running the self-referential constructor $\gamma$ with input given by $c_f$, which is constructed using both $\gamma$ and $g$.

		\begin{corollary}[Total fixed points]\label{cor:lawvere total}
			Consider the same situation as in \cref{thm:Lawvere} except for the additional assumption that 
			\begin{enumerate}\label{cond tot uni}
				\item\label{cond tot uni 1} $F\colon C \otimes C \to B$ is total
				\item\label{cond tot uni 2} $F$ is an $A$-complete parametrization (for some $A$) by total morphisms, i.e.\ for any ${f\colon A \otimes C \to B}$ there exists a total $c_f \colon C \to A$ such that \cref{eq:wps2} holds.
			\end{enumerate}
			Then every morphism $g \colon B \to B$ has a total quasi-fixed point.
		\end{corollary}
		\begin{proof}
			Following the same reasoning as in the proof of \cref{thm:Lawvere}, the additional assumptions imply that the quasi-fixed point (as a morphism in $\cC$) constructed in the leftmost expression in \eqref{eq:Lawvere_3}, is total.
		\end{proof}

	\begin{example}[The non-existence of a total universal Turing machine]\label{ex:no_total_UTM}
		Consider the target--context category $\cat{Tur}^\intr$. 
		We use \cref{cor:lawvere total} to recover the known fact that there does not exist a total universal Turing machine.
		Equivalently, this means that there is no total computable function $u \colon T \times C \to B$ that satisfies 
		\begin{equation}\label{eq:total compl}
			u(t,c) = t(c)
		\end{equation}
		whenever the Turing machine $t$ halts on input $c$.
		To do so, we show that, in $\cat{Tur}^\intr$, there is no universal simulator $s$ with a lax reduction $r^* \colon s \to \id$ such that the composite 
		\begin{equation}\label{eq:total_computable}
			\input{total_computable.tikz}
		\end{equation}
		is a total morphism.
		This is a stronger statement because, for every such simulator, the composite above is a total computable function satisfying \cref{eq:total compl}.
		In particular, since the trivial simulator is universal with a lax reduction given by the identity, we obtain that $\eval$ itself cannot be a total computable function.

		We prove the above claim by contradiction.
		That is, consider a simulator $s$ with the desired properties. 
		Letting $F$ from \cref{cor:lawvere total} be the morphism in diagram \eqref{eq:total_computable}, condition \ref{cond tot uni 1} is satisfied by assumption.
		Moreover, $\eval$ is not just an $I$-complete parametrization as argued in \cref{ex:wps_TM}, but it in fact satisfies the stronger condition \ref{cond tot uni 2} in \cref{cor:lawvere total}.
		Since $s$ is universal, using a similar reasoning as in \cref{prop:weak pt surj} we conclude that $F$ also satisfies this condition. 
		Hence, we conclude that every morphism $g\colon B \to B$ has a total quasi-fixed point by \cref{cor:lawvere total}.
		This, however, is not true.
		For instance, the succesor function does not have a total quasi-fixed point.
		Such a simulator therefore cannot exist.
	\end{example}
	
	Provided an intrinsic behavior structure, a good candidate for a complete parametrization in $\cC$ is the evaluation morphism $\eval \colon T \otimes C \to B$.
	However, in order to apply Lawvere's Fixed Point Theorem, we need to relate its two inputs{\,\textemdash\,}targets and contexts{\,\textemdash\,}like in the above example of Turing machines with an isomorphism between the two sets. 
	One way to do so more generally is to use a simulator with programs identical to the contexts $C$.
	Such a simulator allows us to construct a composite morphism $\eval \circ s$ with identical inputs.
	Let us show that this composite is a complete parametrization as long as $s$ is universal and $\eval$ is itself a complete parametrization.
	Having such a universal simulator allows one to build the self-referential constructor $\gamma$ and in this sense universality allows for self-reference.
	
	\begin{lemma}\label{prop:weak pt surj}
		Consider an intrinsic behavior structure where $\eval$ is an $A$-complete parametrization of maps $C \to B$.
		If $s$ is a universal simulator, then $\eval \circ s \colon P \otimes C \to B$ is also an $A$-complete parametrization.
	\end{lemma}
	\begin{proof}
		By assumption, we have that for any $f \colon A \otimes C \to B$ there is a functional $t_f \colon A \to T$ such that 
		\begin{equation}\label{eq:eval_compl_para}
			\input{eval_compl_para.tikz}
		\end{equation}
		holds.
		Let us denote the reduction of $s$ to the trivial simulator by $r \colon T \to P$.
		Then for every $t_f$ we define $p_f \coloneqq r \circ t_f$, for which   
		\begin{equation}\label{eq:weak_pt_surj}
			\input{weak_pt_surj.tikz}
		\end{equation} 
		holds by \cref{lem:brel_precomp}.
		Note that $p_f$ is a composition of functional morphisms and thus functional itself.
		Combining relations \eqref{eq:weak_pt_surj} and \eqref{eq:eval_compl_para} shows that $\eval \circ s$ is an $A$-complete parametrization of maps $C \to B$.
	\end{proof}
	In the contrapositive form, \cref{prop:weak pt surj} says that given a universal simulator whose evaluation is not a complete parametrization, we necessarily have unreachability in the sense of \cref{def:unreachability}.
	This is the meaning of the implication mentioned in the \nameref{sec:Introduction}:
	\begin{center}
		unreachability (for $P$) + universality $\implies$ unreachability (for particular solutions in $T$),
	\end{center}
	where we think of unreachability for $P$ as saying that the morphism $\eval \circ s$ that is not an $I$-complete parametrization.
	
	\begin{remark}[Relation to Pavlovic's monoidal computers]
		In the framework of monoidal \mbox{computers \cite{pavlovic2018monoidal}}, the above complete parametrization property of $\eval \circ s$ is among the axioms of the framework.
		As mentioned in \cref{sec:monoidal_computer}, $\eval \circ s$ corresponds to the so-called universal evaluation denoted by $\{\}$, which is assumed to be programmable \cite[proposition 3.2]{pavlovic2018monoidal}.
		That is, for every computation $f \colon C \to B$, there exists a program $p \colon I \to P$ such that composing the universal evaluation with $p$ yields $f$.
		In light of \cref{prop:weak pt surj}, it is appropriate to think of a universal evalulation from \cite{pavlovic2018monoidal} as a composition of the evaluation function $\eval$ with a \emph{universal} simulator $s$ rather than a generic simulator.
	\end{remark}
	
	\begin{theorem}[Quasi-fixed points and simulators]\label{thm:fix point sim}
		Consider an intrinsic behavior structure where $\eval$ is a $A$-complete parametrization of maps $C \to B$. 
		If there exists a universal simulator $s$ of type $C \otimes C \to T \otimes C$, then every morphism $g \colon B \to B$ has a quasi-fixed point. 
	\end{theorem}
	\begin{proof}
		By \cref{prop:weak pt surj}, $\eval \circ s$ is a complete parametrization and the claim follows from \cref{thm:Lawvere}.
	\end{proof}
	
	The contrapositive of \cref{thm:fix point sim} says:
	If there is a morphism $B \to B$ without a quasi-fixed point, then either $\eval$ is not a complete parametrization or there is no universal simulator of type $C \otimes C \to T \otimes C$.
	While morphisms without quasi-fixed points can arise in many circumstances, the paradigmatic example is the negation that exchanges the truth values of a Boolean variable. 
	This fixed-point-free morphism is the one that is generally used to prove undecidability and other kinds of unreachability. 
	We can thus summarize the contrapositive of \cref{thm:fix point sim} with the slogan (depicted in \cref{fig:unreachability}): 
	\begin{center} 
		negation + universality $\implies$ unreachability 
	\end{center}
	where
	\begin{itemize}
		\item \emph{negation} refers to a quasi-fixed-point-free morphism of type $B \to B$,
		\item \emph{universality} refers to the existence of a universal simulator with $P = C$, and 
		\item \emph{unreachability} refers to unreachability of the trivial simulator in the sense of \cref{def:unreachability}.
	\end{itemize}\vspace{\parskip}

\begin{figure}[t]\centering
	\includegraphics[width=.9\columnwidth]{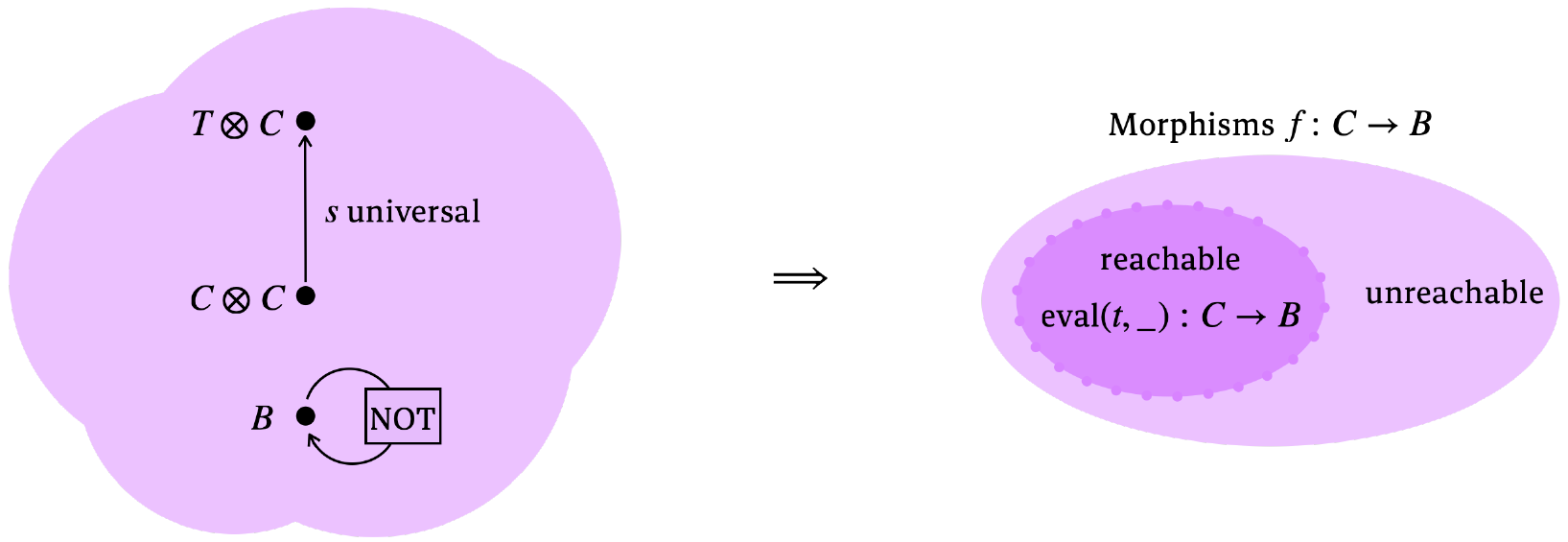}
	\caption{In a target--context category with intrinsic behaviors (\cref{fig:intrinsic-behavior}), if there exists a universal simulator $s$ of type $C\otimes C\to T\otimes C$, and a morphism $B\to B$ without a quasi-fixed point (such as negation), then the trivial simulator has unreachability.
	This means that there is a morphism $f \colon C\to B$ for which $\eval(t,\ph) \imrelop f(\ph)$ does not hold for any $t \colon I\to T$. }
	\label{fig:unreachability}
\end{figure} 

	\begin{example}[Cantor's Theorem]\label{ex:Cantor}
		Consider the ambient category $\cat{Set}$ of sets and functions with the canonical shadow functor given by the inclusion into $\cat{Rel}$.
		For the intrinsic behaviour structure, we take an arbitrary set $C$ as contexts, the power set $[2]^C$ as targets, the evaluation $\eval \colon [2]^C \times C \to [2]$ to be 
		\begin{equation}
			\eval(t,c) \coloneqq t(c)
		\end{equation}
		and $\brel$ as the equality relation on $[2] \coloneqq \{0,1\}$.
		Clearly, there exists a function $g \colon [2] \to [2]$ without a quasi-fixed point, namely the negation given by
		\begin{equation}
			g(0) = 1 \qquad \qquad g(1) = 0.
		\end{equation}
		Moreover, by construction, $\eval$ is a complete parametrization of maps $C \to [2]$. 
		Hence, by \cref{thm:fix point sim}, there does not exist a universal simulator $s \colon C \times C \to [2]^C \times C$.
		
		Cantor's Theorem says that there is no surjective map of type $l \colon C \to [2]^C$.
		If there was such an $l$, then the simulator with compiler $s_T \coloneqq l$ and context reduction $s_C \coloneqq \discard_C \times \id_C$ would clearly be universal with respect to the lax reduction given by any right inverse of $l$.
		
		Let us now argue that the existence of a universal simulator $s$ as above is in fact equivalent to the existence of such a surjection. 
		By definition, if $s$ is universal, then there exists a function $r \colon [2]^C \to C$ such that we have
		\begin{equation}\label{eq:Cantor_1}
			\input{Cantor_1.tikz}
		\end{equation}
		since the imitation relation is also the equality relation by \cref{ex:imitation_functions}.
		Taking $A = C$ and the function $\eval \circ s \colon C \times C \to [2]$ as $f$ in \cref{def:wps_intrinsic}, we obtain that there is a function $l \colon C \to [2]^C$ satisfying
		\begin{equation}\label{eq:Cantor_2}
			\input{Cantor_2.tikz}
		\end{equation}	
		because $\eval$ is a $C$-complete parametrization.
		Precomposing \cref{eq:Cantor_2} with $r \times \id_C$ and using \cref{eq:Cantor_1} to simplify the right-hand side gives
		\begin{equation}\label{eq:Cantor_3}
			\input{Cantor_3.tikz}
		\end{equation}	
		Therefore, by the universal property of $[2]^C$ as an exponential object, we obtain $l \circ r = \id_{[2]^C}$ or in other words that $r$ is a right inverse of $l$.
		This implies that $l$ is indeed a surjection of type $C \to [2]^C$.
		Hence, \cref{thm:fix point sim} allows us to restate Cantor's Theorem as the non-existence of universal simulators of a certain type.
	\end{example}

	\subsection{Complete Parametrization and Singleton Universal Simulators}\label{sec:complete_parametrization_singleton}
	
	In this section, we discuss two general methods to construct singleton universal simulators in target--context categories with reachability. 
	Recall that $s$ is a \emph{singleton} simulator if the functional image of the compiler $s_T$ is a single element $u \colon I \to T$. 
	If, additionally to the assumptions of \cref{prop:weak pt surj}, we have $B=C$, then there exists a singleton universal simulator which uses $\eval$ as the context reduction. 
	If, in addition, $T\otimes C$ can be reversibly embedded within $C$, there also exists another singleton universal simulator with the embedding $\sigma$ as a context reduction. 
	\begin{theorem}[Complete parametrization and singleton universal simulators]\label{thm:singleton}
		Consider an intrinsic behavior structure satisfying $B = C$ and where $\eval$ is a $I$-complete parametrization of maps $C \to C$. 
		\begin{enumerate}
			\item \label{prop:retract singleton 1}
				There is a $t_{\id} \colon I\to T$ such that 
				\begin{equation}
					\input{retract_singleton.tikz}
				\end{equation}
				is a singleton universal simulator of type $T\otimes C \to T \otimes C$.
					
			\item \label{prop:retract singleton 2}
				If $T\otimes C$ is a retract of $C$, with section $\sigma \colon T\otimes C \to C$ and retraction $\pi \colon C \to T\otimes C$, then there is a $t_u \colon I\to T$ such that
				\begin{equation}
					\input{retract_singleton_2.tikz}
				\end{equation}
				is a singleton universal simulator of type $T\otimes C \to T \otimes C$.
		\end{enumerate}
	\end{theorem}

	\begin{proof} 
		Note that property \ref{prop:retract singleton 1} holds under a weaker assumption than that of $\eval$ being a complete parametrization.
		The only fact we need about $\eval$ to prove universality of this simulator is that there be a functional state $t_{\id} \colon I\to T$ satisfying 
		\begin{equation}\label{eq:t_id}
			\input{t_id.tikz}
		\end{equation}
		This, of couse, holds under the assumed complete parametrization property.
		Let us now prove statements \ref{prop:retract singleton 1} and \ref{prop:retract singleton 2}.
		\begin{itemize}
			\item[\ref{prop:retract singleton 1}]
				Precomposing relation \eqref{eq:t_id} with $\eval$ and using the fact that imitation relations are preserved by precomposition (see the proof of \cref{lem:brel_precomp}), we obtain
				\begin{equation}\label{eq:t_id_2}
					\input{t_id_2.tikz}
				\end{equation}
				which says $s_{\rm id} \mrelop \id_{T \otimes C}$, so that $s_{\rm id}$ is a universal simulator indeed.
	
			\item[\ref{prop:retract singleton 2}]
				The composite $u \coloneqq  \eval \circ \pi$ is a morphism of type $C \to B$. 
				As $\eval$ is a complete parametrization by assumption, there is a functional state $t_u \colon I\to T$ such that 
				\begin{equation}\label{eq:singleton_proof_1}
					\input{singleton_proof_1.tikz}
				\end{equation}
				holds, which is nothing but $t_u \otimes \id_C \mrelop \pi$.
				Precomposing the latter relation with $\sigma$ and using the assumed $\pi \circ \sigma = \id_{T \otimes C}$ we get
				\begin{equation}\label{eq:singleton_proof_2}
					\input{singleton_proof_2.tikz}
				\end{equation}
				which proves the claim that $s_u$ is a universal simulator. \qedhere
		\end{itemize}
	\end{proof}
	
	Note that both simulators constructed in \cref{thm:singleton} identify programs and targets, i.e.\ we have ${P=T}$.
	That is why they cannot be used to obtain a non-trivial application of \cref{thm:fix point sim}.
	
\section{Comparing Instances of the Framework}\label{sec:Functors}	
		
	In this paper, we focus on studying individual target--context categories and concepts that can be described within each such instance of our framework{\,\textemdash\,}(universal) simulators, their reductions, morphisms, etc.
	However, we also pave the way for explorations of questions that have to do with multiple, a priori distinct, notions of universality.
	To do so, we define target--context functors  (\cref{fig:target--context-functor}) and show that the construction of simulator categories is functorial with respect to it.
	That is, given a relation at the level of ambient categories (the target--context functor), one obtains a functor between the corresponding simulator categories (\cref{fig:simfun}). 
	Moreover, the induced functor preserves universality of simulators.

\begin{figure}[t]\centering
	\includegraphics[width=.65\columnwidth]{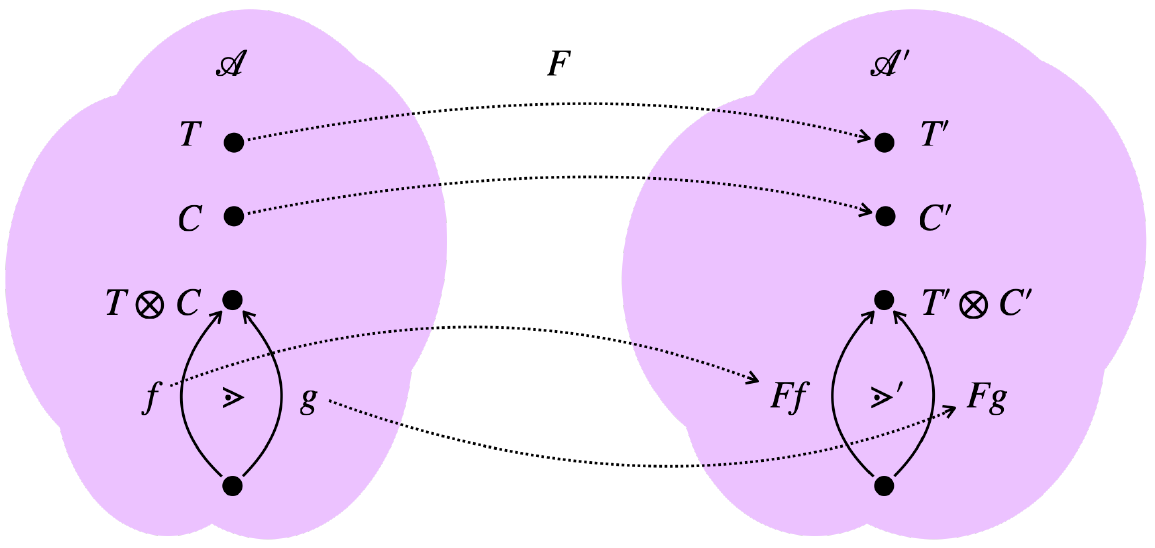}
	\caption{A target--context functor preserves targets, contexts, and the behavioral relation (\cref{def:TCfunctor}).}
	\label{fig:target--context-functor}
\end{figure} 
	
	\begin{definition}\label{def:TCfunctor}
		Consider two target--context categories $(\cC, T, C, \mrel)$ and $(\cC',T',C', \mrel')$.
		A \textbf{target--context functor} of type $(\cC, T, C, \mrel) \to (\cC',T',C', \mrel')$ is a strong gs-monoidal functor $F \colon \cC \to \cC'$ which preserves targets, contexts and the behavioral relation.
		That is, we have
		\begin{align}
			F(T) &= T', \\
			F(C) &= C', 
		\end{align}
		and 
		\begin{equation}
			f \mrelop g \quad \implies \quad F(f) \mrelop' F(g)
		\end{equation}
		for all $A$, and all $f,g \colon A \to T\otimes C$.
	\end{definition}
	
	\begin{figure}[t]\centering
		\includegraphics[width=.65\columnwidth]{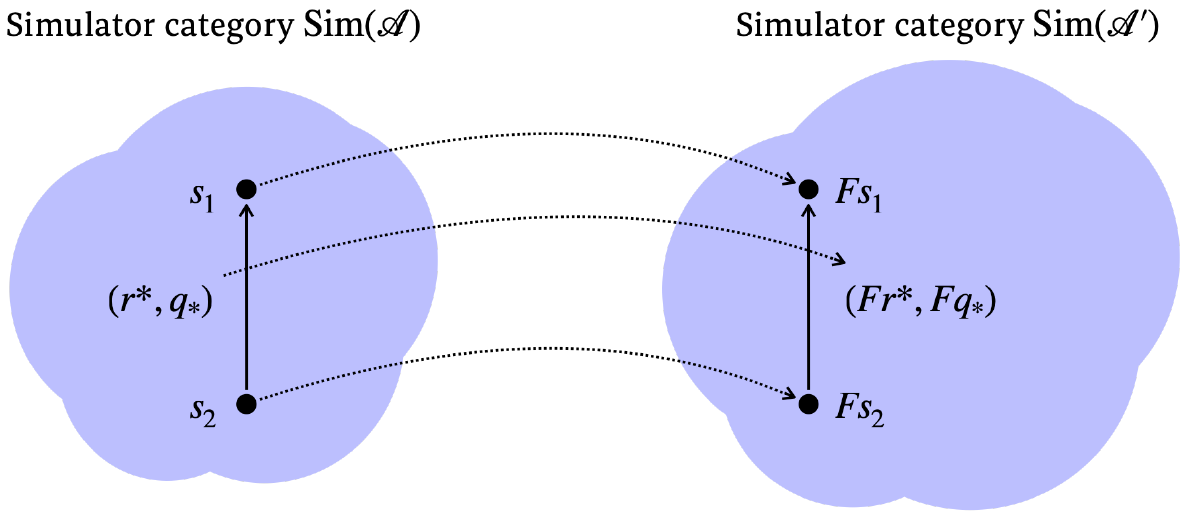}
		\caption{The construction of simulator categories is functorial with respect to target--context functors (\cref{thm:simfun}).}
		\label{fig:simfun}
	\end{figure} 
	\begin{theorem}[Functor between simulator categories]\label{thm:simfun}
		A target--context functor ${F\colon \mathcal{A} \to \mathcal{A}'}$ induces a functor $\mathrm{Sim}(F) \colon \textrm{Sim}(\mathcal{A}) \to  \textrm{Sim}(\mathcal{A}')$ defined as:
		\begin{equation}
			s \mapsto Fs  \qquad  (r^*,q_*) \mapsto \big ((F r)^*,(F q)_*\big)
		\end{equation}
	\end{theorem}
	\begin{proof}
		We omit any coherence isomorphisms of $F$. 
		We have to show that 
		\begin{enumerate}
			\item \label{it:functor_morphisms} given a simulator morphism $(q_*,r^*) \colon s_1 \to s_2$, we obtain a simulator morphism 
				\begin{equation}
					\bigl( (F q)_*,(F r)^* \bigr) \colon F s_1 \to F s_2
				\end{equation}
			\item \label{it:functor_id} $ F$ maps identity morphisms to identity morphisms, and
			\item \label{it:functor_seq} $ F$ preserves the composition of simulator morphisms.
		\end{enumerate}

		First, note that $F$ maps total morphisms of $\cC$ to total morphisms of $\cC'$ and functional morphisms to functional ones, as a direct consequence of being a strong gs-monoidal functor.
		In particular, for any total $f \colon A\to B$, we find
		\begin{equation}
			\discard_{FB} = F(\discard_B) = F(\discard_B \circ f) = \discard_{FB} \circ Ff .
		\end{equation}
		Similarly, for any functional $f \colon A\to B$, we have
		\begin{equation}
			(Ff \otimes Ff)\circ \cop_{FA} = F \bigl( (f\otimes f)\circ \cop_A \bigr) = F(\cop_B \circ f) = \cop_{FB} \circ Ff .
		\end{equation} 
		\begin{itemize}
			\item[\ref{it:functor_morphisms}] 		
				Let $s \colon P\otimes C \to T \otimes C$ be a simulator.
				By definition, $s$ comes equipped with a splitting into $s_C$ and $s_T$ as in diagram \eqref{eq:simulator}. 
				Applying $F$ to $s$, we thus find
				\begin{equation}
					Fs = (Fs_T \otimes Fs_C)\circ (\cop_{FP} \otimes \id_{FC}).
				\end{equation}
				Since $F$ preserves functionality of the compiler $s_T$, $Fs$ has the desired decomposition and is also a simulator.
				
				To show property \ref{it:functor_morphisms} we consider a simulator morphism $(q_*,r^*) \colon s_1 \to s_2$. 
				That is, let the morphism ${q \colon P\otimes T \otimes C \to T \otimes C}$ be a processing, i.e.\ a simulator that satisfies relations \eqref{eq:processing_weak} and \eqref{eq:processing_hit}.
				We now show that $((Fq)_*,(Fr)^*)$ is a simulator morphism of type $F s_1 \to F s_2$.
				Relation \eqref{eq:processing_weak} reads
				\begin{equation}
					\discard_P \otimes \id_T \otimes \id_C \mrelop q
				\end{equation}
				and applying $F$ to both sides yields
				\begin{equation}
					\discard_{FP} \otimes \id_{T'}\otimes \id_{C'} \mrelop' Fq,
				\end{equation}
				which is the appropriate relation \eqref{eq:processing_weak} in $\cC'$.
				Similarly, applying $F$ to both sides of \cref{eq:processing_hit} gives the relevant property in the image.
				Since $F$ preserves the functionality of the reduction $r^*$, we find that $((Fq)_*,(Fr)^*)$ is indeed a simulator morphism of type $F s_1 \to F s_2$.
			
			\item[\ref{it:functor_id}] 
				The identity simulator morphism $\id_s \colon s \to s$ is given by the pair 
				\begin{equation}
					\Bigl( \bigl( \discard_P \otimes \id_{T \otimes C} \bigr)_*, \ \id_P^* \Bigr).
				\end{equation}
				Applying $F$ to its components, we obtain 
				\begin{equation}
					\Bigl( \bigl( \discard_{FP} \otimes \id_{T' \otimes C'} \bigr)_*, \ \id_{FP}^* \Bigr),
				\end{equation}
				which is the identity simulator morphism $\id_{Fs}$.
			
			\item[\ref{it:functor_seq}] 
				Let $(q_{1*}, r_1^*) \colon s_1 \to s_2$ and $(q_{2*},  r_2^*) \colon s_2 \to s_3$ be two morphisms of simulators. 
				Their composition is given by $(\tilde q_*, (r_1\circ r_2)^*)$, where $\tilde q$ is given by \cref{eq:processing_composition}. 
				Since $F$ preserves string diagrams, it also commutes with composition of morphisms. \qedhere
		\end{itemize}
	\end{proof}
	\begin{remark}
		Note that when $F\colon \cC \to \cC'$ and $G \colon \cC' \to \cC''$ then it additionally holds that $\mathrm{Sim}(G \circ F) = \mathrm{Sim}(G) \circ \mathrm{Sim}(F)$. In other words, mapping target--context categories $\cC$ to their corresponding simulator categories $\mathrm{Sim}(\cC)$ and target--context functors $F$ to their induced functors $\mathrm{Sim}(F)$ 
		is itself functorial.
	\end{remark}

	\begin{theorem}\label{thm:TC_functor_universality}
		A target--context functor $F\colon \mathcal{A} \to \mathcal{A}'$ maps universal simulators to universal simulators and singleton simulators to singleton simulators.
	\end{theorem}	
	\begin{proof}
		Let $s \colon P \otimes C \to T \otimes C$ be a universal simulator, so that there exists a functional morphism ${r \colon T \to P}$ such that 
		\begin{equation}
			s \circ (r \otimes \id_C) \mrelop \id_{T \otimes C}
		\end{equation}
		holds.
		Using the assumption that $F$ preserves the ambient relation $\mrel$, we find
		\begin{equation}
			Fs \circ (Fr \otimes \id_{C'}) \mrelop' \id_{T' \otimes C'}.
		\end{equation}
		Since $F$ preserves functional morphisms, we conclude that $Fs$ is universal with the reduction given by $Fr$.
	
		On the other hand, let $s$ be a singleton simulator, that is, one whose compiler satisfies ${s_T = t \circ \discard_P}$ for a state $t \colon I\to T$.
		Since $F$ is strong gs-monoidal, we have 
		\begin{equation}
			Fs_T = F(t \circ \discard_P) = Ft \circ \discard_{FP}
		\end{equation}
		where $Ft$ is a state of type $I'\to FT$. 
		Hence, $Fs_T$ is a singelton simulator.
	\end{proof}
	
	Consider the situation of two instances of the framework with the same ambient category, same targets and same contexts, but distinct behavior relations.
	Then the identity functor is a target--context functor if and only if the domain behavior relation implies the codomain one, i.e.\ if and only if we have\footnotemark{}
	\footnotetext{Since there is a behavioral relation for each $A \in \cC$, we need to interpret the inclusion \eqref{eq:beh_enlarge} also as being quantified over all objects of $\cC$.}%
	\begin{equation}\label{eq:beh_enlarge}
		\mrel \, \subseteq \, \mrel'.
	\end{equation}
	In particular, enlarging behavioral relations preserves universality of simulators by \cref{thm:TC_functor_universality}. 
	
	\begin{example}[Modifying the behavioral relation for Turing machines]\label{ex:TM_mrel_functor}
		Let us exemplify this with Turing machines.
		In \cref{ex:TM}, we define two ambient categories for which the identity functor $\cat{Tur} \to \cat{Tur}^\intr$ is a target--context functor because the respective behavioral relations indeed satisfy condition \eqref{eq:beh_enlarge}.
		By \cref{thm:TC_functor_universality}, every universal simulator in $\cat{Tur}$ is also a universal simulator in $\cat{Tur}^\intr$.  One can easily construct simulators that are universal in $\cat{Tur}^\intr$ but not in $\cat{Tur}$. 
	\end{example}
	
\section{Summary, Conclusions \& Outlook}\label{sec:Conclusions}
		
		In this work we provide a categorical framework for studying universality and related concepts. 
		It comes in two main variants. 
		Let us summarize the key results of our work and highlight which variant they rely on, respectively.

		\parlabel{Results applicable to any target--context category}
		\Cref{def:instance} of a target--context category sets up the minimal structure we need for our analysis. 
		It includes a specification of the morphisms of the ambient category $\cC$, an ambient relation $\mrel$, and objects $T$ and $C$ of targets and contexts.
		Given this structure, we define simulators (\cref{def:simulator}), and their universality (\cref{def:univ sim}).
		In \cref{thm:nogo:cor}, we provide a recipe to generate necessary conditions for universality of a given simulator.
		These are used to show that there is no universal spin model of finite cardinality (\cref{ex:nogo_spin}).
		In order to compare simulators via the so-called parsimony ordering (\cref{def:simulator_strength}), we propose the notion of a simulator category (\cref{sec:simulator_category}), whose morphisms model ways to transform one simulator into another.
		\Cref{thm:morph_stronger,thm:s 2 id} then give conditions that allow one to decide whether such a morphism exists or not.
		In \cref{ex:trivial2universal,ex:universal2trivial}, we use them to show that the simulator associated to a universal Turing machine (\cref{ex:TM univ sim}) is strictly more parsimonious than the trivial simulator.
		Finally, we define target--context functors and show that they preserve simulator categories and universal simulators in \cref{sec:Functors}.
			
		\parlabel{Results applicable to target--context categories with intrinsic behaviors}
		Provided that the target--context category in question additionally has a behavior structure in accordance with \cref{def:instance_beh}, and in particular one that is intrinsic (\cref{def:intrinsic_beh}), we can study fixed point theorems and unreachability (\cref{def:unreachability}).
		 Specifically, \cref{thm:Lawvere} generalizes Lawvere's original fixed point theorem to our framework.
		 One of its consequences is the fact that there exists no total universal Turing machine (\cref{ex:no_total_UTM}).
		 In \cref{thm:fix point sim}, we also prove that universal simulators allow one to further extend this result to the case of two non-isomorphic inputs.
		 To explicate its potential use, we show in \cref{ex:Cantor} that the non-existsence of certain universal simulators  is equivalent to the conclusion of Cantor's Theorem, namely that the power set always has a larger cardinality than the original set.

		\parlabel{Instances of our framework}
		We illustrate our constructions and results with a wealth of examples.
		The two most important examples are those of Turing machine universality and of universal spin models.
		Besides the aforementioned concrete results, expressing spin systems in the framework helps us understand the meaning of spin model universality.
		
		Others mentioned throughout the text include completeness for complexity classes (\cref{ex:NP_complete}), cofinal subsets of preordered sets (\cref{ex:cofinal}), dense subsets of topological spaces (\cref{ex:dense}), and universal sets (\cref{ex:universal_set}).
		Apart from completeness for complexity classes, all examples come with an intrinsic behavior structure. 
		Notably, however, among the two ambient categories for Turing machines that we present, only the one with more permissive ambient relation $\mrel$ has intrinsic behaviors.
		See \cref{rem:TM_exampleS_diff} for more details.
		
		Besides the above, we expect that there are many relevant instances of our framework, whose details are yet to be worked out:
		\begin{enumerate}
			\item \emph{Universal graphs:} A graph $G$ within a given class is called a universal graph if all graphs from this class can be embedded into $G$ (\cite{Ra64}).
				Distinct notions of universality arise for different types of embeddings.
				In the simplest case, a universal graph contains all other ones as subgraphs. 
				
				\item \emph{Neural networks:} In order to construct complex functions from basic building blocks, one can use neural networks.
					By the universal approximation theorem \cite{Cy89,Ho91b,Cs01}, feed-forward neural networks with one hidden layer can approximate any continuous function to arbitrary precision.
					Analogous results hold for restricted Boltzmann machines \cite{Le07}. 
					As such, these neural network architectures are universal for the task of expressing an unknown function in a compositional manner.
					The idea of funtional completeness for categories of polynomial circuits \cite[section 5]{wilson2023axiomatic} provides another perspective on this example and could perhaps be formulated in our framework.
					See also \cite{De20d} for similarities between this type of universality, universal spin models and universal Turing machines. 

				\item \emph{Universal grammar:} In order to explain that an individual is capable of learning any natural language, Chomsky proposed the idea of a universal grammar as a `meta-grammar' with unfixed parameters \cite{Ch65} (see also \cite{Hi13b}). 
					``Programming'' (in the sense of our discussion from the \nameref{sec:Introduction}) is synonymous with fixing particular parameter values, upon which the universal grammar specializes to the grammar of any natural language. 
					
				\item \emph{Generating sets:} Every algebraic structure gives rise to the notion of a generating set. 
					For instance any set of vectors $S \coloneqq \{v_1, \ldots , v_k\}$ whose span is the whole vector space itself is termed a generating set of this vector space.
					Generating sets can be seen as universal in the sense that every vector $v$ is a linear combination $\lambda_1 \cdot v_1+ \ldots + \lambda_k \cdot v_k$ for some coefficients, which we think of as the program for $v$.
					
				\item \emph{Universal quantum spin models:} While \cref{sec:spinmodel} captures universality for classical spin models, there are also notions of universality for quantum spin models \cite{Cu17,Zh21}. 
		\end{enumerate}
		While all of the above can be instantiated in uninsightful ways{\,\textemdash\,}e.g.\ as a special case of \cref{ex:cofinal} of a cofinal subset (or, specifically, the top element) of a preordered set{\,\textemdash\,}the challenge is to find more informative target--context categories.
		Not only should these notions of universality correspond to universal simulators therein, but one should also be able to recover relevant information about \emph{the way in which they are universal} from the data provided, such as the context reduction $s_C$ or the lax reduction $r^*$ to the trivial simulator.
		
		For instance, one could also think of a universal spin model in the sense of \cite{De16b} as a cofinal subset of the set of all spin systems, given an appropriate preorder relation.
		Two spin systems would be related if there is a spin system simulation \cite[definition 3]{De16b} from one to the other.
		In our treatment of spin model universality in \cref{sec:spinmodel}, we go beyond such austere description.
		The context reduction of a 2D Ising model simulator specifies, in particular, how to relate the respective spin configurations, which is part of the data needed to identify a simulation of spin systems.
		Nevertheless, even our approach here leaves much to be desired.
		For example, a generic context reduction does not satisfy some of desirable properties of a spin system simulation, such as the preservation of the hypergraph structure at the level of spin configurations \cite[\mbox{definition 3 (ii)}]{De16b}.
		This is precisely the reason why a generic universal simulator in the target--context category $\cat{SpinSys}$ need not correspond to a universal spin model in the sense of \cite{De16b}.
		
		\parlabel{Potential further development of the framework}
		A related open problem is to establish desiderata for distinguishing between trivial and insightful ways to instantiate examples of universality in our framework.
		One is that simulators (of particular type) can be shown to be universal \emph{if and only if} an independently motivated notion of universality is present in the example.
		This is the case for the ambient category $\cat{Tur}$ of Turing machines, but as we argue above, it is not true for the ambient category $\cat{SpinSys}$ of spin systems.
		Certainly, an insightful instantiation would lead to an improved understanding of the situation at hand.
		In this sense, our present spin system example is already non-trivial.

		Additionally, we would like to understand the parsimony preorder of universal simulators better, both in specific instances and in the abstract.
		This preorder identifies ``non-trivial'' universal simulators as those universal simulators that are strictly more parsimonious than the trivial simulator.
		One may wish to characterize which target--context categories feature such non-trivial universality.
		\Cref{thm:morph_stronger,thm:s 2 id} serve as first steps in this direction.
		
		We envision a generalization of the framework with unfixed targets $T$ and contexts $C$. 
		This would allow one to describe Turing completeness via the corresponding notion of universal simulators. 
		
		\parlabel{Universality beyond our framework}
		On a broader and less precise level, we would like to understand whether our framework relates to notions of universality that are not as easily expressible in the mathematical language used here.
		One of them is the idea of \emph{xenobots} \cite{blackiston2021cellular}, which are biological systems that are capable of programmable biosynthesis. 
		Controlling (i.e.\  programming) their environment allows one to implement a range of biological process that are not accessible by such a relatively simple system under normal circumstances.
		A \newemph{chemputer} \cite{gromski2020universal}{\,\textemdash\,}on the other hand{\,\textemdash\,}is a programmable modular system capable of chemical synthesis. 
		Universals in biology \cite{Go17}, complex systems \cite{So00}, and statistical physics refer to \newemph{emergent properties}, which can serve as universal solutions (see the \nameref{sec:Introduction}) to the problem of explaining the phenomena of specific systems enjoying such emergence.
		In this sense, instantiating these examples in our framework would amount to expressing universal explanations as the universality of certain simulators. 
		A general version of this problem is known as the \newemph{problem of universals in metaphysics} \cite{Ca10d,Lo17}.

%\newpage
\appendix

%%===========================================
\section{Background on Category Theory}
\label{sec:category_theory}

	\subsection{Basic Definitions}
	We start by providing some basic definitions. 
	For an introduction to category theory for readers with some background in mathematics, we recommend \cite{perrone2024starting}.
	
	\begin{definition}\label{def:category}
		\cite[Section 1.1]{perrone2024starting}
		A \textbf{category} $\cat{C}$ consists of a collection of objects and a collection of morphisms (or arrows) with the following properties:
		\begin{enumerate}[label=(\roman*)]
			\item For every morphism $f$ there are two objects $A$ and $B$ called the domain and codomain of $f$, in which case we express this information as $f \colon  A \to B$ and refer to $A \to B$ as the type of $f$.
			\item For every object $A$, there is a morphism $\id_A \colon  A \to A$ called the identity morphism (on $A$).
			\item For every pair of morphisms $f \colon  A \to B$ and $g \colon  B \to C$, there is a morphism $g\circ f \colon  A \to C$ called the sequential composite of $f$ and $g$.
			\item Composing every morphism $f \colon  A \to B$, with identities gives $f$, i.e.\ $\id_B\circ f$ and $f\circ \id_A$ are both equal to $f$.
			\item Composition is associative, i.e.\ for all morphisms $f \colon  A \to B$, $g \colon  B \to C$ and $h \colon  C \to D$, we have $(h\circ g)\circ f = h\circ (g\circ f)$.
		\end{enumerate}
			\end{definition}
		In words, the objects can be thought of as types (in the computer science way) and morphisms as functions mapping from one type to another. 
		Composition of morphisms is hence only possible when the output type of one morphism is the input type of the other.
		Graphically, a (small) category can be imagined as a directed graph, where the objects are nodes and the morphisms are directed edges.
		
		We depict morphisms of categories by string diagrams.
		Each morphism $f \colon A \to X$ corresponds to a box such as
		\begin{equation}
			\input{box.tikz}
		\end{equation}
		while each object (and thus also each identity morphism) is a wire carrying a label with a name of the object.
		For more details on string diagrams, see \cite{Coecke2009,piedeleu2023introduction}.

		Given two objects $A$ and $B$, we denote the set of morphisms from $A$ to $B$ (i.e.\ the set of morphisms with domain $A$ and codomain $B$) by $\cat{C}(A,B)$. 
		This set is also called the hom-set of $A$ and $B$.
		
		A \emph{functor} is a map between two categories that preserves morphism composition. %structure and identity morphisms
	\begin{definition}\label{def:functor}
		\cite[Section 1.3]{perrone2024starting}
		A \textbf{functor} $F \colon \cat{C} \to \cat{D}$ consists of 
		\begin{enumerate}
			\item a map on objects $A \in Ob(\cat{C}) \mapsto FA \in Ob(\cat{D})$ and
			\item a map on morphisms $(f \colon  A \to B )\mapsto (Ff \colon  FA \to FB)$,
		\end{enumerate}
		such that for any objects $A$, $B$, and $C$ and any morphisms $f \colon A\to B$, $g \colon B\to C$: 
		\begin{enumerate}
			\item the identity morphism $\id_A$ is mapped to $F\id_A=\id_{FA}$ and
			\item the composition $g\circ f$ is mapped to $F(g\circ f)=Fg\circ Ff$.
		\end{enumerate}

	\end{definition}
	\begin{definition}\label{def:subcategory}
		\cite[Section 1.5]{perrone2024starting}
		A \textbf{subcategory} $\cat{D}$ of a category $\cat{C}$ consists of a subcollection of objects and a subcollection of morphisms such that $\cat{D}$ is itself a category. 
	\end{definition}
	
	\subsection{Categories With Copying and Discarding}
	
	The definition of a gs-monoidal category is described in \cref{sec:ambient}, for more details we recommend \cite[section 4.2.4]{piedeleu2023introduction}.

	\begin{definition}[\cite{MacLane1971}]
	\label{def:monoidal_functor}\		
		A functor $F \colon \cat{C} \to \cat{D}$ between two monoidal categories $\cat{C}$ and $\cat{D}$ is called \textbf{lax monoidal} if there exists a morphism 
		\begin{equation}
			\psi_0 \colon I_\cat{D} \to F(I_\cat{C})
		\end{equation}
		and a natural transformation
		\begin{equation}
			\psi_{A,B} \colon F(A) \otimes F(B) \to F(A \otimes B)
		\end{equation}
		such that the following diagrams commute:
		\begin{enumerate}
			\item For all objects $A, B, C \in \cat{C}$ 
				\begin{equation}
					\begin{tikzcd}
						(F(A)\otimes F(B)) \otimes F(C) \arrow[rr, "{\alpha_{F(A),F(B),F(C)}}"] \arrow[d, "{\psi_{A,B}\otimes \id_{F(C)}}"'] &  & F(A)\otimes (F(B) \otimes F(C)) \arrow[d, "{\id_{F(A)}\otimes \psi_{B,C}}"] \\
						F(A\otimes B)\otimes F(C) \arrow[d, "{\psi_{A\otimes B, C}}"']                                                       &  & F(A)\otimes F(B\otimes C) \arrow[d, "{\psi_{A,B\otimes C}}"]                \\
						F((A\otimes B)\otimes C) \arrow[rr, "{F(\alpha_{A,B,C})}"']                                                           &  & F(A\otimes (B\otimes C))                                                   
					\end{tikzcd}
				\end{equation}
				where $\alpha$ are the associators on $\cat{C}$ and $\cat{D}$.
			\item For all objects $A \in \cat{C}$
				\begin{equation}
					\begin{tikzcd}
						I_\cat{D}\otimes F(A) \arrow[r, "\psi_0\otimes \id_{F(A)}"] \arrow[d, "\ell_{F(A)}"'] & F(I_\cat{C})\otimes F(A) \arrow[d, "{\psi_{I_\cat{C},A}}"] \arrow[rr, "\mathrm{and}", phantom, shift right=9] &  & F(A)\otimes I_\cat{D} \arrow[d, "r_{F(A)}"] \arrow[r, "\id_{F(A)}\otimes \psi_0"] & F(A) \otimes F(I_\cat{C}) \arrow[d, "{\psi_{A, I_\cat{C}}}"] \\
						F(A)                                                                                  & F(I_\cat{C}\otimes A) \arrow[l, "F(\ell_A)"]                                                                  &  & F(A)                                                                            & F(A\otimes I_\cat{A}) \arrow[l, "F(r_A)"]                   
					\end{tikzcd}
				\end{equation}
				where $\ell$ and $r$ are the left and right unitors on $\cat{C}$ and $\cat{D}$.
		\end{enumerate}
		If $\cat{C}$ and $\cat{D}$ are symmetric monoidal categories, then a lax monoidal functor $F$ is called \textbf{lax symmetric monoidal} if the following diagrams commute for any objects $A,B$ of $\cat{C}$:
		\begin{equation}
			\begin{tikzcd}
				F(A)\otimes F(B) \arrow[d, "{\psi_{A,B}}"'] \arrow[r, "{\beta_{F(A), F(B)}}"] & F(B)\otimes F(A) \arrow[d, "{\psi_{B,A}}"] \\
				F(A\otimes B)  \arrow[r, "{F(\beta_{A,B})}"']                                 & F(B\otimes A)                             
			\end{tikzcd}
		\end{equation}
		where $\beta$ is the braiding map on $\cat{C}$ and $\cat{D}$.
	\end{definition}
	
	\begin{definition}\cite{fritz2022lax}
	\label{def:gs-monoidal_functor}
		A lax symmetric monoidal functor $F \colon \cat{C} \to \cat{D}$ between two gs-monoidal categories $\cat{C}$ and $\cat{D}$ is called \textbf{lax gs-monoidal} if the following diagrams commute for any object $A$ of $\cat{C}$:
		\begin{equation}
			\begin{tikzcd}
				F(A) \arrow[rd, "\cop_{F(A)}"'] \arrow[rr, "F(\cop_A)"] &                                              & F(A\otimes A) &  & F(A) \arrow[rr, "F(\discard_A)"] \arrow[rd, "\discard_{F(A)}"'] &                         & F(I_\cat{C}) \\
																	  & F(A)\otimes F(A) \arrow[ru, "{\psi_{A,A}}"'] &               &  &                                                                 & I_\cat{D} \arrow[ru, "\psi_0"'] &     
				\end{tikzcd}
		\end{equation}
		where $I_\cat{C}, I_\cat{D}$ are the monoidal units on $\cat{C}$ and $\cat{D}$, respectively and where $\psi_{A,A}, \psi_0$ are the lax monoidal structure morphisms.
	\end{definition}
	
	\begin{proposition}\label{prop:sha_dom}
		Let $F \colon \cat{C} \to \cat{D}$ be a lax gs-monoidal functor.
		Then for every morphism $g$ in $\cat{C}$ we have
		\begin{equation}\label{eq:sha_dom}
			F \bigl( \dom(g) \bigr) = \dom \bigl( F(g) \bigr)
		\end{equation}
	\end{proposition}
	Here, given $f \colon A \to X$, domain of $f$ (\cref{def:domain}) is the morphism $\dom(f) \colon A \to A$ given by
	\begin{equation}
		\input{domain.tikz}
	\end{equation}
	\begin{proof} 
		We denote with an overline objects and morphisms in the image of $F$, e.g.\ $\overline{A}\coloneqq F(A)$. 
		Using the definition of domain in gs-monoidal categories, we can write the left hand side as the top path in the following commutative diagram:
		\begin{equation}
			% https://q.uiver.app/?q=WzAsOCxbMCwwLCJcXG92ZXJsaW5le0F9Il0sWzIsMCwiXFxvdmVybGluZXtBIFxcb3RpbWVzIEF9Il0sWzQsMCwiXFxvdmVybGluZXtBIFxcb3RpbWVzIEJ9Il0sWzYsMiwiXFxvdmVybGluZXtBIFxcb3RpbWVzIEl9Il0sWzAsMiwiXFxvdmVybGluZXtBfSBcXG90aW1lcyBcXG92ZXJsaW5le0F9Il0sWzIsMiwiXFxvdmVybGluZXtBfSBcXG90aW1lcyBcXG92ZXJsaW5le0J9Il0sWzQsMiwiXFxvdmVybGluZXtBfSBcXG90aW1lcyBcXG92ZXJsaW5le0l9Il0sWzQsNCwiXFxvdmVybGluZXtBfSBcXG90aW1lcyBJIl0sWzAsMSwiXFxvdmVybGluZXtcXGNvcF9BfSJdLFsxLDIsIlxcb3ZlcmxpbmV7QSBcXG90aW1lcyBnfSJdLFsyLDMsIlxcb3ZlcmxpbmV7QSBcXG90aW1lcyBcXGRpc2NhcmRfQn0iXSxbMCw0LCJcXGNvcF97XFxvdmVybGluZXtBfX0iLDJdLFs0LDEsIlxccHNpX3tBLEF9IiwyXSxbNCw1LCJcXG92ZXJsaW5le0F9IFxcb3RpbWVzIFxcb3ZlcmxpbmV7Z30iLDJdLFs1LDIsIlxccHNpX3tBLEJ9IiwyXSxbNSw2LCJcXG92ZXJsaW5le0F9IFxcb3RpbWVzIFxcb3ZlcmxpbmV7XFxkaXNjYXJkX0J9IiwyXSxbNiwzLCJcXHBzaV97QSxJfSIsMl0sWzUsNywiXFxvdmVybGluZXtBfSBcXG90aW1lcyBcXGRpc2NhcmRfe1xcb3ZlcmxpbmV7Qn19IiwyXSxbNywzLCJcXGNvbmciLDJdXQ==
			\begin{tikzcd}
				{\overline{A}} && {\overline{A \otimes A}} && {\overline{A \otimes B}} \\
				\\
				{\overline{A} \otimes \overline{A}} && {\overline{A} \otimes \overline{B}} && {\overline{A} \otimes \overline{I}} && {\overline{A \otimes I}} \\
				\\
				&&&& {\overline{A} \otimes I}
				\arrow["{\overline{\cop_A}}", from=1-1, to=1-3]
				\arrow["{\overline{A \otimes g}}", from=1-3, to=1-5]
				\arrow["{\overline{A \otimes \discard_B}}", from=1-5, to=3-7]
				\arrow["{\cop_{\overline{A}}}"', from=1-1, to=3-1]
				\arrow["{\psi_{A,A}}"', from=3-1, to=1-3]
				\arrow["{\overline{A} \otimes \overline{g}}"', from=3-1, to=3-3]
				\arrow["{\psi_{A,B}}"', from=3-3, to=1-5]
				\arrow["{\overline{A} \otimes \overline{\discard_B}}"', from=3-3, to=3-5]
				\arrow["{\psi_{A,I}}"', from=3-5, to=3-7]
				\arrow["{\overline{A} \otimes \discard_{\overline{B}}}"', from=3-3, to=5-5]
				\arrow["\cong"', from=5-5, to=3-7]
			\end{tikzcd}
		\end{equation}
		The inner cells commute by the assumed coherences and we can read off the bottom path as the right hand side of \cref{eq:sha_dom}.
	\end{proof}

\section{Supplementary Material on gs-Monoidal Categories}
\label{sec:additional_def}

	\begin{lemma}[Functionality of domains]\label{lem:dom_fun}
		Let $f \colon A \to X$ be a quasi-total morphism in a gs-monoidal category. 
		Then its domain $\dom(f) \colon A \to A$ is a functional morphism.
	\end{lemma}
	\begin{proof}
		We can show this directly from the definition of functional morphisms:
		\begin{equation}\label{eq:dom_fun}
			\input{dom_fun.tikz}
		\end{equation}
		The first equation holds because $f$ is quasi-total and the second used just the associativity of copying.
	\end{proof}
	
	The following definitions are needed only needed for the construction of the pointed shadow functor in \cref{sec:pointed_shadow}. 
	We start with notions of causality and supports, which have been defined elsewhere for Markov categories (i.e.\ gs-monoidal categories in which all morphisms are total).
		
	\begin{definition}[{\cite[definition 11.31]{fritz2019synthetic}}]\label{def:causality}
		A gs-monoidal category $\cat{C}$ is \textbf{causal} if whenever ${f \colon A \to W}$, $g \colon W \to X$ and $h_1, h_2 \colon X \to Y$ satisfy
		\begin{equation}\label{eq:causal1}
			\input{causal1.tikz}
		\end{equation}
		then the stronger equation
		\begin{equation}\label{eq:causal2}
			\input{causal2.tikz}
		\end{equation} 
		also holds.
	\end{definition}
	
	\begin{definition}\label{def:ac}
		Let $f \colon A \to X$ and $g \colon Z \to X$ be morphisms in a gs-monoidal category. 
		We say that \textbf{$\bm{f}$ is absolutely continuous with respect to $\bm{g}$}, denoted $f \ll g$, if the implication
		\begin{equation}\label{eq:abs_cont}
			\input{abs_cont.tikz}
		\end{equation}
		holds for arbitrary objects $Y$, $W$ and any two parallel morphisms $h, h' \colon W \otimes X \to Y$.
	\end{definition}
	
	\begin{definition}[{\cite{fritz2023supports}}]\label{def:supp}
		Let $\cC$ be a gs-monoidal category.
		A \textbf{support} for $f \colon A \to X$ is given by an object $\Supp{}$ and a deterministic morphism $\suppinc{} \colon \Supp{} \to X$, such that we have
		\begin{equation}\label{eq:supp_asfaithful}
			\input{supp_factor_1.tikz}
		\end{equation}
		for all $h, h' \colon W \otimes X \to Y$ with arbitrary $W$ and $Y$.
	\end{definition}
	
	\begin{definition}\label{def:point_liftings}
		We say that $\cC$ has \textbf{point liftings} if for all $f \colon A \to X$ and all $x \in \cC_\det (I,X)$ satisfying $x \ll f$, there exists an $a \in \cC_\det (I,A)$ satisfying $x \ll f \circ a$.
	\end{definition}

	\begin{theorem}[Pointed shadow, proof can be found in {\cite{fritz2023supports}}]\label{thm:point_sha}
		Let $\cC$ be a causal gs-monoidal category with supports and point liftings.
		Then we have a functor $\Upsilon \colon \cC \to \cat{Rel}$ defined on objects as  
		\begin{align}\label{eq:point_sha_obj}
			\Upsilon(A) \coloneqq \cC_\det (I,A)  
		\end{align}
		and on morphisms as $f \mapsto \Upsilon(f)$ where $\Upsilon(f)$ is a relation in $\cat{Rel}(\Upsilon(A), \Upsilon(X))$ defined via
		\begin{equation}\label{eq:point_sha_morph}
			\Upsilon(f) \, (a) \coloneqq \Set*[\big]{ x \in \cC_\det(I,X)  \given  x \ll f \circ a }.
		\end{equation}
		Furthermore, $\Upsilon$ is a strong gs-monoidal functor\footnotemark{} with respect to the monoidal structure of $\cat{Rel}$ as in \cref{ex:Rel}.
		\footnotetext{A strong gs-monoidal functor $\cC \to \cat{Rel}$ is similar to a lax gs-monoidal functor (\cref{def:gs-monoidal_functor}), but moreover the structure morphisms $\psi_0$ and $\psi_{A,B}$ are isomorphisms.
			See \cite{fritz2022lax} for more details on gs-monoidal functors.}%
	\end{theorem}
	While the theorem has been proven in \cite{fritz2023supports} for the case of Markov categories, it works just as well for gs-monoidal categories.
	We discuss $\Upsilon$ further in the next section under the name of pointed shadow functor.
	
\section{Pointed Shadow Functor}\label{sec:pointed_shadow}
	
	We think of the shadow functor, which is needed to specify a behavior structure from \cref{def:beh_structure}, as a concrete relational model of the ambient category.
	At first sight, it is unclear how strong is the assumption that such a shadow functor exists.
	Here, we show how to construct a specific shadow functor, called the \emph{pointed shadow functor} purely in terms of the data provided by the ambient category itself, as long as it satisfies a few additional properties.
	One should keep in mind that a given ambient category may have multiple plausible shadow functors.
	A pointed shadow functor, however, is the one we use most often because it is consistent with interpreting objects in the ambient category as particular sets of ``points'', e.g.\ like $T$ in $\cat{Tur}$ is the set of Turing machines.
	
	This interpretation, thinking of the objects $T$ and $C$ as \emph{sets} of targets and contexts respectively, also allows us to express the universality of a simulator $s$ in terms of equality of behaviors (of $s$ and the trivial simulator) for every target $t \in T$ and every context $c \in C$ as in relation \eqref{eq:univ_concrete} from \nameref{sec:motivation}.
	
	\begin{figure}[t]\centering
		\includegraphics[width=.9\columnwidth]{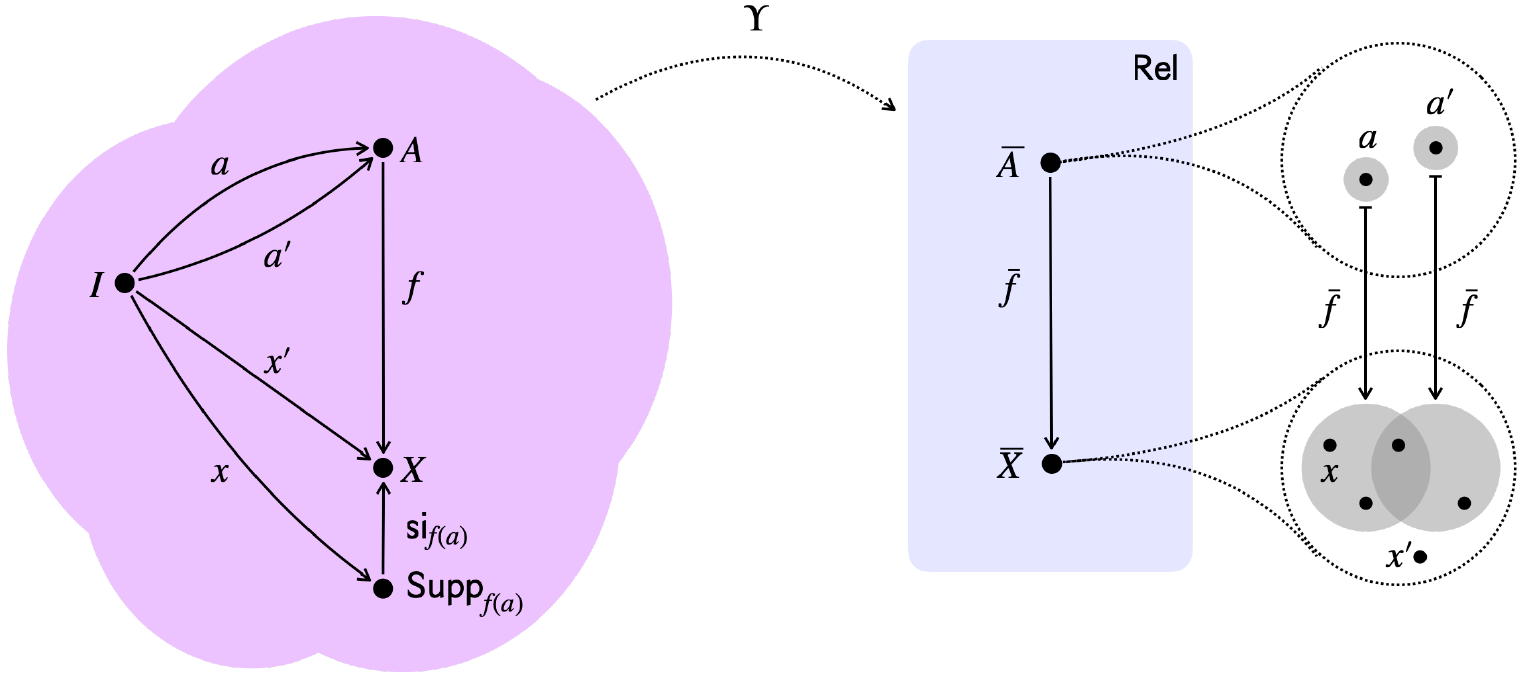}
		\caption{The pointed shadow functor takes an abstract process $f$ to a relation $\behim{f}$.
		When applied to a deterministic point $a$, this relation returns the set $\behim{f}(a)$ of all deterministic points $x$ that are within the support of $f(a)$ in the ambient category.}
		\label{fig:pointed_shadow}
	\end{figure}
		
	\parlabel{Construction of pointed shadows}
	Consider an ambient category $\cC$ that is causal, has supports and point liftings (a detailed presentation of these properties can be found in \cref{sec:additional_def}).
	This means, in particular, that for every morphism $f \colon A \to X$ we can refer to a deterministic subobject $\suppinc{f} \colon \Supp{f} \to X$ which is both necessary and sufficient for equality of processes in the future of $f$, i.e.\ we have equivalence \eqref{eq:supp_asfaithful}.
	In $\cat{Rel}$, the support object $\Supp{f}$ is the union of all the images of points in $A$ under relation $f$:
	\begin{equation}
		\Supp{f} = \Set*[\big]{ x \in X  \given  \exists \, a \in A \text{ such that } x \in f(a) }
	\end{equation}
	This set is canonically embedded within $X$ as a subset thereof.
	We can use supports in a general ambient category $\cC$ to give a canonical description of its morphisms as relations, and this allows us to define the pointed shadow functor $\Upsilon$.
		
	For any object $A$ of $\cC$, let the set $\behim{A}$ consist of the deterministic states of $A$, i.e.\ it is the hom-set $\cC_\det(I,A)$.
	Furthermore, for any morphism $f \in \cC(A,X)$, the relation $\behim{f} \in \cat{Rel}(\behim{A},\behim{X})$ is given by supports as subsets of $\behim{X}$:\footnotemark{}
	\footnotetext{While \cref{eq:point_sha_morph_2} superficially differs from \cref{eq:point_sha_morph} used to define $\Upsilon(f)$ in the statement of \cref{thm:point_sha}, they are in fact equivalent under the assumptions of the theorem \cite{fritz2023supports}.}%
	\begin{equation}\label{eq:point_sha_morph_2}
		\behim{f}(a) \coloneqq \cC_\det \left( I, \Supp{f(a)} \right) .
	\end{equation}
	One can view $\behim{f}(a)$ as a subset of $\behim{X}$ by composing its elements with the associated support inclusion ${\suppinc{f(a)} \colon \Supp{f(a)} \to X}$. 
	The mapping 
	\begin{align}
		\Upsilon(A) &\coloneqq \behim{A}  &  \Upsilon(f) &\coloneqq \behim{f}
	\end{align}
	then gives a functor $\Upsilon \colon \cC \to \cat{Rel}$, which we term the \textbf{pointed shadow functor} (\cref{fig:pointed_shadow}).
	Moreover, $\Upsilon$ can be provided with a strong gs-monoidal structure (\cref{thm:point_sha}), so that it can be used to define a behavior structure on the ambient category (\cref{def:beh_structure}).
	
	There are two main sources of examples of pointed shadow functors in this article.
	One of them is provided by forgetful functors as in the case of spin systems.
	The other is given by inclusion functors for those subcategories of $\cat{Rel}$, which are valid target--context categories in their own right.
	\begin{proposition}[Inclusion as a pointed shadow]\label{prop:inclusion_is_pointed}
		Let $\cC$ be a (quasi-total) gs-monoidal category that has strong supports and is a subcategory $\iota \colon \cC \hookrightarrow \cat{Rel}$ of the category of relations.
		Assume that for all objects $X$ and morphisms $f$ of $\cC$ we have
		\begin{align}
			\label{eq:incl_det_states} \cC_\det(I,X) &= \cat{Rel}_\det \bigl( I,\iota(X) \bigr),  \\  
			\label{eq:incl_supp} \iota \left( \Supp{f} \right) &= \Supp{\iota(f)}.
		\end{align}
		Then the inclusion functor $\iota \colon \cC \hookrightarrow \cat{Rel}$ coincides with the pointed shadow functor $\Upsilon$.
	\end{proposition}
	\begin{proof}
		This is an immediate consequence of the definition of the pointed shadow functor.
	\end{proof}
	
	\begin{example}[Pointed shadow for functions]\label{ex:function_pointed}
		For the category $\cat{Set}$ of sets and functions, the inclusion ${\iota \colon \cat{Set} \hookrightarrow \cat{Rel}}$ is its pointed shadow.
		Indeed, all morphisms in $\cat{Set}$ are deterministic, and deterministic morphisms of $\cat{Rel}$ are precisely functions, so that we have
		\begin{equation}
			\cat{Set}_\det = \cat{Set} = \cat{Rel}_\det,
		\end{equation}
		which implies \cref{eq:incl_det_states}.
		Furthermore, the support of a function (both in $\cat{Set}$ as well as in $\cat{Rel}$) is its image:
		\begin{equation}
			\Supp{f} = \im(f) = \Supp{\iota(f)}
		\end{equation}
		so that \cref{eq:incl_supp} follows.
	\end{example}
	
	\parlabel{Ambient imitation relation for intrinsic behaviors and pointed shadow}
	Given a target--context category with intrinsic behaviors whose shadow functor is the pointed shadow functor, we can think of the behavioral relation as a relation on the hom-set $\cC_\det(I,B)$, which in this case coincides with $\behim{B}$.
	The resulting ambient imitation relation $f \mrelop g$ then holds ifand only if for all $a \in \behim{A} \coloneqq \cC_\det(I,A)$, we have that 
	\begin{enumerate}
		\item for every deterministic state $v \in \cC_\det(I,B)$ within the support of the composite
			\begin{equation}\label{eq:behavior_of_g}
				\begin{tikzcd}
					I & A & {T \otimes C} & B
					\arrow["a", from=1-1, to=1-2]
					\arrow["g", from=1-2, to=1-3]
					\arrow["\eval", from=1-3, to=1-4]
				\end{tikzcd}
			\end{equation}
			there exists a deterministic state $u_v \in \cC_\det(I,B)$ within the support of
			\begin{equation}\label{eq:behavior_of_f}
				% https://q.uiver.app/?q=WzAsNCxbMCwwLCJJIl0sWzEsMCwiQSJdLFsyLDAsIlQgXFxvdGltZXMgQyJdLFszLDAsIkIiXSxbMCwxLCJhIl0sWzEsMiwiZiJdLFsyLDMsIlxcZXZhbCJdXQ==
				\begin{tikzcd}
					I & A & {T \otimes C} & B
					\arrow["a", from=1-1, to=1-2]
					\arrow["f", from=1-2, to=1-3]
					\arrow["\eval", from=1-3, to=1-4]
				\end{tikzcd}
			\end{equation}
			satisfying $u_v \brelop v$; and conversely
			
		\item for every $u \in \cC_\det(I,B)$ within the support of the state from \eqref{eq:behavior_of_f}, either 
			\begin{itemize}
				\item there exists a $v_u \in \cC_\det(I,B)$ within the support of the state from \eqref{eq:behavior_of_g}, such that $u \brelop v_u$ holds, or
				\item there are no deterministic states in the support of the state from \eqref{eq:behavior_of_g}.
			\end{itemize}
	\end{enumerate}
	These two conditions of course come from the enhancement and degradation requirements respectively, as expressed in \cref{def:brel}.

\section{Intrinsification}\label{app:intrinsification}

	In this section, starting from an arbitrary target--context category with behaviors, $\cC$, we construct a target--context category with intrinsic behaviors, $\cC_{\intr}$, which is defined in terms of the same ambient category, same $T$ and $C$, and the same shadow functor.
	We call $\cC_{\intr}$ the \textbf{intrinsification of} $\bm{\cC}$. 
	The construction of $\cC_{\intr}$ is such that whenever $f$ behaviourally subsumes $g$ within $\cC_{\intr}$ then $f$ also behaviourally subsumes $g$ within $\cC$. 
	In particular, universality in $\cC_{\intr}$ implies universality \mbox{in $\cC$}.
	\begin{theorem}[Intrinsification]
		Let $(\cC,T,C,\Beh,\behim{B},\behim{\eval},\brel)$ be a target--context category with behaviors. Then there is a target--context category with intrinsic behaviors given by $(\cC,T,C,\Beh,B_{\intr},\eval_{\intr},\brel')$ such that for any $f,g \colon A \to T \otimes C$, we have
		\begin{equation}\label{eq:intrinsification_weakening}
			f \mrelop' g \implies f \mrelop g,
		\end{equation}
		where $\mrel$ and $\mrel'$ are the associated ambient imitation relations (\cref{def:mrel}).
	\end{theorem}
	Note that implication \eqref{eq:intrinsification_weakening} equivalently means that the identity is a target--context functor from $\cC_{\intr}$ to $\cC$. 
	\begin{proof}
		We define 
		\begin{align}
			B_{\intr} &\coloneqq T \otimes C,  & \eval_{\intr} &\coloneqq \id_{T \otimes C},
		\end{align}
		and for $b_1, b_2 \in \behim{B_{\intr}}$ also
		\begin{equation}\label{eq:newrel}
			b_1 \brelop' b_2 \quad \coloniff \quad \behim{\eval} \circ b_1 \imrelop \behim{\eval} \circ b_2,
		\end{equation}
		where by abuse of notion we denote by $b_1$ and $b_2$ both the elements of $\behim{B_{\intr}}$ as well as the corresponding deterministic states. 
		
		Let us now prove that $f \mrelop' g$ implies $f \mrelop g$.
		First, we show that if $\behim{\eval}\circ \behim{g}(a)$ is defined\footnotemark{} for an $a \in \behim{A}$ (which implies that $\behim{g}(a)$ is defined), then so is $\behim{\eval} \circ \behim{f}(a)$.
		\footnotetext{When we say that a relation of type $I \to X$ is \emph{defined}, we mean that it is non-empty as a subset of $I \times X \cong X$, or in other words that it is total as a morphism in $\cat{Rel}$.}%
		The assumption $f \mrelop' g$ is equivalent to $\behim{f} \imrelop' \behim{g}$ since $\eval_{\intr}$ is just the identity on $T \otimes C$.
		Hence, by \cref{def:brel} of the imitation relation $\imrelop'$, $\behim{f}(a)$ is defined whenever $\behim{g}(a)$ is, and there exists a $\brel'$-enhancement 
		\begin{equation}\label{eq:old enh type}
			\mathrm{enh}_a \colon \behim{g}(a) \to \behim{f}(a).
		\end{equation}
		This means that for any $y \in \behim{g}(a)$ we have $\mathrm{enh}_a(y) \brelop' y$, which is equivalent to
		\begin{equation}\label{eq:enhy}
			 \behim{\eval} \circ \mathrm{enh}_a(y) \imrelop \behim{\eval} (y)
		\end{equation}
		by the definition of $\brelop'$ in \eqref{eq:newrel}.
		Since $\behim{\eval}\circ \behim{g}(a)$ is defined by assumption, there exists at least one $y \in \behim{g}(a)$ such that $\behim{\eval}(y)$ is defined.
		But then, by relation \eqref{eq:enhy}, $\behim{\eval} \circ \mathrm{enh}_a(y)$ is also defined. 
		Since $\mathrm{enh}_a(y)$ is an element of $\behim{f}(a)$, this in particular implies that that $\behim{\eval}\circ \behim{f}(a)$ is defined. 
	
		Next, we prove that $f \mrelop g$ holds.
		We consider an arbitrary $a \in \behim{A}$ such that $\behim{\eval} \circ \behim{g}(a)$ is defined.
		Then, the task is to construct a $\brel$-enhancement 
		\begin{equation}
			\mathrm{Enh}_a \colon \behim{\eval}\circ \behim{g}(a) \to \behim{\eval}\circ \behim{f}(a)
		\end{equation}
		as well as a $\brel$-degradation 
		\begin{equation}
			\mathrm{Deg}_a \colon \behim{\eval}\circ \behim{f}(a) \to \behim{\eval}\circ \behim{g}(a).
		\end{equation}
		We only show the construction of $\mathrm{Enh}_a$ explicitly; the one of $\mathrm{Deg}_a$ is analogous.
		
		By assumption, we have $f \mrelop' g$ and as before this implies that there exists a $\brel'$-enhancement as in \eqref{eq:old enh type}.
		That is, relation \eqref{eq:enhy} holds for every $y \in \behim{g}(a)$, which implies that there exists a $\brel$-enhancement 
		\begin{equation}\label{eq:double enh}
			\mathrm{enh}_{a,y} \colon \behim{\eval}(y) \to \behim{\eval} \circ \mathrm{enh}_a(y).
		\end{equation}
		Let $b$ be an arbitrary element of the domain $\behim{\eval}\circ \behim{g}(a)$ of the desired enhancement $\mathrm{Enh}_a$.
		We can choose a $y \in \behim{g}(a)$ such that $b$ is within its image under $\behim{\eval}$.
		Using it, we then define 
		\begin{equation}
			\mathrm{Enh}_a(b) \coloneqq \mathrm{enh}_{a,y}(b).
		\end{equation}
		This indeed defines a $\brel$-enhancement of correct type. 
	\end{proof}

\section{Polynomially Bounded Relations}\label{sec:relpoly}

	In this section, we construct the ambient category $\cat{Rel}_{\mathrm{poly}}$ used in \cref{sec:spinmodel} to describe universality of spin models.
	An objects of $\cat{Rel}_{\mathrm{poly}}$ is a set $S$ equipped with a size function $\abs{\ph}_S \colon S \to \R_{\geq 0}$.
	A generic morphism $f \colon S \to R$ is given by a relation $f \in \cat{Rel}( S, R)$ that is polynomially bounded, i.e.\ such that there exists a non-decreasing polynomial $\mathfrak{p}_f\in \R_{\geq 0}[x]$ satisfying
	\begin{equation}\label{eq:poly constr}
		r \in f(s) \quad \implies \quad \abs{r}_R \leq \mathfrak{p}_f(\abs{s}_S)
	\end{equation} 
	for all $r\in R$ and all $s\in S$.
	It is straightforward to check that this condition is preserved by the composition of relations.
	Moreover, choosing $\mathfrak{p}_{\id}$ to be the identity function shows that $\cat{Rel}_{\mathrm{poly}}$ contains all identities and is thus a category. 
	
	Next we equip $\cat{Rel}_{\mathrm{poly}}$ with the relevant structure to turn it into a gs-monoidal category.
	This structure is inherited from $\cat{Rel}$ as in \cref{ex:Rel}, we just need to show that all relevant morphisms involved in the construction satisfy implication \ref{eq:poly constr}.  
	Specifically, we define the tensor product of objects to be
	\begin{equation}
		\left( S, \abs{\ph}_S \right) \otimes \left(R, \abs{\ph}_R \right) \coloneqq \left(S \times R, \abs{\ph}_{S\otimes R} \right),
	\end{equation}
	where $\times$ is the cartesian product of sets and the size function is
	\begin{equation}
		\abs{\ph}_{S \otimes R}(s,r) \coloneqq \max \left\{ \abs{s}_S, \abs{r}_R \right\} .
	\end{equation}
	The tensor product of morphisms $f \colon S \to M$ and $ g \colon R \to N$ is defined by 
	\begin{equation}
		(m,n) \in (f \otimes g)(s,r) \quad \iff \quad  m \in f(s)  \text{ and } n \in g(r).
	\end{equation} 
	Note that this defines a morphism in $\cat{Rel}_{\mathrm{poly}}$, since we can derive
	\begin{equation}
		\abs{(m,n)}_{M \otimes N} \leq (\mathfrak{p}_f + \mathfrak{p}_g) \left( \abs{(s,r)}_{S \otimes R} \right)
	\end{equation}
	via the definition of $\abs{\ph}_{S \otimes R}$ and the fact that $\mathfrak{p}_f$ and $\mathfrak{p}_g$ are non-decreasing.
	The unit of $\otimes$ is given by $I \coloneqq \{\sngltn \}$ with $\abs{\sngltn}_I = 0$.
	Finally, we define $\cop_S(s)=(s,s)$ and $\discard_S(s)=\{\sngltn\}$ as in $\cat{Rel}$. 
	It is straightforward to see that these morphisms satisfy implication \eqref{eq:poly constr}.
	Moreover, $\cat{Rel}_{\mathrm{poly}}$ is quasi-total for the very same reason as $\cat{Rel}$.
 
\section{Monoidal Computers}
\label{sec:monoidal_computer}

	A related framework to the one we present here is that of monoidal computers \cite{pavlovic2013monoidal,pavlovic2018monoidal,pavlovic2023programs}, which provides a categorical model of computation and in this sense also shares many motivations with our approach.
	Here, we discuss this relation more precisely. 
	In particular, we show what additional data, on top of a target--context category with intrinsic behaviors, is sufficient to obtain a monoidal computer.
	
	The ambient category $\cC$ in works on monoidal computers is assumed to be gs-monoidal, just like in our case.
	A \textbf{monoidal computer}, as characterized by \cite[Proposition 3.2]{pavlovic2018monoidal}, consists of 
	\begin{enumerate}
		\item a fixed object $\prog$,
		\item a morphism $\{\} \colon \prog \otimes W \to B$ called program evaluator for all objects $W$ and $B$, and
		\item a deterministic morphism $[] \colon \prog \otimes X \to \prog$ called partial evaluator for all objects $X$,
	\end{enumerate}
	which satisfy compatibility conditions
	\begin{equation}\label{eq:pavlovic_compatibility}
		\input{pavlovic_compatibility.tikz}
	\end{equation}
	and the following `surjectivity' conditions:
	\begin{equation}\label{eq:pavlovic_universality}
		\input{pavlovic_universality.tikz}
	\end{equation}
	where the search for $r_f \colon A \to \prog$ is restricted to the deterministic subcategory $\cC_\det$ of $\cC$.
	Equations \eqref{eq:pavlovic_universality} are the programmability conditions that justify calling the program evaluator $\{\}$ a \emph{universal} evaluator.
	
	This notion does not exactly coincide with the concept of universality introduced in \cref{sec:Reductions}.
	Instead, it corresponds to reachability in the sense of \cref{sec:wps}, since conditions \eqref{eq:pavlovic_compatibility} and \eqref{eq:pavlovic_universality} are demanding that $\{\} \colon \prog \otimes W \to B$ is an $A$-complete parametrization of maps $W \to B$ (\cref{def:wps_intrinsic})
	
	Properties \eqref{eq:pavlovic_compatibility} and \eqref{eq:pavlovic_universality} allow one to express any morphism $f \colon A \otimes W \to B$ as an application of the universal evaluator to some program $r_f(a)$ in $\prog$ and to reduce computations with $n$ inputs to computations with $n-1$ inputs via the partial evaluator.
	While there is no explicit partial evaluator in our approach, we can compare a universal evaluator with the evaluation of a given simulator $s$ in the form of $\eval \circ s \in \cC(\prog \otimes C,  B)$. 
	We formalize this observation in \cref{prop:road_to_monoidal_computer} below.
	
	One of the consequences of the definition of monoidal computers is that every object $A \in \cC$ has an associated idempotent $p_A$
	\begin{equation}\label{eq:objects_from_idempotents}
		\input{objects_from_idempotents.tikz}
	\end{equation}
	that splits through $A$ and moreover $\sigma_A$ is deterministic.
	\begin{proposition}\label{prop:road_to_monoidal_computer}
		Consider a target--context category with intrinsic behaviors $(\cC,T,C,\eval,\brel)$ such that $C$ is identical to $B$ and denoted by $\prog$. 
		For every $A \in \cC$, assume that 
		\begin{itemize}
			\item the ambient relation $\mrel$ on $\cC(A,T \otimes \prog)$ is trivial, i.e.\ it coincides with equality,
			\item $\eval \colon T \otimes \prog \to \prog$ is an $A$-complete parametrization of maps $\prog \to \prog$, and
			
			\item there is a split idempotent $p_A \colon \prog \to \prog$ as in \cref{eq:objects_from_idempotents}.
		\end{itemize}
		If there is a universal simulator $s$ of type $\prog \otimes \prog \to T \otimes \prog$, then $\cC$ is a monoidal computer. 
	\end{proposition}
	\begin{proof}
		We need to construct program evaluators that satisfy conditions \eqref{eq:pavlovic_universality}.
		We define the basic one of type $\prog \otimes \prog \to \prog$ from the universal simulator and the evaluation morphism.
		Other program evaluators can be obtained from this one by split idempotents.
		
		That is, we define
		\begin{equation}\label{eq:program_evaluator_1}
			\input{program_evaluator_1.tikz}
		\end{equation}
		By \cref{prop:weak pt surj}, this program evaluator is an $A$-complete parametrization of maps $\prog \to \prog$ for any $A$, so that it satisfies conditions \eqref{eq:pavlovic_universality}.
		We define other partial evaluators as
		\begin{equation}\label{eq:program_evaluator_2}
			\input{program_evaluator_2.tikz}
		\end{equation}
		For any $f \colon A \otimes W \to B$, define 
		\begin{equation}\label{eq:program_evaluator_test}
			\input{program_evaluator_test.tikz}
		\end{equation}
		Then the deterministic morphism $r_f \colon A \to \prog$ for which we have \cref{eq:pavlovic_universality} is given by $r_{\tilde f}$, which exists by the fact that the program evaluator defined by \cref{eq:program_evaluator_1} is an $A$-complete parametrization of maps $\prog \to \prog$.
	\end{proof}
	Partial evaluators can also be constructed from the data provided by the assumptions of \cref{prop:road_to_monoidal_computer}, as one can use the reachability of program evaluators to define them so that the compatibility conditions are satisfied.
	This is essentially because in condition \eqref{eq:pavlovic_universality}, we can choose $f$ to be $\{  \} \colon (\prog \otimes A) \otimes W \to B$ and the partial evaluator $[] \colon \prog \otimes A \to \prog$ is then the corresponding morphism $r_{\{\}}$ coming from \cref{def:wps_intrinsic}.
	
	\begin{remark}
		While \cref{prop:road_to_monoidal_computer} is technically correct, it is not clear whether its direct application is of particular interest.
		In particular, note that one of its assumptions, that the ambient relation $\mrel$ is equality, can only be satisfied if all morphisms in $\cC$ are total, i.e.\ if it is a Markov category.
		Additionally, the assumption of intrinsic behaviors is not satisfied for the target--context category $\cat{Tur}$ in which singleton universal simulators coincide with universal Turing machines.
		
		Nevertheless, we believe that the result is of interest as it can inspire strategies to build monoidal computers even in case its assumptions are not in fact satisfied.
		Let us illustrate this with $\cat{Tur}$ from \cref{ex:TM,ex:TM_behaviors}.
		There contexts $C$ are given by strings in $\Sigma^\star$, but behaviors do not.
		Instead, $\behim{B}$ is the set $\Sigma^\star \times \pwrset(\Sigma^\star)$.
		Let us identify $\prog$ as the set of finite strings, so that $\prog = C$ holds.
		Then, instead of \cref{eq:program_evaluator_1}, we use
		\begin{equation}\label{eq:program_evaluator_3}
			\input{program_evaluator_3.tikz}
		\end{equation}
		to define the program evaluator, where $\eval^\intr$ is the evaluation morphism in the target--context category $\cat{Tur}^\intr$ (see \cref{ex:TM_behaviors}), which can be expressed as
		\begin{equation}
			\input{eval_intr.tikz}
		\end{equation}
		in terms of the evaluation function $\behim{\eval}$ in $\cat{Tur}$.
	\end{remark}

%\clearpage
\bibliographystyle{abbrvnat}
\bibliography{all-my-bibliography,references_tomas}{}
	
\end{document}